\documentclass[a4paper,12pt]{amsart}
\usepackage[round]{natbib}
\usepackage{url} % Break URL and DOI in bibliography across lines

\usepackage[utf8]{inputenc}
\usepackage[all]{xy}
\usepackage{amssymb}
\usepackage{amsmath}
\usepackage{mathtools}
\usepackage{upgreek} %Alternate greek letters
\usepackage{prodint} % Product integral
\usepackage{csquotes}
\usepackage{subfig}
\usepackage{booktabs} % For tables
\usepackage[shortlabels]{enumitem}
\usepackage{graphicx}
\graphicspath{{Figures/}} % Location of the graphics files
\usepackage{amsaddr}
\usepackage{color}
\usepackage{appendix}
\usepackage{mwe}
\usepackage{subfig}
\usepackage{tikz} % DAG
\usetikzlibrary{arrows.meta,shapes} % DAG
\usepackage{pgf,tikz}
\usetikzlibrary{arrows,shapes.arrows,shapes.geometric,shapes.multipart,decorations.pathmorphing,positioning,shapes.swigs,}

\usepackage{multirow}
\usepackage{soul}
\usepackage{xcolor}
\usepackage[left=1in, right=1in, top = 1.3in, bottom = 1.3in]{geometry}
\usepackage{empheq}

\usepackage{setspace}
\theoremstyle{definition} % Set theorem environments to roman
\newcommand\independent{\protect\mathpalette{\protect\independenT}{\perp}}
\def\independenT#1#2{\mathrel{\rlap{$#1#2$}\mkern2mu{#1#2}}}
\newtheorem{lemma}{Lemma}

\newtheorem{theorem}{Theorem}
\newtheorem{proposition}{Proposition}
\newtheorem{definition}{Definition}
\newtheorem{assumption}{Assumption}

\newtheorem*{example*}{Example}

\newtheorem{dgm}{Data generating mechanism}

\DeclareMathOperator*{\expit}{expit}

\usepackage{afterpage}

\title[]{Quantification of vaccine waning as a challenge effect}

\author{Matias Janvin$^1$, Mats J. Stensrud$^2$}
\address{$^1$Oslo Centre for Biostatistics and Epidemiology, University of Oslo, Norway}
\address{$^2$Institute of Mathematics, École Polytechnique Fédérale de Lausanne, Switzerland}

\date{\today}

\MakeOuterQuote{"}\EnableQuotes
\DeclareUnicodeCharacter{00A0}{ }

\DeclareMathOperator{\var}{var}
\DeclareMathOperator{\cov}{cov}

\allowdisplaybreaks % Allow page breaks in align environment

\begin{document}

% \clearpage
\begin{abstract}
    Knowing whether vaccine protection wanes over time is important for health policy and drug development. However, quantifying waning effects is difficult. A simple contrast of vaccine efficacy at two different times compares different populations of individuals: those who were uninfected at the first time versus those who remain uninfected until the second time. Thus, the contrast of vaccine efficacy at early and late times can not be interpreted as a causal effect.   
    We propose to quantify vaccine waning using the challenge effect, which is a contrast of outcomes under controlled exposures to the infectious agent following vaccination. 
    We identify sharp bounds on the challenge effect under non-parametric assumptions that are broadly applicable in vaccine trials using routinely collected data. We demonstrate that the challenge effect can differ substantially from the conventional vaccine efficacy due to depletion of susceptible individuals from the risk set over time. 
    Finally, we apply the methods to derive bounds on the waning of the BNT162b2 COVID-19 vaccine using data from a placebo-controlled randomized trial. 
    Our estimates of the challenge effect suggest waning protection after 2 months beyond administration of the second vaccine dose. 
\end{abstract}
\maketitle

\clearpage

\section{Introduction}

There are two prevailing approaches for quantifying vaccine waning. One approach is to use immunological assays to measure antibody levels after vaccination, and then use these measurements as a surrogate variable to infer the degree of vaccine protection \citep{levin_waning_2021}. However,  measurement of antibodies can fail to detect immunity, e.g.\ due to resident memory cells, and therefore can be insufficient to fully characterize immunity from prior vaccination or infection \citep{bergwerk_covid-19_2021,khoury_neutralizing_2021,rubin_audio_2022}.

A second approach is to contrast interval-specific cumulative incidences (CI) of infectious outcomes across vaccine and placebo recipients in randomized controlled trials (RCTs) by computing the vaccine efficacy (VE), often defined as one minus the ratio of cumulative incidences during a given time interval. Waning is quantified from direct observations of infection related events, rather than immunological surrogate markers. In this context, it is conventional to define vaccine waning as the decline over time of the VE \citep{halloran_design_1999, halloran_study_1997,halloran_design_2012,follmann_assessing_2020,follmann_deferred-vaccination_2021,follmann_estimation_2022,fintzi_assessing_2021,lin_evaluating_2021,tsiatis_estimating_2022}.

Estimands that quantify how the accrual of infectious outcomes changes over time are important when deciding booster vaccination regimes, see e.g.\ \citet{goldberg_waning_2021}. Similarly, empirical evidence of vaccine waning is also important to decide when to schedule seasonal vaccines; for example, when influenza vaccine protection wanes, these vaccines should be administered close to the time of the influenza wave \citep{ray_intraseason_2019}. Making such decisions based on conventional VE estimands is problematic because the VE at two different times compares different populations of individuals: those who were uninfected at the first time versus those who remain uninfected until the second time. Thus, the VE could decline over time only due to a depletion of susceptible individuals \citep{lipsitch_depletion--susceptibles_2019,ray_depletion--susceptibles_2020,halloran_design_2012,kanaan_estimation_2002,hudgens_endpoints_2004}.  

As stated by \citet{halloran_study_1997}, ``[an] open challenge is that of distinguishing among the possible causes of time-varying [VE] estimates.'' \citet{smith_assessment_1984} described two models of stochastic individual risk illustrating distinct mechanisms by which VE can decline over time, later known as the ``leaky'' versus ``all-or-nothing'' models \citep{halloran_interpretation_1992}, which were generalized to the  ``selection model'' and ``deterioration model'', respectively \citep{kanaan_estimation_2002}. These models parameterize individual risk of infection by introducing an unmeasured variable encoding the state of an individual's vaccine response, but rely on strong parametric assumptions that the investigators may be unwilling to adopt.

In this work, we propose to formally define waning in a causal (potential outcomes) framework as a ``challenge effect.'' This effect is defined with respect to interventions on both vaccination and exposure to the infectious agent, which in principle can be realized in a future experiment.  An interventionist definition of vaccine waning is desirable because it is closely aligned with health policy decisions \citep{robins_interventionist_2021,thomas_s_richardson_single_2013} and establishes a language for articulating testable claims about vaccine waning.  Furthermore, the challenge effect can guide development of new vaccines, say, to achieve a longer durability of protection. 

The challenge effect can, in principle, be identified by executing a challenge trial where the exposure to the infectious agent is controlled by the trialists. However, conducting such challenge trials is often unethical and infeasible \citep{hausman_challenge_2021}, in particular in vulnerable subgroups for which we may be most interested in quantifying vaccine protection. Thus, one of our main contributions is to describe assumptions that partially identify the challenge effect under commonly arising data structures, such as conventional randomized placebo controlled vaccine trials, where individuals are exposed to the infectious agent through their community interactions. The identification results do not require us to measure community exposure status, which is often difficult to ascertain and therefore often not recorded in trial data.

\subsection{Motivating example: COVID-19 vaccines\label{sec:motivating example}}
The safety and efficacy of the vaccine BNT162b2  against COVID-19 was tested in an RCT that assigned 22,085 individuals to receive the vaccine and 22,080 individuals to receive placebo. The trial recorded infection times and adverse reactions after vaccination. An overall vaccine efficacy of $0.913 \ (95\% \text{ CI: } 0.890 - 0.932)$ was reported, computed as one minus the incidence rate ratio of laboratory confirmed COVID-19 infection at 6 months of follow-up in individuals with no previous history of COVID-19. However, the interval-specific vaccine efficacy was as high as $0.917 \ (0.796 , 0.974)$ in the time period starting 11 days after receipt of first dose up to receipt of second dose, and later fell to $0.837 \ (0.747 , 0.899)$ after 4 months past the receipt of the second dose. One possible explanation for the difference in estimates between these two time periods is that the vaccine protection decreased (waned) over time. However, the difference might also be explained by a depletion of individuals who were susceptible to infection during time interval 1; more susceptible individuals were depleted in the placebo group compared to the vaccine group, which could have reduced the hazard of infection in the placebo group during interval 2 and thereby led to a smaller VE at later times. This observation prompts a question that we address in this work: does the protection of BNT162b2 wane over time, and if so, by how much?

\section{Observed data structure\label{sec:observed data}}
Consider a study where individuals are randomly assigned to treatment arm $A\in\{0,1\}$, such that $A=1$ denotes vaccine and $A=0$ denotes placebo. Suppose that individuals are followed up over two time intervals $k\in\{1,2\}$, where the endpoint of interval 1 coincides with the beginning of interval 2. 
In Appendix~E, we consider extensions to $K \geq 2$ time intervals and losses to follow-up. 
Let $Y_k\in\{0,1\}$ indicate whether the outcome of interest has occurred by the end of interval $k$, e.g.\ COVID-19 infection confirmed by nucleic acid amplification test in the COVID-19 example. Then, $\Delta Y_k=Y_{k}-Y_{k-1}$ is an indicator that the outcome occurred during interval $k$, and we define $Y_0=0$. Finally, let $L$ denote a vector of baseline covariates. We assume that the data are generated under the Finest Fully Randomized Causally Interpretable Structured Tree Graph (FFRCISTG) model \citep{thomas_s_richardson_single_2013,robins_alternative_2011,robins_new_1986}, which generalizes the perhaps more famous Non-Parametric Structural Equation Model with Independent Errors  (NPSEM-IE).\footnote{Based on the FFRCISTG  model, we let causal DAGs encode single world independencies between the counterfactual variables. In particular, the FFRCISTG model includes the NPSEM-IE as a strict submodel \citep{thomas_s_richardson_single_2013,robins_new_1986,pearl_causality_2009}. Because all estimands and identification assumptions in this manuscript are single world, it would also be sufficient, but not necessary, to assume that data are generated from an NPSEM-IE model \citep{pearl_causality_2009}.}  Similar to most vaccine trials \citep{tsiatis_estimating_2022,halloran_estimability_1996}, we will assume that there is no interference between individuals, because they are drawn from a larger study population and therefore infectious contacts between the trial participants are negligible. 
In Appendix~A, we clarify that this challenge effect is also practically relevant for a (plausible) target  population \textit{with interference}. 
A causal directed acyclic graph (DAG) illustrating the observed data structure is presented in Figure~\ref{fig: observed data}, and a dictionary of notation is given in Table~\ref{tab:notation}.
\begin{table}[htbp]
  \centering
  \caption{Summary of notation. Interventions for counterfactual quantities are denoted by superscripts.}
    \resizebox{\linewidth}{!}{
    \begin{tabular}{lp{11cm}}
         \hline Symbol & Definition \\\hline
    $A$ & Vaccine ($A=0$) versus control ($A=1$)  \\
    $Y_k$ & Indicator that the infectious outcome has occurred  by the end of time interval $k$, $Y_{k}\in\{0,1\}$. We define $Y_0=0$  \\
    $\Delta Y_k=Y_k-Y_{k-1}$ & Indicator that the infectious outcome has occurred during time interval $k$, $\Delta Y_{k}\in\{0,1\}$\\
    $E_k$ & Indicator of exposure to the infectious agent during time interval $k$, $E_k\in\{0,1\}$ \\
    $L$ & Vector of baseline covariates \\
    $U_{EY}$ & Unmeasured common cause of $E_k$ and $\Delta Y_{k^\prime}$ for some $k,k^\prime\in\{1,2\}$ \\
    $U_{E}$ & Unmeasured common cause of $E_1$ and $E_2$  \\
    $U_{Y}$ & Unmeasured common cause of $\Delta Y_1$ and $\Delta Y_2$  \\
    $E_k^a,\Delta Y_k^a$ & Exposure status and outcome indicator during time interval $k$ under assignment to vaccination level $A=a$  \\
    $E_1^{a,e_1=1},\Delta Y_1^{a,e_1=1}$ & Exposure status and outcome indicator under joint assignment to vaccination level $A=a$ and challenge with infectious inoculum during interval 1 ($e_1=1$). $E_k^{a,e_1=1},\Delta Y_k^{a,e_1=1}\in\{0,1\}$ \\
    $E_2^{a,e_1=0,e_2=1},\Delta Y_2^{a,e_1=0,e_2=1}$ & Exposure status and outcome indicator under joint assignment to vaccination level $A=a$, isolation from the infectious agent during interval 1 ($e_1=0$) and challenge with infectious inoculum during interval 2 ($e_2=1$). $E_2^{a,e_1=0,e_2=1},\Delta Y_2^{a,e_1=0,e_2=1}\in\{0,1\}$ \\
    $\text{VE}_1^\mathrm{obs}(l)=1-\frac{E[\Delta Y_1\mid A=1,L=l]}{E[\Delta Y_1\mid A=0,L=l]}$ & Observed (conventional) vaccine efficacy during interval 1 for baseline covariate level $L=l$. $\text{VE}_1^\mathrm{obs}(l)\in (-\infty,1]$ \\
    $\text{VE}_2^\mathrm{obs}(l)=1-\frac{E[\Delta Y_2\mid \Delta Y_1=0, A=1,L=l]}{E[\Delta Y_2\mid \Delta Y_1=0, A=0,L=l]}$ & Observed (conventional) vaccine efficacy during interval 2 for baseline covariate level $L=l$. $\text{VE}_2^\mathrm{obs}(l)\in (-\infty,1]$ \\
    $\text{VE}_1^{\mathrm{challenge}}(l)=1-\frac{E[\Delta Y_1^{a=1,e_1=1}\mid L=l]}{E[\Delta Y_1^{a=0,e_1=1}\mid L=l]}$ & Challenge effect during interval 1 for baseline covariate level $L=l$. $\text{VE}_1^\mathrm{challenge}(l)\in (-\infty,1]$ \\
    $\text{VE}_2^{\mathrm{challenge}}(l)=1-\frac{E[\Delta Y_2^{a=1,e_1=0,e_2=1}\mid L=l]}{E[\Delta Y_2^{a=0,e_1=0,e_2=1}\mid L=l]}$ & Challenge effect during interval 2 for baseline covariate level $L=l$. $\text{VE}_2^\mathrm{challenge}(l)\in (-\infty,1]$ \\
    $\psi(l)=\frac{E[\Delta Y_1^{a=1,e_1=1}\mid L=l]}{E[\Delta Y_2^{a=1,e_1=0,e_2=1}\mid L=l]}$ & Relative challenge effect for interval 1 versus interval 2 in baseline covariate level $L=l$. $\psi(l)\in [0,\infty)$ \\
    $\mathcal{L}_2(l),\mathcal{U}_2(l)$ & Sharp lower and upper bound of $\text{VE}_2^{\mathrm{challenge}}(l)$. $\mathcal{L}_2(l),\mathcal{U}_2(l)\in(-\infty,1]$ \\
    $\mathcal{L}_\psi(l),\mathcal{U}_\psi(l)$ & Sharp lower and upper bound of $\psi(l)$. $\mathcal{L}_\psi(l),\mathcal{U}_\psi(l)\in [0,\infty)$ 
    \end{tabular}%
     }
  \label{tab:notation}%
\end{table}%

\begin{figure}
    \centering
 \resizebox{0.4\columnwidth}{!}{
    \begin{tikzpicture}
\tikzset{line width=1.5pt, outer sep=0pt,
ell/.style={draw,fill=white, inner sep=2pt,
line width=1.5pt},
swig vsplit={gap=5pt,
inner line width right=0.5pt}};
\node[name=A,ell,  shape=ellipse] at (3,0) {$A$};
\node[name=E1,ell,  shape=ellipse] at (6,0) {$E_1$};
\node[name=E2,ell,  shape=ellipse] at (9,0) {$E_2$};
\node[name=Y1,ell,  shape=ellipse] at (6,2) {$\Delta Y_1$};
\node[name=Y2,ell,  shape=ellipse] at (9,2) {$\Delta Y_2$};
\node[name=L,ell,  shape=ellipse] at (3,2) {$L$};
\node[name=UY,ell,  shape=ellipse] at (7.5,4) {$U_{Y}$};
\node[name=UE,ell,  shape=ellipse] at (7.5,-2) {$U_{E}$};
% \node[name=UAY,ell,  shape=ellipse] at (0,1) {$U_{AY}$};
% \node[name=UEY1,ell,  shape=ellipse] at (12,2) {$U_{EY,1}$};
% \node[name=UEY2,ell,  shape=ellipse] at (12,0) {$U_{EY}$};
% \node[name=UY,ell,  shape=ellipse] at (7.5,4) {$U_{Y}$};
% \node[name=UE,ell,  shape=ellipse] at (7.5,-2) {$U_{E}$};
\begin{scope}[transparency group, opacity=0.3] 
    % \path [->] (A) edge (E1);
    \path [->,>={Stealth[black]}] (A) edge (Y1);
    \path [->,>={Stealth[black]}] (A) edge (Y2);
    % \path [->,>={Stealth[black]}] (A) edge[bend right] (E2);
    %\path [->,>={Stealth[black]}] (E1) edge (E2);
    \path[->,>={Stealth[black]}]  (E1) edge  (Y1);
    \path[->,>={Stealth[black]}]  (E2) edge  (Y2);
    \path [->,>={Stealth[black]}] (Y1) edge (Y2);
    \path [->,>={Stealth[black]}] (Y1) edge (E2);
    \path [->,>={Stealth[black]}] (L) edge (E1);
    \path [->,>={Stealth[black]}] (L) edge (E2);
    \path [->,>={Stealth[black]}] (L) edge (Y1);
    \path [->,>={Stealth[black]}] (L) edge[bend left] (Y2);
    \path[->,>={Stealth[black]}]  (UY) edge[black] (Y1);
    \path[->,>={Stealth[black]}]  (UY) edge[black] (Y2);
    \path[->,>={Stealth[black]}]  (UE) edge[black] (E1);
    \path[->,>={Stealth[black]}]  (UE) edge[black] (E2);
\end{scope}
\begin{scope}[>={Stealth[black]},
              every edge/.style={draw=black,very thick}]
    
    % \path[->,>={Stealth[red]}]  (UAY) edge[red] (A);
    % \path[->,>={Stealth[red]}]  (UAY) edge[red] (Y1);
    % \path[->,>={Stealth[red]}]  (UEY1) edge[red] (E1.140);
    % \path[->,>={Stealth[red]}]  (UEY1) edge[red, bend right] (Y1);
    % \path[->,>={Stealth[red]}]  (UEY2) edge[red, bend left] (E2.215);
    % \path[->,>={Stealth[red]}]  (UEY2) edge[red] (Y2);
    % \path[->,>={Stealth[red]}]  (E1) edge[red] (Y2);
    % \path[->,>={Stealth[red]}]  (E1) edge[red] (E2);
    % \path[->,>={Stealth[red]}]  (A) edge[red] (E1);
    % \path[->,>={Stealth[red]}]  (A) edge[red, bend right] (E2);

\end{scope}
\end{tikzpicture}
 }
\caption{Causal DAG illustrating a data generating mechanism for the observed variables}
    \label{fig: observed data}
\end{figure}
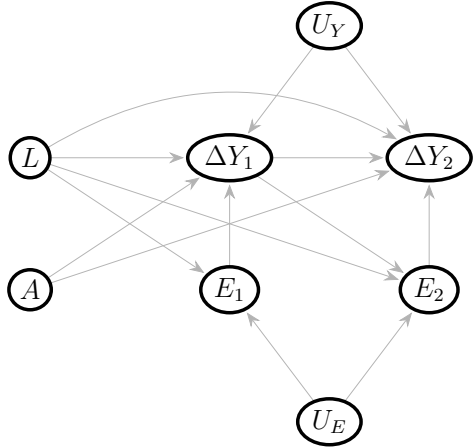

\section{Questions and estimands of interest\label{sec: questions and estimand}}
Let $\Delta Y_1^{a,e_1=1}$ be a counterfactual indicator of the outcome $\Delta Y_1$, had individuals been given treatment $A=a$ at baseline and subsequently, in time interval 1,  been exposed to an infectious inoculum through a controlled procedure ($e_1 =1$). Furthermore, let $\Delta Y_2^{a,e_1=0,e_2=1}$ be the counterfactual outcome under an intervention that assigns treatment $A=a$, then isolates the individual from the infectious agent during time interval 1 ($e_1 =0$) and finally exposes the individual to an infectious inoculum in the same controlled manner at the beginning of time interval 2 ($e_2 =1$). 

We define the conditional challenge effect during time intervals 1 and 2, respectively, by
\begin{align}
    &\text{VE}_1^{\mathrm{challenge}}(l)=1-\frac{E[\Delta Y_1^{a=1,e_1=1}\mid L=l]}{E[\Delta Y_1^{a=0,e_1=1}\mid L=l]} ~,\notag\\ &\text{VE}_2^{\mathrm{challenge}}(l)=1-\frac{E[\Delta Y_2^{a=1,e_1=0,e_2=1}\mid L=l]}{E[\Delta Y_2^{a=0,e_1=0,e_2=1}\mid L=l]} ~. \label{eq: challenge VE}
\end{align}
The challenge effect quantifies the mechanism by which the vaccine exerts protective effects, outside of pathways that involve changes in exposure pattern, by targeting hypothetical challenge trials where an infectious challenge is administered after an isolation period (versus no isolation) in vaccinated individuals. The practical relevance of the challenge effect is, e.g.,\ illustrated by the concrete proposal of \citet{ray_depletion--susceptibles_2020}, who suggested to study waning of influenza vaccines by enrolling participants to receive a vaccine during a random week from August to November, and then contrasting the incidence of influenza infection between early and late vaccinees. \citet{monge_imprinting_2023} and \citet{hernan_selection_2023} proposed a related hypothetical challenge trial to describe selection bias in quantification of immune imprinting of COVID-19 vaccines. However, while these challenge trials are rarely conducted, to our knowledge, previous work has not considered identification and estimation of such estimands from conventional vaccine trials. In Sections~\ref{sec:identification}-\ref{sec:maintext estimation}, we clarify how to identify and estimate the challenge effect using routinely collected data from conventional vaccine trials.

We denote the conventional (observed) vaccine efficacy estimands by 
\begin{align}
    &\textrm{VE}_1^{\mathrm{obs}}(l)=1-\frac{E[\Delta Y_1\mid A=1,L=l]}{E[\Delta Y_1\mid A=0,L=l]} ~, \notag\\
&\textrm{VE}_2^{\mathrm{obs}}(l)=1-\frac{E[\Delta Y_2\mid \Delta Y_1=0, A=1,L=l]}{E[\Delta Y_2\mid \Delta Y_1=0, A=0,L=l]} ~.\label{eq: observed VE}
\end{align}
To reduce clutter, we will write $\textrm{VE}_k^{\mathrm{challenge}}$ and $\text{VE}_k^\mathrm{obs}$ for the challenge effect and observed vaccine efficacy at time $k$, omitting the argument $l$. However, in general, both quantities could vary with $l$. 
For a given controlled exposure to the infectious agent, the challenge effect $\text{VE}_k^\mathrm{challenge}$ does not change with infection prevalence. 
In contrast, $\text{VE}_k^{\mathrm{obs}}$ can depend on the prevalence of infection in the communities of the trial participants \citep{struchiner_randomization_2007}.

We take the position that ``waning'' refers to a contrast of counterfactual outcomes  under different interventions, as formalized in the following definition.

\begin{definition}[Challenge waning\label{def:waning}]
    \begin{align*}
        \text{VE}_1^{\mathrm{challenge}} > \text{VE}_2^{\mathrm{challenge}}~.
    \end{align*}
    We say that the vaccine effect wanes from interval 1 to interval 2 if the challenge effect decreases from interval 1 to interval 2. 
\end{definition}
Indeed, $\text{VE}_1^{\mathrm{challenge}} \neq \text{VE}_2^{\mathrm{challenge}}$ does not imply, nor is it implied by a change in a conventional vaccine efficacy measure, $\text{VE}_1^{\mathrm{obs}} \neq \text{VE}_2^{\mathrm{obs}}$.  Thus, in the COVID-19 example it is not sufficient to know that $\text{VE}^{\mathrm{obs}}_k$ decreased over time in order to ascertain that the vaccine protection has waned. We illustrate this point by simulating data generating mechanisms with different values of $\text{VE}_k^\mathrm{obs}$ and $\text{VE}_k^\mathrm{challenge}$ in the Supplementary Material (Appendix~I).

So far we have introduced a hypothetical exposure intervention without characterizing in detail the properties of such an intervention. In the next section, we describe a list of properties that the exposure intervention should satisfy, give examples of real life challenge trials that plausibly meet these conditions, and present examples where the conditions fail.

\subsection{Exposure interventions\label{sec:interventions on exposure}}
Returning to the COVID-19 example, we will now outline a hypothetical trial for the joint intervention $(a,e_1,e_2)$. Suppose that assignment to vaccine versus placebo is blinded, and that the intervention $e_k=0$ for $k \in \{1,2\}$ denotes perfect isolation from the infectious agent, for example by confining individuals under $e_k=0$ such that they are not in contact with the wider community. Suppose further that $e_k=1$ denotes intranasal challenge by pipette at the beginning of time interval $k$ with a dose of virus particles that is representative of a typical infectious exposure in the observed data; the controlled procedure could, e.g., be similar to the COVID-19 challenge experiment described by \citet{killingley_safety_2022}. Furthermore, let $E_k=1$ be an (unmeasured) indicator that an individual in the observed data is exposed to a quantity of virus particles that exceeds a threshold believed to be necessary to develop COVID-19 infection.

\begin{assumption}[Consistency\label{ass:consistency 2 intervals}]
We assume that interventions on treatment $A$ and exposures $E_1,E_2$ are well-defined such that the following consistency conditions hold for all $a,e_1,e_2\in\{0,1\}$:
    \begin{enumerate}[(i)]
        \item $\text{ if }    A=a \text{ then } E_1=E_1^a, \Delta Y_1=\Delta Y_1^a, E_2= E_2^a,  E_2^{e_1=0}= E_2^{a,e_1=0}, \Delta Y_2 = \Delta Y_2^a, $ $\Delta Y_2^{e_1=0} = \Delta Y_2^{a,e_1=0} ~,$
        \item $\text{ if }   A=a,E_1=e_1  \text{ then } \Delta Y_1=\Delta Y_1^{a,e_1}, E_2=E_2^{a,e_1}, \Delta Y_2=\Delta Y_2^{a,e_1} ~,$
        \item $\text{ if } A=a,E_2^{a,e_1=0}=e_2  \text{ then } \Delta Y_2^{a,e_1=0}=\Delta Y_2^{a,e_1=0,e_2} ~.$
    \end{enumerate}
\end{assumption}
Assumption~\ref{ass:consistency 2 intervals} implicitly subsumes that the counterfactual outcomes of one individual do not depend on the treatment of another individual \citep{pearl_consistency_2010}, i.e.\ no interference. 
We discuss the assumption of no interference further in Appendix~A.
Consistency assumptions are routinely invoked when doing causal inference \citep{hernan_causal_nodate} and require that the target trial exposure produces the same outcomes as the exposures that occurred in the observed data. In other words, the intervention $e_1=1$ must be representative of the observed exposures for individuals with $E_1=1$, and similarly for time interval 2. However, Assumption~\ref{ass:consistency 2 intervals} does not specify exactly what this representative exposure is, in particular, what the dose of the viral inoculum is in the target challenge trial.

While routinely invoked, consistency assumptions, like Assumption~\ref{ass:consistency 2 intervals}, can be violated if multiple versions of exposure, which have different effects on future outcomes, are present in the data \citep{hernan_does_2016}. For the motivating exposure intervention, this could happen if there exist subgroups that have substantially different quantities of viral particles per exposure compared to the rest of the population, and if the risk of acquiring infection is highly sensitive to such differences in viral particles. The same ambiguity would occur if the number of exposures vary substantially per individual within each time interval.

Appendix~D discusses how Assumption~\ref{ass:consistency 2 intervals} can be weakened under multiple treatment versions, building on \citet{vanderweele_constructed_2022} and \citet{vanderweele_causal_2013}.
In particular, we show that an analogous identification argument holds when the number of viral particles in the pipette used to challenge individuals is a random variable sampled from a suitable distribution, or when this viral inoculum has a constant representative size that exists in a non-trivial class of settings. 
It is possible to test the strict null hypothesis that the vaccine does not wane under \textit{any} (observed) size of viral inoculum that satisfies a set of assumptions formalized in Appendix~D, assuming that the distribution of viral inocula amongst exposed individuals remains the same between intervals $k=1$ and $k=2$. 
This is closely related to the ``Similar Study Environment'' assumption adopted by  \citet{fintzi_assessing_2021}, who give several examples of changes in study environment that could lead to changing $\text{VE}^\mathrm{obs}$ over time; for example, changes in viral strands over time, or changes in mask wearing behavior that could lead to different quantities of viral particles per exposure at different times.

Violations of Assumption~\ref{ass:consistency 2 intervals} can be mitigated by adopting a blinded crossover trial design \citep{follmann_deferred-vaccination_2021}, where individuals are randomized to vaccine or placebo at baseline and subsequently receive the opposite treatment after a fixed interval of time. In such trials, one can minimize differences in background infection prevalence or in viral particles per exposure between recent versus early recipients of the active vaccine by contrasting the cumulative incidence of outcomes during the \textit{same interval of calendar time} \citep{lipsitch_depletion--susceptibles_2019,ray_depletion--susceptibles_2020}.

To establish a relation between exposures and outcomes, we introduce the following assumption.
\begin{assumption}[Exposure necessity\label{ass: exposure necessity 2 intervals}]
For all $a\in\{0,1\}$ and $k\in\{1,2\}$,
 \begin{align*}
     E^{a}_k = 0 \implies \Delta Y_k^{a} = 0  \text{ and } E_2^{a,e_1=0} = 0 \implies \Delta Y_2^{a,e_1=0} = 0 ~.
 \end{align*}
\end{assumption}
The exposure necessity assumption \citep{stensrud_identification_2023} states that any individual who develops the infection, must have been exposed. Standard infectious disease models typically express the infection rate as a product of a contact rate and a per exposure transmission probability, see, e.g., (2.14) in \citet{halloran_design_2012} or (2) in \citet{tsiatis_estimating_2022}. Such models not only imply that exposure is necessary for infection, but also impose strong parametric assumptions on the infection transmission mechanism, and it is not clear how these parametric assumptions can be empirically falsified. In contrast, exposure necessity can be falsified by observing whether any individuals develop the outcome without being exposed. 
In the COVID-19 example, exposure necessity is plausible because COVID-19 is primarily believed to spread through respiratory transmission \citep{meyerowitz_transmission_2021}, where viral particles come into contact with the respiratory mucosa.

Similarly to other works on vaccine effects, we require the exposure to be unaffected by the treatment assignment \citep{halloran_design_1999}.
\begin{assumption}[No treatment effect on exposure in the unexposed\label{ass:no effect on exposure 2 intervals}]
\begin{align*}
    E_1^{a=0}=E_1^{a=1} \text{ and } E_2^{a=0,e_1=0}=E_2^{a=1,e_1=0} ~.
\end{align*}
\end{assumption}
In blinded placebo controlled RCTs, such as the COVID-19 example introduced in Section~\ref{sec:motivating example}, patients do not know whether they have been assigned to vaccine or placebo shortly after treatment assignment. Therefore, their community interactions are unlikely to be affected by the treatment assignment, and we find it plausible that $E_1^{a=0}=E_1^{a=1}$ \citep{halloran_causal_1995,stensrud_identification_2023}. However, an individual who develops the outcome during time interval 1 may change their subsequent behavior during time interval 2. If more individuals develop the outcome under placebo compared to the active vaccine, then treatment could affect exposures during time interval 2 via infection status in time interval 1, as illustrated by a path $A\rightarrow \Delta Y_1 \rightarrow E_2$ (Figure~\ref{fig: observed data}). This could reflect a retention of highly exposed vaccine recipients in the risk set \citep{hudgens_endpoints_2004}. Under an intervention that eliminates exposure during time interval 1, there are no such selection effects during time interval 2. Thus, Assumption~\ref{ass:no effect on exposure 2 intervals} is plausible in our motivating target trial. 

Assumptions about balanced exposure between treatment arms are standard in vaccine research in order to interpret VE estimates as protective effects of treatment that are not due to changes in behavior \citep{hudgens_endpoints_2004}. For example, \citet{tsiatis_estimating_2022} used a related assumption, stating that the counterfactual contact rate $c_a^b(t)$ under a blinded assignment ($b$) to treatment $a$ is equal for $a=0$ and $a=1$ at all times and for all individuals.

To identify outcomes under an intervention that isolates individuals during interval 1, we introduce the following assumption.
\begin{assumption}[Exposure effect restriction\label{ass: exclusion 2 intervals}]
For all $a\in\{0,1\}$,
\begin{align}
    E[\Delta Y_2\mid A=a,L] \leq E[\Delta Y_2^{e_1=0}\mid A=a,L] \leq E[\Delta Y_1+\Delta Y_2\mid A=a,L] \text{ w.p. 1} ~. \label{eq: sufficient single-world inequality}
\end{align}
\end{assumption}
To give intuition for Assumption~\ref{ass: exclusion 2 intervals}, consider the following examples  where the expected counterfactual outcome under isolation reaches the upper or lower limits (Figure~\ref{fig:weak_exclusion_restriction}). Suppose that 3 out of 40 individuals in stratum $A=a,L=l$ developed the infectious outcome during interval 1. Subsequently, 5 individuals experienced the outcome during interval 2. In the worst case scenario, all 3 individuals who developed the outcome in interval 1 would also have the outcome during interval 2 if they were isolated during interval 1, and in the best case scenario none of the 3 individuals would have the outcome after isolation. If the outcomes in the remaining 37 individuals were identical under isolation versus no isolation, a total of 5/40 (best case) to 8/40 (worst case) individuals would experience the outcome during interval 2 after isolation. In the example, suppose further that all proportions represent expectations.
\begin{figure}
    \centering
    \includegraphics[width=\linewidth]{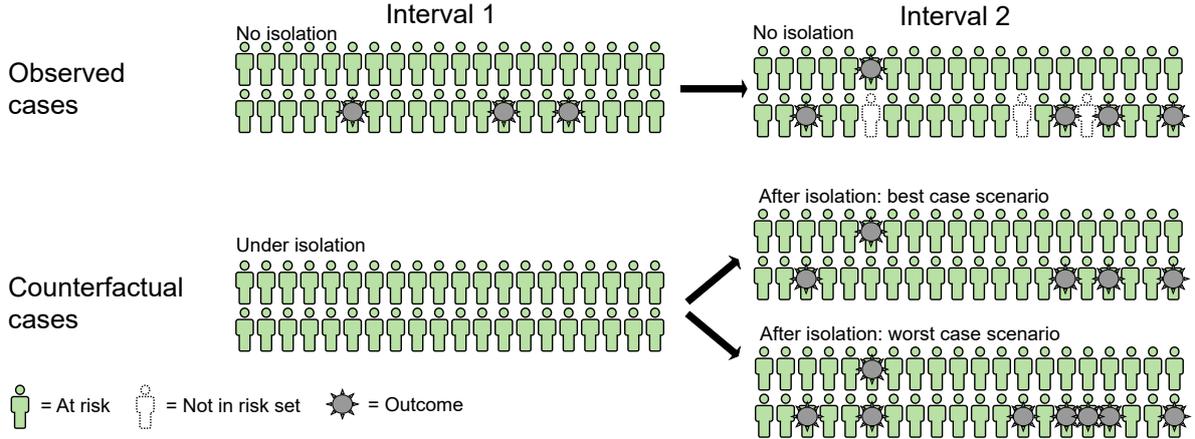}
    \caption{Illustration of  scenarios where the bounds in Assumption~\ref{ass: exclusion 2 intervals} are reached. Each panel is evaluated in a stratum $A=a,L=l$. Under isolation during time interval 1, there would be 3 more individuals at risk during time interval 2 compared to the observed data with no isolation. In the best case scenario, none of the 3 individuals would develop the outcome during interval 2 after isolation. In the worst case scenario, all 3 individuals would develop the outcome during interval 2 after isolation. 
    }
    \label{fig:weak_exclusion_restriction}
\end{figure}

Assumption~\ref{ass: exclusion 2 intervals} can be violated if there exist causal paths from $E_1$ to $\Delta Y_2$ that are not intersected by $\Delta Y_1$, e.g.\ the path $E_1\rightarrow \Delta Y_2$ in Figure~\ref{fig:SWIG_waning_2_intervals}(A). 
In Appendix~C, we show that Assumption~\ref{ass: exclusion 2 intervals} is implied by an exclusion restriction assumption under an intervention that prevents the outcome from occurring during time interval 1.
For example, Assumption~\ref{ass: exclusion 2 intervals} can fail if isolation during time interval 1 precludes exposure to the infectious pathogen that contributes to sustained natural immunity without causing the outcome during interval 1. In other words, Assumption~\ref{ass: exclusion 2 intervals} can fail if a substantial proportion of infections are not detected, e.g.\ when infections are asymptomatic. In such cases, individuals may become more susceptible to infectious exposures during time interval 2 if they are isolated during time interval 1. For example, natural immunity against severe malaria wanes over time in individuals who migrate from malaria endemic countries to non-endemic countries \citep{mischlinger_imported_2020}. Likewise, increases in RSV infections after COVID-19 lockdown may have been caused by prolonged periods without viral exposure, reducing naturally acquired immunity \citep{bardsley_epidemiology_2023}.  
It can be possible to detect whether asymptomatic infections occur in the trial population by observing whether any placebo recipients develop antibodies or other disease-specific immune markers. However, Assumption~\ref{ass: exclusion 2 intervals} is not necessarily violated even if some of the trial participants develop natural immunity from undetected infectious exposures: the lower limit of \eqref{eq: sufficient single-world inequality} is unlikely to be violated, because natural immunity makes individuals more protected against infection, and not less. The upper limit of \eqref{eq: sufficient single-world inequality} reflects a scenario with an extremely heterogeneous risk of infection (Appendix~C), or with prominent effects of natural immunity, and therefore this upper limit might hold even if some degree of natural immunity is present.

Assumption~\ref{ass: exclusion 2 intervals} can also be violated if an exposure to the infectious agent that does not result in COVID-19 infection during time interval 1 leads to a change in exposure behavior during time interval 2, corresponding to the path $E_1\rightarrow E_2 \rightarrow \Delta Y_2$ in Figure~\ref{fig:SWIG_waning_2_intervals}(A).

In the following example, we consider another exposure intervention.

\begin{example*}[Alternative exposure intervention\label{ex:alternative exposure}] 
   Let $E_k=1$  denote individuals who are free to interact in an environment where they can be exposed to the infectious agent. Conversely, define $E_k=0$  to mean that individuals are isolated, i.e.\ confined to an environment where they cannot be exposed to the infectious agent. Likewise, let $e_k=0$ and $e_k=1$ denote controlled procedures whereby individuals are isolated versus introduced into such an environment. For example, \citet{ray_depletion--susceptibles_2020} discuss a trial design where individuals are vaccinated for influenza before the seasonal outbreak, and are therefore initially isolated until the seasonal outbreak begins, disregarding infections out of season. Alternatively, in a trial contrasting early versus late vaccination of individuals before travelling to areas where the infectious agent is widespread, individuals are isolated before travel, and then exposed on arrival. Under this definition of infectious exposure, being part of a population of infective individuals is viewed as an infectious challenge in itself.   Thus, in the observed data described in Section~\ref{sec:observed data}, $E_k=1$ w.p.~1 under this alternative exposure definition. Then, Assumptions~\ref{ass: exposure necessity 2 intervals}-\ref{ass:no effect on exposure 2 intervals} and Assumption~\ref{ass: exposure exchageability 2 intevals} hold by design, although Assumption~\ref{ass: exclusion 2 intervals} can still fail. Furthermore, to interpret the challenge effect as a measure of vaccine protection, we still require the study environment, e.g.\ quantity of viral particles per exposure, to be constant across intervals 1 and 2 (Appendix~D).
\end{example*}

\section{Identification\label{sec:identification}}
\subsection{Identification assumptions}
The following additional assumption, which concerns common causes of exposure and the outcome, is useful for identification of the challenge effect.

\begin{assumption}[Exposure exchangeability\label{ass: exposure exchageability 2 intevals}\footnote{In the original published version of the manuscript \citep{Janvin02012025}, the independence relation \eqref{eq: homogeneous_exch} was included in Assumption~\ref{ass: exposure exchageability 2 intevals}. In this updated version,  \eqref{eq: homogeneous_exch} has been moved to Proposition~\ref{prp: homogeneity}, because it was only intended to be invoked in Proposition~\ref{prp: homogeneity} and is not needed in Theorem~\ref{thm: waning minimal ass 2 intevals}. Furthermore, because we only require unconfoundedness of $\Delta Y_2$ and $E_1,E_2$ under an intervention which sets $e_1=0$ in our proofs, we have weakened the conditional independence assumptions accordingly in Assumptions~\ref{ass: exposure exchageability 2 intevals}, \ref{ass: exposure exchageability multiple versions}, \ref{ass: exposure exchageability K intevals} and in \eqref{eq: homogeneous_exch}.}]\hfill
%,  \eqref{eq: homogeneous_exch} in Theorem~\ref{thm: waning minimal ass 2 intevals}.

For all $a,e_1,e_2\in\{0,1\}$,
\begin{align*}
    \Delta Y_1^{a,e_1} \independent E_1^a\mid A=a,L \text{ and } \Delta Y_2^{a,e_1=0,e_2} \independent E_2^{a,e_1=0}\mid A=a,L ~.
\end{align*}
\end{assumption}
Exposure exchangeability states that exposure  and the outcome  are unconfounded conditional on baseline covariates. \citet{tsiatis_estimating_2022} used an assumption closely related to Assumption~\ref{ass: exposure exchageability 2 intevals}; they assumed that $\{\pi_1(t,\tau),\pi_0(t,\tau)\}\independent \{S,c^b\}\mid X$, where $\pi_a(t,\tau)$ is the counterfactual individual-specific transmission probability per contact under treatment $A=a$, $S$ denotes a vaccination site, $X$ is a set of baseline covariates and $c^b$ is a contact rate. 

The causal graphs in Figure~\ref{fig:SWIG_waning_2_intervals}(A) and (B) illustrate two different data generating mechanisms that can violate Assumptions~\ref{ass:no effect on exposure 2 intervals}-\ref{ass: exposure exchageability 2 intevals} by paths $A\rightarrow E_k$, $E_1\rightarrow \Delta Y_2$, $E_1\rightarrow E_2$ and $E_2^{a,e_1}\leftarrow U_{EY} \rightarrow \Delta Y_2^{a,e_1,e_2}$. In contrast, the causal graph in Figure~\ref{fig: observed data} satisfies Assumptions~\ref{ass: exposure necessity 2 intervals}-\ref{ass: exposure exchageability 2 intevals}.

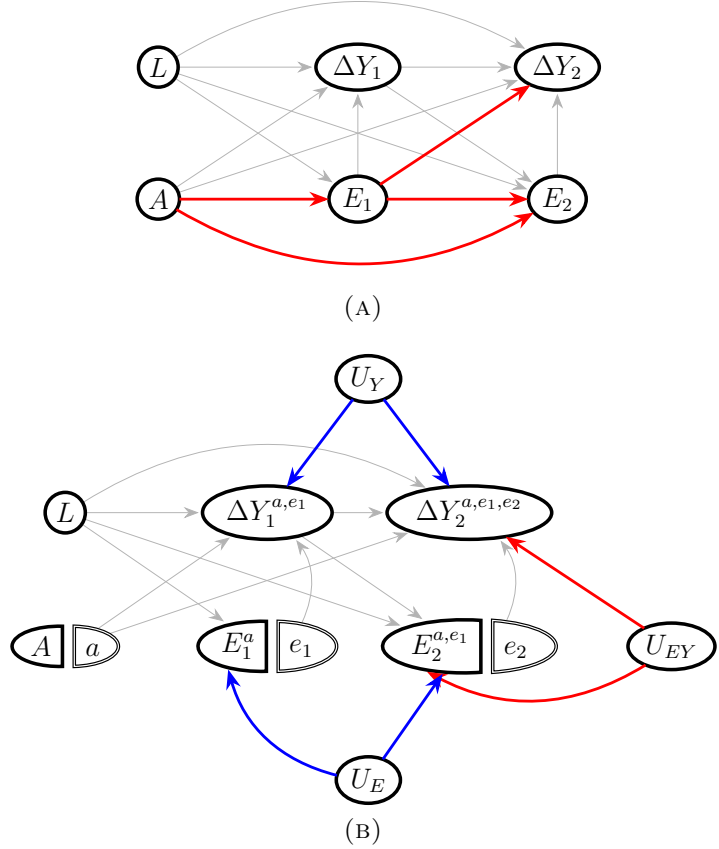
\begin{figure}
    \centering
    \subfloat[]{
 \resizebox{0.4\columnwidth}{!}{
    \begin{tikzpicture}
\tikzset{line width=1.5pt, outer sep=0pt,
ell/.style={draw,fill=white, inner sep=2pt,
line width=1.5pt},
swig vsplit={gap=5pt,
inner line width right=0.5pt}};
\node[name=A,ell,  shape=ellipse] at (3,0) {$A$};
\node[name=E1,ell,  shape=ellipse] at (6,0) {$E_1$};
\node[name=E2,ell,  shape=ellipse] at (9,0) {$E_2$};
\node[name=Y1,ell,  shape=ellipse] at (6,2) {$\Delta Y_1$};
\node[name=Y2,ell,  shape=ellipse] at (9,2) {$\Delta Y_2$};
\node[name=L,ell,  shape=ellipse] at (3,2) {$L$};
% \node[name=UAY,ell,  shape=ellipse] at (0,1) {$U_{AY}$};
% \node[name=UEY1,ell,  shape=ellipse] at (12,2) {$U_{EY,1}$};
% \node[name=UEY2,ell,  shape=ellipse] at (12,0) {$U_{EY}$};
% \node[name=UY,ell,  shape=ellipse] at (7.5,4) {$U_{Y}$};
% \node[name=UE,ell,  shape=ellipse] at (7.5,-2) {$U_{E}$};
\begin{scope}[transparency group, opacity=0.3] 
    % \path [->] (A) edge (E1);
    \path [->,>={Stealth[black]}] (A) edge (Y1);
    \path [->,>={Stealth[black]}] (A) edge (Y2);
    % \path [->,>={Stealth[black]}] (A) edge[bend right] (E2);
    %\path [->,>={Stealth[black]}] (E1) edge (E2);
    \path[->,>={Stealth[black]}]  (E1) edge  (Y1);
    \path[->,>={Stealth[black]}]  (E2) edge  (Y2);
    \path [->,>={Stealth[black]}] (Y1) edge (Y2);
    \path [->,>={Stealth[black]}] (Y1) edge (E2);
    \path [->,>={Stealth[black]}] (L) edge (E1);
    \path [->,>={Stealth[black]}] (L) edge (E2);
    \path [->,>={Stealth[black]}] (L) edge (Y1);
    \path [->,>={Stealth[black]}] (L) edge[bend left] (Y2);
\end{scope}
\begin{scope}[>={Stealth[black]},
              every edge/.style={draw=black,very thick}]
    
    % \path[->,>={Stealth[red]}]  (UAY) edge[red] (A);
    % \path[->,>={Stealth[red]}]  (UAY) edge[red] (Y1);
    % \path[->,>={Stealth[red]}]  (UEY1) edge[red] (E1.140);
    % \path[->,>={Stealth[red]}]  (UEY1) edge[red, bend right] (Y1);
    % \path[->,>={Stealth[red]}]  (UEY2) edge[red, bend left] (E2.215);
    % \path[->,>={Stealth[red]}]  (UEY2) edge[red] (Y2);
    \path[->,>={Stealth[red]}]  (E1) edge[red] (Y2);
    \path[->,>={Stealth[red]}]  (E1) edge[red] (E2);
    % \path[->,>={Stealth[blue]}]  (UY) edge[blue] (Y1);
    % \path[->,>={Stealth[blue]}]  (UY) edge[blue] (Y2);
    % \path[->,>={Stealth[blue]}]  (UE) edge[blue, bend left] (E1.215);
    % \path[->,>={Stealth[blue]}]  (UE) edge[blue] (E2);
    \path[->,>={Stealth[red]}]  (A) edge[red] (E1);
    \path[->,>={Stealth[red]}]  (A) edge[red, bend right] (E2);
\end{scope}
\end{tikzpicture}
 }}

\subfloat[]{
 \resizebox{0.6\columnwidth}{!}{
    \begin{tikzpicture}
\tikzset{line width=1.5pt, outer sep=0pt,
ell/.style={draw,fill=white, inner sep=2pt,
line width=1.5pt},
swig vsplit={gap=5pt,
inner line width right=0.5pt}};
\node[name=A,shape=swig vsplit] at (3,0){
\nodepart{left}{$A$}
\nodepart{right}{$a$} };
\node[name=E1,shape=swig vsplit] at (6,0){
\nodepart{left}{$E_1^a$}
\nodepart{right}{$e_1$} };
\node[name=E2,shape=swig vsplit] at (9,0){
\nodepart{left}{$E_2^{a,e_1}$}
\nodepart{right}{$e_2$} };
\node[name=Y1,ell,  shape=ellipse] at (6,2) {$\Delta Y_1^{a,e_1}$};
\node[name=Y2,ell,  shape=ellipse] at (9,2) {$\Delta Y_2^{a,e_1,e_2}$};
\node[name=L,ell,  shape=ellipse] at (3,2) {$L$};
% \node[name=UAY,ell,  shape=ellipse] at (0,1) {$U_{AY}$};
% \node[name=UEY1,ell,  shape=ellipse] at (12,2) {$U_{EY,1}$};
\node[name=UEY2,ell,  shape=ellipse] at (12,0) {$U_{EY}$};
\node[name=UY,ell,  shape=ellipse] at (7.5,4) {$U_{Y}$};
\node[name=UE,ell,  shape=ellipse] at (7.5,-2) {$U_{E}$};
\begin{scope}[transparency group, opacity=0.3] 
    % \path [->] (A) edge (E1);
    \path [->,>={Stealth[black]}] (A) edge (Y1);
    \path [->,>={Stealth[black]}] (A) edge (Y2);
    % \path [->,>={Stealth[black]}] (A) edge[bend right] (E2);
    %\path [->,>={Stealth[black]}] (E1) edge (E2);
    \path[->,>={Stealth[black]}]  (E1.40) edge[bend right]  (Y1);
    \path[->,>={Stealth[black]}]  (E2.40) edge[bend right]  (Y2);
    \path [->,>={Stealth[black]}] (Y1) edge (Y2);
    \path [->,>={Stealth[black]}] (Y1) edge (E2);
    \path [->,>={Stealth[black]}] (L) edge (E1);
    \path [->,>={Stealth[black]}] (L) edge (E2);
    \path [->,>={Stealth[black]}] (L) edge (Y1);
    \path [->,>={Stealth[black]}] (L) edge[bend left] (Y2);
\end{scope}
\begin{scope}[>={Stealth[black]},
              every edge/.style={draw=black,very thick}]
    
    % \path[->,>={Stealth[red]}]  (UAY) edge[red] (A);
    % \path[->,>={Stealth[red]}]  (UAY) edge[red] (Y1);
    % \path[->,>={Stealth[red]}]  (UEY1) edge[red] (E1.140);
    % \path[->,>={Stealth[red]}]  (UEY1) edge[red, bend right] (Y1);
    \path[->,>={Stealth[red]}]  (UEY2) edge[red, bend left] (E2.235);
    \path[->,>={Stealth[red]}]  (UEY2) edge[red] (Y2);
    % \path[->,>={Stealth[red]}]  (E1) edge[red] (Y2);
    % \path[->,>={Stealth[red]}]  (E1) edge[red] (E2);
    \path[->,>={Stealth[blue]}]  (UY) edge[blue] (Y1);
    \path[->,>={Stealth[blue]}]  (UY) edge[blue] (Y2);
    \path[->,>={Stealth[blue]}]  (UE) edge[blue, bend left] (E1.215);
    \path[->,>={Stealth[blue]}]  (UE) edge[blue] (E2);
    % \path[->,>={Stealth[red]}]  (A) edge[red] (E1);
\end{scope}
\end{tikzpicture}
 }
}

\caption{(A) The red paths $E_1\rightarrow \Delta Y_2$ and $E_1\rightarrow E_2$ can violate exposure effect restriction (Assumption~\ref{ass: exclusion 2 intervals}) as well as \eqref{eq: no UY}-\eqref{eq: no UE}, and the red paths $A\rightarrow E_k$ can violate no effect of treatment on exposure in the unexposed (Assumption~\ref{ass:no effect on exposure 2 intervals}). (B) Backdoor paths between exposure $E_k$ and outcome $\Delta Y_k$ that are not blocked by baseline covariates $L$, such as the red path $E_2^{a,e_1}\leftarrow U_{EY} \rightarrow \Delta Y_2^{a,e_1,e_2}$, can violate exposure exchangeability (Assumption~\ref{ass: exposure exchageability 2 intevals}). In the presence of the blue paths $\Delta Y_1^{a,e_1}\leftarrow U_Y \rightarrow \Delta Y_2^{a,e_1,e_2}$ and $E_1^a\leftarrow U_E \rightarrow E_2^{a,e_1}$, $\text{VE}_2^{\mathrm{challenge}}$ can differ from $\text{VE}_2^{\mathrm{obs}}$ due to depletion of susceptible individuals during time interval 1. 
}
    \label{fig:SWIG_waning_2_intervals}
\end{figure}

\subsection{Identification results}
In the following theorem, we give bounds on the challenge effect under the assumptions introduced so far.

\begin{theorem}\label{thm: waning minimal ass 2 intevals}
Suppose that Assumptions~\ref{ass:consistency 2 intervals}-\ref{ass: exposure exchageability 2 intevals} hold in a conventional vaccine trial (formalized in Appendix~B). 
Then, the challenge effect during time interval 1 is point identified, 
\begin{align}
    \text{VE}_1^{\mathrm{challenge}}(l)=\text{VE}_1^{\mathrm{obs}}(l)=1-\frac{E[Y_1\mid A=1,L=l]}{E[Y_1\mid A=0,L=l]} ~, \label{eq: VE_challenge_1 and VE_obs_1}
\end{align}
whenever $E[Y_1\mid A=a,L=l]>0$ for all $a\in\{0,1\}$, and the challenge effect during interval 2 is partially identified by sharp bounds $\mathcal{L}_2(l)\leq \text{VE}_2^{\mathrm{challenge}}(l)\leq \mathcal{U}_2(l)$, where
\begin{align}
    \mathcal{L}_2(l) &= 1-\frac{E[Y_2 \mid A=1,L=l]}{E[Y_2-Y_1\mid A=0,L=l]} ~, \label{eq:L_VE_challenge_2}\\
     \mathcal{U}_2(l) &= 1-\frac{E[Y_2-Y_1\mid A=1,L=l]}{E[Y_2\mid A=0,L=l]} ~, \label{eq:U_VE_challenge_2}
\end{align}
whenever $E[Y_2-Y_1\mid A=a,L=l]>0$ for all $a\in\{0,1\}$.
\end{theorem}
A proof of Theorem~\ref{thm: waning minimal ass 2 intevals} is given in Appendix~B.  
Additionally, in the Supplementary Material, we give \texttt{R} code illustrating a counterfactual data generating mechanism that attains the bounds $\mathcal{L}_2(l)$ and $ \mathcal{U}_2(l)$. 

The upper bound $ \mathcal{U}_2(l)$ is reached when $i$) all the individuals that were infected during time interval 1 in the placebo arm would have become infected if they were isolated during interval 1 and challenged with the exposure during interval 2,  $ii$) none of the individuals that were infected during time interval 1 in the vaccine arm would have become infected if they were isolated during interval 1 and challenged with the exposure during interval 2 and $iii$) every individual that was uninfected during interval 1 would have an unchanged outcome during interval 2 if they were isolated during interval 1. We can use a similar argument to find a scenario where the lower bound $\mathcal{L}_2(l)$ is reached.

It is straightforward to show that $\mathcal{L}_2(l)\leq \text{VE}_2^{\mathrm{obs}}(l)\leq  \mathcal{U}_2(l)$ by re-expressing the bounds in Theorem~\ref{thm: waning minimal ass 2 intevals} in terms of discrete hazard functions (Appendix~B). If few events occur during time interval 1, it follows from Theorem~\ref{thm: waning minimal ass 2 intevals} that the resulting bounds both approach $\text{VE}_2^{\mathrm{obs}}$. A heuristic observation along these lines was made by \citet{follmann_deferred-vaccination_2021,fintzi_assessing_2021}. Furthermore, Theorem~\ref{thm: waning minimal ass 2 intevals} clarifies plausible and testable assumptions for partial identification of (challenge) vaccine waning and provides sharp bounds under the identifying assumptions.

Because $e_k=1$ denotes challenge by a representative dose of the infectious agent (Assumption~\ref{ass:consistency 2 intervals}), a comparison across studies of challenge effects identified by Theorem~\ref{thm: waning minimal ass 2 intevals} typically compares infectious inocula of different sizes that e.g.\ depend on the prevalence of infection in the respective background populations.

The marginal challenge effect, involving quantities $1-E[\Delta Y_1^{a=1,e_1=1}]/E[\Delta Y_1^{a=0,e_1=1}]$ and  $1-E[\Delta Y_2^{a=1,e_1=0,e_2=1}]/E[\Delta Y_2^{a=0,e_1=0,e_2=1}]$, can be expressed as a weighted average over the conditional challenge effect in \eqref{eq: challenge VE} with weights that are unidentified when exposure status $E$ is unmeasured \citep{stensrud_identification_2023,huitfeldt_collapsibility_2019}. However, under the additional assumption that infectious exposure deterministically causes the outcome in placebo recipients, the marginal challenge effect is also identified in terms of the conditional challenge effect \citep{stensrud_identification_2023}.

\begin{proposition}\label{prp: homogeneity}
Suppose that Assumptions~\ref{ass:consistency 2 intervals}-\ref{ass: exposure exchageability 2 intevals} hold in a conventional vaccine trial (formalized in Appendix~B). 
Assume further that for all $a\in\{0,1\}$,
\begin{align}
    \Delta Y_2^a &\independent E_1^a \mid \Delta Y_1^a,E_2^a,A=a,L ~, \label{eq: no UY}\\
    E_2^a &\independent E_1^a\mid \Delta Y_1^a,A=a,L ~, \label{eq: no UE} \\
    \Delta Y_2^{a,e_1=0} &\independent E_1^a\mid A=a,L~, \label{eq: homogeneous_exch}
\end{align}
and that $P(E_1=0\mid A=a,L)>0$ and $E[\Delta Y_2\mid A=a,L]>0$ w.p. 1.
Then,
\begin{align}
    E[\Delta Y_2^{e_1=0}\mid A=a,L] = E[\Delta Y_2\mid \Delta Y_1=0,A=a,L] \text{ w.p. 1 }~,\label{eq: homogeneity in Y and E}
\end{align}
which implies that $\text{VE}_2^\mathrm{challenge}(l)=\text{VE}_2^\mathrm{obs}(l)$ for all $l$.
\end{proposition}
Equality~\eqref{eq: homogeneity in Y and E} implies Assumption~\ref{ass: exclusion 2 intervals} (shown in Appendix~B), and thus \eqref{eq: no UY}-\eqref{eq: no UE} imply Assumption~\ref{ass: exclusion 2 intervals}.

Expressions~\eqref{eq: no UY}-\eqref{eq: no UE} imply strong homogeneity assumptions, because they can be violated by the presence of $U_Y$ or $U_E$ in Figure~\ref{fig:SWIG_waning_2_intervals}(B), and are not necessary for the bounds in Theorem~\ref{thm: waning minimal ass 2 intevals} to be informative about vaccine waning. Furthermore,  under causal faithfulness, \eqref{eq: no UY}-\eqref{eq: no UE} are violated by paths $E_1\rightarrow E_2$ or $E_1\rightarrow \Delta Y_2$, but \eqref{eq: sufficient single-world inequality} is not necessarily violated by these paths, as discussed in Section~\ref{sec: questions and estimand}. In this sense, \eqref{eq: sufficient single-world inequality} may be more robust than \eqref{eq: no UY}-\eqref{eq: no UE} to the presence of natural immunity and changes in exposure behavior among trial participants.
Proposition~\ref{prp: homogeneity} formalizes sufficient conditions under which $\text{VE}_2^\mathrm{challenge}(l)$ is identified by $\text{VE}_2^\mathrm{obs}(l)$. A special case arises when Assumptions~\ref{ass:consistency 2 intervals}-\ref{ass: exposure exchageability 2 intevals}, Assumptions~S1-S2 (Appendix~B) and \eqref{eq: no UY}-\eqref{eq: no UE} hold \textit{without} conditioning on baseline covariates; then, $\text{VE}_k^\mathrm{challenge}=\text{VE}_k^\mathrm{obs}$ marginally for all $k$, and conventional vaccine efficacy estimates are equal to the marginal challenge effect. This implies an even stronger homogeneity condition than assuming \eqref{eq: no UY}-\eqref{eq: no UE} with baseline covariates, because it can also be violated by paths such as $\Delta Y_1 \leftarrow L \rightarrow \Delta Y_2$ in Figure~\ref{fig:SWIG_waning_2_intervals}(A).
Investigators who want to quantify vaccine waning should decide on a case-by-case basis whether to report estimates of $\mathcal{L}_2(l),\mathcal{U}_2(l), \text{VE}_2^\mathrm{obs}(l)$ or marginal estimands, and justify their assumptions accordingly, using subject matter knowledge. We do not rely on these homogeneity conditions in our analysis of the BNT162b2 COVID-19 vaccine trial in Section~\ref{sec:example Pfizer trial}.

Next, suppose that the effect of placebo does not wane in the sense that the risk of the outcomes under a challenge immediately after placebo administration is equal to the risk of outcomes under isolation during interval~1, and subsequent challenge during interval~2.
\begin{assumption}[No waning of placebo\label{ass: no waning placebo 2 intervals}]
\begin{align}
E[\Delta Y_1^{a=0,e_1=1}\mid L] =E[\Delta Y_2^{a=0,e_{1}=0,e_2=1}\mid L] \text{ w.p.~1} \label{eq: no waning placebo 2 intervals}~.
\end{align}
\end{assumption}

Assumption~\ref{ass: no waning placebo 2 intervals} states that a controlled exposure leads to the same outcomes in conditional expectation, whether or not the exposure is preceded by an isolation period. The assumption could be violated if isolation during interval 1 leads to a loss of natural immunity acquired before baseline, but this violation is unlikely in \citet{thomas_safety_2021} because around $95\%$ of participants had no prior history of COVID-19 infections.

\subsection{An alternative target trial}
Under no waning of placebo (Assumption~\ref{ass: no waning placebo 2 intervals}), it follows straightforwardly from Theorem~\ref{thm: waning minimal ass 2 intevals} that we can bound the ratio $\psi(l)=E[\Delta Y_1^{a=1,e_1=1}\mid L=l]/E[\Delta Y_2^{a=1,e_1=0,e_2=1}\mid L=l]$
 by
\begin{align}
    \mathcal{L}_\psi(l) \leq \psi(l) \leq \mathcal{U}_\psi(l) ~, \label{eq: bounds psi}
\end{align}
where
\begin{align}
    \mathcal{L}_\psi(l)&=\frac{1-\text{VE}_{1}^\mathrm{obs}(l)}{1-\mathcal{L}_2(l)}  ~, \label{eq: LB minimal ass 2 intervals}\\
    \mathcal{U}_\psi(l)&=\frac{1-\text{VE}_{1}^\mathrm{obs}(l)}{1-\mathcal{U}_2(l)} ~.   \label{eq: UB minimal ass 2 intervals}
\end{align}
The estimand $\psi(l)$ can also be expressed as $\psi(l)=(1-\text{VE}_1^\mathrm{challenge}(l))/(1-\text{VE}_2^\mathrm{challenge}(l))$ under Assumption~\ref{ass: no waning placebo 2 intervals}, and corresponds to the following target trial \citep{hernan_target_2022}: Let a group of individuals be randomized to one of two treatment groups on a calendar date $X$. On the same date, all individuals in both  groups are vaccinated.  In one  group, all individuals are isolated against infectious exposures until calendar date $Z$, and then they are challenged through a controlled procedure. In the second group, all individuals are vaccinated and directly challenged with the infectious agent through the same controlled procedure. The outcome of interest is the cumulative incidence of infection during a pre-specified duration of time after calendar date $Z$, for example 2 months. If there are more outcome events under a challenge that \textit{is preceded by} an isolation period after vaccination ($\Delta Y_2^{a=1,e_1=0,e_2=1}$) compared to an immediate challenge after vaccination ($\Delta Y_1^{a=1,e_1=1}$), that is, $\psi<1$, then the vaccine effect has waned. The target trial is similar to the estimand identified by the clinical experiment proposed by \citet{ray_depletion--susceptibles_2020}, which is a contrast of the observed incidence of influenza infection in recently vaccinated individuals versus individuals vaccinated further in the past. 

Finally, under the homogeneity conditions \eqref{eq: no UY}-\eqref{eq: no UE} in Proposition~\ref{prp: homogeneity}, $\psi(l)$ is identified by the naive contrast of cumulative incidences $\psi^\mathrm{obs}(l)=(1-\text{VE}_1^\mathrm{obs}(l))/(1-\text{VE}_2^\mathrm{obs}(l))$.

\section{Estimation\label{sec:maintext estimation}}
Suppose we have access to data for individuals $i\in\{1,\dots,n\}$ consisting of treatment $A_i$, baseline covariates $L_i$, and event times $T_i$ subject to losses to follow-up (censoring), indicated by $C_i\in\{0,1\}$. We assume that individuals are sampled into the study through a procedure such that the random vectors $(A_i,L_i,T_i,C_i)$ are i.i.d.\ \citep{cox_planning_1958}.  Then, for each treatment group $a\in\{0,1\}$, we estimate the conditional cumulative incidence functions, $E[Y_k\mid A=a,L=l]$, at the end of interval $k$ (time $t_k$) by $\widehat \mu_{k,a,l}= 1-\exp(-\widehat \Lambda_{0,a}(t_k))^{r(l;\widehat{\beta}_a )}$ using the Breslow estimator $\widehat \Lambda_{0,a}(t)$ and  estimated coefficients $\widehat\beta_{a}$ that maximize the partial likelihood with respect to the proportional hazards model $\lambda(t\mid A=a,L=l)=\lambda_{0,a}(t)r(l;\beta_a)$ for $t\in[0,t_2]$ \citep{cox_regression_1972}. Here, $\lambda_{0,a}(t)$ denotes the baseline hazard of the infectious outcome at time $t$ in treatment group $A=a$.  The estimator $\widehat\mu_{k,a,l}$ is a standard cumulative incidence estimator \citep{therneau_modeling_2000}, and can easily be implemented using standard statistical software. We give an example in \texttt{R} using the \texttt{survival} package \citep{terry_therneau_package_2021} in the Supplementary Material. In our data example, we assumed that the parametric part of the hazard model is given by $r(l;\beta_a)=\exp(\beta_a l)$. We chose to handle tied event times in the Cox model using the Efron approximation \citep{therneau_modeling_2000}. 
It is also possible to estimate the cumulative incidence function through other frequently used regression models, such as logistic regression, as we discuss in Appendix~F, or additive hazards models \citep{aalen_survival_2008}, to name a few. 

Finally, expressions \eqref{eq: observed VE}, \eqref{eq: VE_challenge_1 and VE_obs_1}-\eqref{eq:U_VE_challenge_2} and \eqref{eq: LB minimal ass 2 intervals}-\eqref{eq: UB minimal ass 2 intervals} motivate the plugin-estimators 
\begin{align*}
    \widehat{\text{VE}}_{1}^\mathrm{obs}(l) &= 1-\frac{\widehat{\mu}_{k=1,a=1,l}}{\widehat{\mu}_{k=1,a=0,l}}~, \\
    \widehat{\text{VE}}_{2}^\mathrm{obs}(l) &= 1-\frac{\widehat{\mu}_{k=2,a=1,l}-\widehat{\mu}_{k=1,a=1,l}}{\widehat{\mu}_{k=2,a=0,l}-\widehat{\mu}_{k=1,a=0,l}}\cdot\frac{1-\widehat{\mu}_{k=1,a=0,l}}{1-\widehat{\mu}_{k=1,a=1,l}}~,  \\
    \widehat{\mathcal{L}}_{2}(l)&=1-\frac{\widehat{\mu}_{k=2,a=1,l}}{\widehat{\mu}_{k=2,a=0,l}-\widehat{\mu}_{k=1,a=0,l}}~,  \\
    \widehat{\mathcal{U}}_{2}(l)&=1-\frac{\widehat{\mu}_{k=2,a=1,l}-\widehat{\mu}_{k=1,a=1,l}}{\widehat{\mu}_{k=2,a=0,l}}~,   \\
    \widehat{\mathcal{L}}_{\psi}(l)
    &=(1-\widehat{\text{VE}}_{1}^\mathrm{obs}(l))/(1- \widehat{\mathcal{L}}_{2}(l))~, \label{eq: L_psi plugin} \\
    \widehat{\mathcal{U}}_{\psi}(l)&=(1-\widehat{\text{VE}}_{1}^\mathrm{obs}(l))/(1- \widehat{\mathcal{U}}_{2}(l))  ~.
\end{align*}
Pointwise confidence intervals  can, e.g., be estimated with individual-level data using non-parametric bootstrap, which we illustrate in Appendix~H using a publicly available synthetic dataset resembling the RTS,S/AS01 malaria vaccine trial \citep{noauthor_phase_2012}, described by \citet{benkeser_estimating_2019}.

Suppose a decision-maker is interested in the lower bound $\mathcal{L}_2(l)$, because they are concerned about the worst-case scenario corresponding to the greatest extent of waning. Then, we propose to use a lower one-sided $95\%$ confidence interval for $\widehat{\mathcal{L}}_2(l)$, as $\mathcal{L}_2(l)$. Conversely, for a decision-maker who is interested in testing whether any waning is present, we propose to use a one-sided upper $95\%$ confidence interval for $\widehat{\mathcal{U}}_2(l)$.
Finally, decision-makers who seek to weigh the lower and upper bounds evenly  may prefer to use a joint confidence for the lower and upper bound \citep{horowitz_nonparametric_2000}.

In Appendix~G, we describe estimators of the bounds $\mathcal{L}_2$ and $\mathcal{U}_2$ using summary data for the number of recorded events and person time at risk, and characterize their asymptotic distribution using the delta method. 
Furthermore, in Appendix~F, we describe estimators of $\text{VE}_2^\mathrm{obs}$ and of bounds of $\text{VE}_2^\mathrm{challenge}$ that use logistic regression to estimate the cumulative incidences, and illustrate the approach with a simulated example in Appendix~I.

\section{Example: BNT162b2 against COVID-19\label{sec:example Pfizer trial}}
We analyzed data from a blinded, placebo controlled vaccine trial described by \citet{thomas_safety_2021}, where individuals were randomized to two doses of the mRNA vaccine BNT162b2 against COVID-19 ($A=1$) or placebo ($A=0$), 21 days apart. Participants were 12 years or older, and were enrolled during a period of time from July 27, 2020 to October 29, 2020 (older than 16) and from October 15, 2020 to January 12, 2021 (aged 12-15), in 152 sites in the United States (130 sites), Argentina (1 site), Brazil (2 sites), South Africa (4 sites), Germany (6 sites) and Turkey (9 sites). By January 12, 2021, there had been 21.94 million cases of COVID-19 in the United States \citep{mathieu_coronavirus_2020}, amounting to roughly 7\% of the US population \citep{US_census_population_nodate}. The study included systematic measures to test participants for infection with COVID-19, with 5 follow-up visits within the first 12 months, and an additional sixth follow-up visit after 24 months \citep[Protocol]{thomas_safety_2021}. Here, participants were questioned about respiratory symptoms. Additionally, they were instructed to report any respiratory symptoms via a telehealth visit after symptom onset.

Vaccine efficacy estimates ($\text{VE}^\mathrm{obs}$) were reported to decrease with time since vaccination \citep{thomas_safety_2021}. In principle, the decrease of $\text{VE}^\mathrm{obs}$ over time could be due to declining protection of the vaccine, or alternatively also due to a higher depletion of susceptible individuals in the placebo group compared to the vaccine group during time interval 1. To distinguish between these two explanations, we conducted inference on \eqref{eq: challenge VE} using publicly available summary data from \citet{thomas_safety_2021}, reported in Table~\ref{tab:Pfizer CI 3}. 
A detailed description of the estimators is given in Appendix~G.

We let $k=1$ denote the time interval from 11 days after dose 1 until 2 months after dose 2, and $k=2$ to denote the time interval from 2 months after dose 2 until 4 months after dose 2 (Figure~\ref{fig: time intervals}). Individuals received the second dose of the vaccine 10 days into interval $k=1$. The estimate $\widehat{\mathrm{VE}}^\mathrm{challenge}_1$ can be interpreted as a conservative estimate of the challenge effect if individuals had been isolated from dose 1 until shortly after dose 2 and then challenged ($e_1=1$), under the assumption that vaccine protection was at its greatest shortly after dose 2. 
We give a detailed argument for this claim in Appendix~J, and present a sensitivity analysis which does not use this assumption in Table~S7 (Appendix~J).
\begin{figure}
    \centering
    \includegraphics{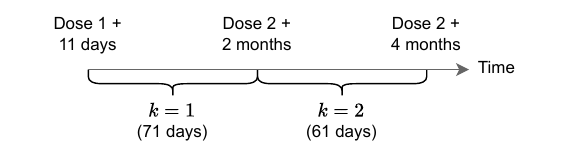}
    \caption{Illustration of time intervals $k=1$ and $k=2$}
    \label{fig: time intervals}
\end{figure}

For a given interval 1, investigators face a trade-off when choosing an appropriate length of interval 2: a longer interval 2 gives narrower bounds \eqref{eq:L_VE_challenge_2}-\eqref{eq:U_VE_challenge_2}, in addition to narrower confidence intervals, as more events are accrued. Thereby, a longer interval 2 can give a more sensitive test of vaccine waning. However, Assumption~\ref{ass:consistency 2 intervals} is more likely to be violated if interval 2 greatly exceeds interval 1 in length, as the quantity of viral particles per exposure may change over time, and the risk of multiple exposures per interval increases with a longer interval 2. This could lead the contrast  of estimates $\widehat{\text{VE}}_1^\mathrm{challenge}$\ vs.\ $\widehat{\text{VE}}_2^\mathrm{challenge}$ to compare different versions of infectious challenges if the lengths of intervals 1 and 2 differ greatly; for instance, $\widehat{\text{VE}}_1^\mathrm{challenge}$ and $\widehat{\text{VE}}_2^\mathrm{challenge}$ may quantify  challenge effects in a target trial where individuals are challenged with larger viral inocula during interval 1 compared to interval 2, and a difference between $\widehat{\text{VE}}_1^\mathrm{challenge}$ and $\widehat{\text{VE}}_2^\mathrm{challenge}$ may therefore not be due to vaccine waning.

\begin{table}[htbp]
  \centering
  \caption{Estimates and 95\% confidence intervals for \eqref{eq: observed VE}, \eqref{eq:L_VE_challenge_2}-\eqref{eq:U_VE_challenge_2} and \eqref{eq: LB minimal ass 2 intervals}-\eqref{eq: UB minimal ass 2 intervals}. Interval 1 ranged from 11 days after dose 1 until 2 months after dose 2 (71 days in total) and interval 2 ranged from 2 months after dose 2 until 4 months after dose 2 (61 days in total). Confidence intervals for the bounds $\mathcal{L}_\bullet,\mathcal{U}_\bullet$ were one-sided, whereas two-sided confidence intervals were used for VE estimates. The point estimates of $\widehat{\text{VE}}_1^{\mathrm{obs}}, \widehat{\mathcal{U}}_2$ and $\widehat{\mathcal{U}}_\psi$ are consistent with waning vaccine protection, but the confidence intervals include the null value of no waning. }
    % \resizebox{\linewidth}{!}{
    \begin{tabular}{lr}
         \hline Estimator & \multicolumn{1}{l}{Estimate (95\% CI)}  \\\hline
     $\widehat{\text{VE}}_1^{\text{obs}},\widehat{\text{VE}}_1^{\text{challenge}}$ & $0.95(0.93,0.97)$ \\
    $\widehat{\text{VE}}_2^{\text{obs}}$ & $0.90(0.87,0.93)$ \\
    $\widehat{\mathcal{L}}_2$ & $0.87 (0.84, -)$ \\
    $\widehat{\mathcal{U}}_2$ & $0.94(-,0.95)$ \\
    $\widehat{\mathcal{L}}_\psi$ & $0.36 (0.26, - )$\\
    $\widehat{\mathcal{U}}_\psi$ & $0.81(-,1.27)$
    \end{tabular}%
    % }
  \label{tab:Pfizer CI 3}%
\end{table}%

The upper confidence limit of $\widehat{ \mathcal{U}}_2$ was close to  $\widehat{\text{VE}}_1^\mathrm{challenge}$ (Table~\ref{tab:Pfizer CI 3}), although the upper confidence limit of $\widehat{\mathcal{U}}_\psi$ exceeded 1. 
In Table~S6 (Appendix~J), where interval 2 was extended until day 190, the upper confidence limit of $\widehat{ \mathcal{U}}_2$ was close to the lower confidence limit of $\widehat{\text{VE}}_1^\mathrm{challenge}$. 
Additionally, the upper confidence limit of $\widehat{\mathcal{U}}_\psi$ was smaller than 1. In Table~S7, we considered an additional choice of intervals to illustrate possible depletion of susceptible individuals between dose 1 and dose 2 of the vaccine, which gave wide bounds that included the null (no waning). Overall, our analyses suggest waning vaccine protection from interval 1 to interval 2.

Our primary analysis did not condition on any baseline covariates $L$, as individual patient data were not available. To illustrate the use of estimators described in Section~\ref{sec:maintext estimation} with individual-level data, we have included an additional analysis motivated by a malaria vaccine trial in Appendix~H.   As a sensitivity analysis, we conducted subgroup analyses of \citet{thomas_safety_2021} in Appendix~J,  which show that the cumulative incidence of the outcome by the end of follow-up was nearly constant across the following baseline covariates ($L$): age over 65, Charlson Comorbidity Index category $\geq 1$ and obesity, for both vaccine and placebo treatment. 
Therefore, we expect that estimates $\widehat{\mathrm{VE}}_1^\mathrm{obs}(l),\widehat{\mathcal{L}}_2(l),\widehat{\mathcal{U}}_2(l)$ conditional on the baseline covariates $l$ would have been close to the marginal estimates reported in Table~\ref{tab:Pfizer CI 3} for all baseline covariates $l$. Although we cannot guarantee the absence of residual confounding given $L$ ($U_{EY}$ in Figure~\ref{fig:SWIG_waning_2_intervals}(B)), it is plausible that the above choice of baseline covariates $L$ is sufficient to block open backdoor paths between exposure and infection status.

Our main analysis suggests that depletion of susceptible individuals could not alone account for the decline in vaccine efficacy over time.  Large real-world effectiveness studies that established waning of the BNT162b2 vaccine, such as \citet{goldberg_waning_2021} and \citet{levin_waning_2021}, were published in 2021. However, it would have been possible to estimate challenge effects \eqref{eq: VE_challenge_1 and VE_obs_1}-\eqref{eq:U_VE_challenge_2} using preliminary trial data from the BNT162b2 vaccine trial, published already in December 2020 \citep{polack_fernando_p_safety_2020}. This could have provided earlier evidence of vaccine waning. Such analyses could guide future vaccination policies during the window of time before booster trials are available.

In the Supplementary Material, we include \texttt{R} code to simulate a data generating mechanism that reaches the bounds in Theorem~\ref{thm: waning minimal ass 2 intevals} for the observed data distribution. This illustrates that the bounds are sharp, i.e., that the true value of $\text{VE}_2^\mathrm{challenge}$ in the vaccine trial could lie on the bounds $\widehat{\mathcal{L}}_2,\widehat{\mathcal{U}}_2$ under Assumptions~\ref{ass:consistency 2 intervals}-\ref{ass: exposure exchageability 2 intevals}.

\section{Discussion\label{sec:discussion}}

The challenge effect quantifies the vaccine waning that would be observed in a hypothetical challenge trial. Often, investigators only have access to data from a conventional randomized vaccine trial where participants are freely exposed to the infectious agent in the community. We have shown that sharp bounds on the challenge effect can be derived under plausible assumptions, and we illustrate that these bounds can confirm clinically significant waning. 

Our results are broadly applicable to vaccine trials using routinely collected data and can be extended to account for treatment outcome confounding in observational vaccine studies. Furthermore, the bounds on the challenge effect can be estimated using standard statistical methods that are implemented in commonly used software packages. As illustrated in Section~\ref{sec:example Pfizer trial}, the estimators for summary data can also be applied to re-analyze historical data from vaccine studies when individual-level data are not available, for example due to privacy concerns.

The challenge effect makes it possible to distinguish between settings where the observed vaccine efficacy diminishes solely due to depletion of susceptible individuals and settings where the vaccine protection wanes over time, and thereby addresses, and formalizes, an open problem in the analysis of vaccine trials \citep{halloran_study_1997}. The proposed methods can offer new empirical evidence of interest in health policy questions, for example about timing of vaccinations or booster doses.

\section{Acknowledgements}
This work was done while MJ was a PhD student at École Polytechnique Fédérale de Lausanne and was supported by the Swiss National Science Foundation. We would like to thank Julien D. Laurendeau and the anonymous referees for their insightful comments.

\section{Disclosure statement}
The authors report there are no competing interests to declare.

\bibliography{references}
\bibliographystyle{plainnat}

\appendix

\section{Identification under interference \label{app: interference}}

As in \citet{tsiatis_estimating_2022},  \citet{halloran_estimability_1996}, and (implicitly) in most randomized vaccine  trials, we have assumed that interference between the participants in the observed trial data is negligible, such that one participant's outcome does not  depend on another participant's treatment assignment (no interference among the participants). While it is often plausible that interference between the trial participants is negligible, there will often be interference in a setting where a vaccine program is rolled out in a human population. Thus, when applying conventional estimators under i.i.d.\ assumptions to the trial data, we draw valid superpopulation inference in a (fictive) population with potentially limited practical relevance. Yet, we will argue that the interference is not an issue when studying the challenge effect, unlike the usual vaccine efficacy estimand ($\text{VE}^\mathrm{obs}$). This is because the challenge effect is insensitive to the interference that arises in most infectious disease settings when treatments are rolled out.

To be explicit, consider a classical randomized vaccine trial, i.e., a trial without any controlled infectious challenges. 
Then, the vaccine efficacy ($\text{VE}^\mathrm{obs}$) in the (fictive) superpopulation with no interference will not correspond to $\text{VE}^\mathrm{obs}$ in a realistic target population, because there will be interference (e.g.\ herd immunity) in the realistic target population.

Consider now a challenge trial corresponding to \eqref{eq: challenge VE}.  Then, the challenge effect in the (fictive) superpopulation with no interference would be identical to the challenge effect in the (more realistic) target population with interference, because the interference is trivial under an intervention on exposure; indeed, there is no longer any interference when the exposure to the infectious agent is controlled (fixed) for all individuals. In this sense, considering effects under interventions on exposure will often be more practically relevant.

\section{Identification with two time intervals\label{app: identification with 2 time intervals}}

\begin{assumption}[Treatment exchangeability\label{ass: treatment exchageability 2 intevals}]\hfill

For all $a,e_1,e_2\in\{0,1\}$,

\begin{align*}
    &E_1^a, \Delta Y_1^a,\Delta Y_1^{a,e_1}, E_2^{a,e_1},\Delta Y_2^a, \Delta Y_2^{a,e_1,e_2}   \independent A \mid L ~.
\end{align*}
\end{assumption}

\begin{assumption}[Positivity\label{ass: exposure positivity 2 intervals}]
\begin{align*}
P(A =a \mid L) &> 0 \text{ for all } a\in\{0,1\}~ \text{ w.p. 1} ~.
\end{align*} 
\end{assumption}

Assumptions~\ref{ass: treatment exchageability 2 intevals} and \ref{ass: exposure positivity 2 intervals} hold by design when $A$ randomly assigned, for example in an RCT.

\subsection{Proof of Theorem~\ref{thm: waning minimal ass 2 intevals}\label{sec: proof of theorem}}
For the first time interval, 
\begin{align}
    &P(\Delta Y_1=1\mid A=a,L) \\
    \stackrel{\text{Assumption~\ref{ass:consistency 2 intervals}}}{=} &P( \Delta Y_1^{a} =1\mid A=a,L) \notag\\ 
    \stackrel{\text{Assumption~\ref{ass: exposure necessity 2 intervals}}}{=} &P(\Delta Y_1^{a}=1,E_1^a=1\mid A=a,L) \notag\\
 \stackrel{\text{Assumption~\ref{ass:consistency 2 intervals} }}{=} &P(\Delta Y_1^{a,e_1=1}=1,E_1^a=1\mid A=a,L)  \notag\\
 \stackrel{\text{Assumption~\ref{ass: exposure exchageability 2 intevals}}}{=} &P(\Delta Y_1^{a,e_1=1}=1\mid A=a,L)P(E_1^a=1\mid A=a,L) \notag\\
 \stackrel{\text{Assumptions~\ref{ass: treatment exchageability 2 intevals},\ref{ass: exposure positivity 2 intervals}}}{=} &P(\Delta Y_1^{a,e_1=1}=1\mid L)P(E_1^a=1\mid L)~. \label{eq: interval 1 a}
\end{align}
Taking the ratio of \eqref{eq: interval 1 a} for $a=1$ vs. $a=0$, and using Assumption~\ref{ass:no effect on exposure 2 intervals} to cancel the ratio of exposure probabilities gives
\begin{align*}
    \frac{E[\Delta Y_1^{a=1,e_1=1}\mid L]}{E[\Delta Y_1^{a=0,e_1=1}\mid L]} = \frac{E[\Delta Y_1\mid A=1,L]}{E[\Delta Y_1\mid A=0,L]}~,
\end{align*}
and therefore $\text{VE}_1^{\mathrm{challenge}}=\text{VE}_1^\mathrm{obs}$.  

For the second time interval,
\begin{align}
    &P(\Delta Y_2^{e_1=0}=1\mid A=a,L) \notag\\
    \stackrel{\text{Assumption~\ref{ass:consistency 2 intervals}}}{=}
    &P(\Delta Y_2^{a,e_1=0}=1\mid A=a,L) \notag\\
    \stackrel{\text{Assumption~\ref{ass: exposure necessity 2 intervals}}}{=} &P(E_2^{a,e_1=0}=1,\Delta Y_2^{a,e_1=0}=1\mid A=a,L) \notag\\
    \stackrel{\text{Assumption~\ref{ass:consistency 2 intervals}}}{=}
    &P(E_2^{a,e_1=0}=1,\Delta Y_2^{a,e_1=0,e_2=1}=1 \mid A=a,L) \label{eq: cross world intermediate step} \\
    \stackrel{\text{Assumption~\ref{ass: exposure exchageability 2 intevals}}}{=} &P(\Delta Y_2^{a,e_1=0,e_2=1}=1 \mid A=a,L)P(E_2^{a,e_1=0}=1\mid A=a,L) \notag\\
    \stackrel{\text{Assumptions~\ref{ass: treatment exchageability 2 intevals},\ref{ass: exposure positivity 2 intervals}}}{=} &P(\Delta Y_2^{a,e_1=0,e_2=1}=1 \mid L)P(E_2^{a,e_1=0}=1\mid L) ~. \label{eq: identified single world inequality}
\end{align}

Taking the ratio of \eqref{eq: identified single world inequality} for $a=1$ vs. $a=0$  and using Assumption~\ref{ass: exclusion 2 intervals} gives
\begin{align*}
    &\frac{E[\Delta Y_2 \mid A=1,L]}{E[\Delta Y_1 + \Delta Y_2\mid A=0,L]} \\
    &\quad \leq \frac{E[\Delta Y_2^{a=1,e_1=0,e_2=1}\mid L]}{E[\Delta Y_2^{a=0,e_1=0,e_2=1}\mid L]}\cdot\frac{P(E_2^{a=1,e_1=0}=1\mid L)}{P(E_2^{a=0,e_1=0}=1\mid L)} \\
    &\qquad \leq \frac{E[\Delta Y_1 + \Delta Y_2\mid A=1,L]}{E[\Delta Y_2\mid A=0,L]} ~.
\end{align*}
Finally, using Assumption~\ref{ass:no effect on exposure 2 intervals} to cancel the ratio of exposure probabilities, we obtain
\begin{align*}
    \frac{E[\Delta Y_2 \mid A=1,L]}{E[\Delta Y_1 + \Delta Y_2\mid A=0,L]} \leq \frac{E[\Delta Y_2^{a=1,e_1=0,e_2=1}\mid L]}{E[\Delta Y_2^{a=0,e_1=0,e_2=1}\mid L]} \leq \frac{E[\Delta Y_1 + \Delta Y_2\mid A=1,L]}{E[\Delta Y_2\mid A=0,L]} ~.
\end{align*}

To establish sharpness of the bounds \eqref{eq:L_VE_challenge_2} and \eqref{eq:U_VE_challenge_2}, it is sufficient to show that there exists a counterfactual data generating mechanism that attains the bounds. An example that attains the lower bound $\mathcal{L}_2$ is given below. We define  $p_{k,a,l}=E[\Delta Y_k\mid A=a,L=l]$ for all $k,a,l$ and denote the observed laws of $A,L$ by $P_A,P_L$ respectively.
\begin{dgm} \label{dgm: attaining bounds} \hfill 
    \begin{enumerate}[(I)]
    \item $L\sim P_L$
    \item $A\sim P_A$
    \item $U_Y\sim \mathrm{Unif}[0,1]$
    \item $E_1^{a}=E_2^{a,e_1}=1$ for all $a,e_1$
    \item \begin{enumerate}
        \item $\Delta Y_1^{a,e_1=1}=I(U_Y\leq p_{k=1,a,L})$ for all $a$
        \item $\Delta Y_1^{a,e_1=0}=0$ for all $a$
    \end{enumerate}
    \item \begin{enumerate}[(a)]
        \item $\Delta Y_2^{a=1,e_1=0,e_2=1}=I(U_Y\leq p_{k=1,a=1,L}+p_{k=2,a=1,L})$
        \item $\Delta Y_2^{a=0,e_1=0,e_2=1}=I(p_{k=1,a=0,L} < U_Y \leq p_{k=1,a=0,L}+p_{k=2,a=0,L})$
        \item $\Delta Y_2^{a,e_1=1,e_2=1}=\Delta Y_2^{a,e_1=0,e_2=1}I(\Delta Y_1^{a,e_1=1}=0)$ for all $a$
        \item $\Delta Y_2^{a,e_1,e_2=0}=0$ for all $a,e_1$
    \end{enumerate}
\end{enumerate}
All other counterfactuals are understood to be recursively related to (I)-(VI) through Definition~1 of \citet{thomas_s_richardson_single_2013}.
\end{dgm}

The data generating mechanism generalizes the example in Figure~\ref{fig:weak_exclusion_restriction}. It is straightforward to verify that Data generating mechanism~\ref{dgm: attaining bounds} satisfies Assumptions~\ref{ass:consistency 2 intervals}-\ref{ass: exposure exchageability 2 intevals} and \ref{ass: treatment exchageability 2 intevals}-\ref{ass: exposure positivity 2 intervals}. Furthermore, the lower bound $\mathcal{L}_2$ is attained since
\begin{align*}
    &E[\Delta Y_2^{a=1,e_1=0,e_2=1}\mid L] \\
   =&E[I(U_Y\leq p_{k=1,a=1,L}+p_{k=2,a=1,L})\mid L] \\
   =& p_{k=1,a=1,L}+p_{k=2,a=1,L} \\
   =&E[\Delta Y_1+\Delta Y_2\mid A=1,L] ~,
\end{align*}
and
\begin{align*}
    &E[\Delta Y_2^{a=0,e_1=0,e_2=1}\mid L] \\
    =&E[I(p_{k=1,a=0,L}<U_Y\leq p_{k=1,a=0,L}+p_{k=2,a=0,L})\mid L] \\
    =&p_{k=2,a=0,L} \\
    =&E[\Delta Y_2\mid A=0,L] ~.
\end{align*}
Additionally, Data generating mechanism~\ref{dgm: attaining bounds} is consistent with the observed cumulative incidences $E[\Delta Y_k\mid A=a,L]$ for all $a,k$, since
\begin{align*}
    &E[\Delta Y_1\mid A=a,L] \\
    \stackrel{\text{Assumption~\ref{ass:consistency 2 intervals}}}{=} &E[\Delta Y_1^a \mid A=a,L] \\
    \stackrel{\text{Assumption~\ref{ass:consistency 2 intervals}}}{=}&E[\Delta Y_1^{a,e_1=1} \mid A=a,L] \text{ since } E_1^a=1 \text{ w.p. 1} \\
    =&E[I(U_{Y}\leq p_{k=1,a,L}) \mid A=a,L]  \\
    =& p_{k=1,a,L}~,
\end{align*}
and
\begin{align*}
    &E[\Delta Y_2\mid A=a,L] \\
    \stackrel{\text{Assumption~\ref{ass:consistency 2 intervals}}}{=} &E[\Delta Y_2^a \mid A=a,L] \\
    \stackrel{\text{Assumption~\ref{ass:consistency 2 intervals}}}{=}&E[\Delta Y_2^{a,e_1=1,e_2=1} \mid A=a,L] \text{ since } E_1^a=E_2^{a,e_1}=1 \text{ w.p. 1} \\
    =&E[I(p_{k=1,a,L}< U_{Y}\leq p_{k=1,a,L} +p_{k=2,a,L}) \mid A=a,L]  \\
    =& p_{k=2,a,L}~.
\end{align*}

Finally, the upper bound $\mathcal{U}_2$ can be reached by permuting treatment groups $a=0\leftrightarrow a=1$ in (VI) (a) and (b) of Data generating mechanism~\ref{dgm: attaining bounds}. An implementation of Data generating mechanism~\ref{dgm: attaining bounds} in \texttt{R} for a simplified setting with $A\sim \mathrm{Ber}(1/2)$ and without $L$ is given in the Supplementary Material.

\subsection{Proof of Proposition~\ref{prp: homogeneity}}
We have that 
\begin{align}
    &P(\Delta Y_2=1\mid \Delta Y_1=0,A=a,L) \notag\\
    \stackrel{\text{Assumption~\ref{ass:consistency 2 intervals}}}{=} &P(\Delta Y_2^a=1\mid \Delta Y_1^a=0,A=a,L) \notag\\
    \stackrel{\text{Assumption~\ref{ass: exposure necessity 2 intervals}}}{=}& P(\Delta Y_2^{a}=1,E_2^a=1\mid \Delta Y_1^{a}=0, A=a,L) \notag\\
    =&P(\Delta Y_2^{a}=1\mid E_2^a=1, \Delta Y_1^a=0,A=a,L)P(E_2^a=1\mid \Delta Y_1^a=0,A=a,L) \notag\\
    =&P(\Delta Y_2^a=1\mid E_1^a=0,E_2^a=1,\Delta Y_1^a=0,A=a,L)\notag\\
    &\quad\times P(E_2^a=1\mid E_1^a=0,\Delta Y_1^a=0,A=a,L) W_{Y,a}W_{E,a} \notag\\
    \stackrel{\text{Assumption~\ref{ass: exposure necessity 2 intervals}}}{=}&P(\Delta Y_2^a=1\mid E_1^a=0,E_2^a=1,A=a,L)P(E_2^a=1\mid E_1^a=0,A=a,L) W_{Y,a}W_{E,a} \notag\\
    =&P(\Delta Y_2^a=1,E_2^a=1\mid E_1^a=0,A=a,L) W_{Y,a}W_{E,a} \notag\\
    \stackrel{\text{Assumption~\ref{ass: exposure necessity 2 intervals}}}{=}&P(\Delta Y_2^a=1\mid E_1^a=0,A=a,L) W_{Y,a}W_{E,a} \notag\\
    \stackrel{\text{Assumption~\ref{ass:consistency 2 intervals}}}{=} 
    &P(\Delta Y_2^{a,e_1=0}=1\mid E_1^a=0,A=a,L)W_{Y,a}W_{E,a}\notag\\
    \stackrel{\eqref{eq: homogeneous_exch}}{=}&P(\Delta Y_2^{a,e_1=0}=1\mid A=a,L)W_{Y,a}W_{E,a} \notag\\
    \stackrel{\text{Assumption~\ref{ass:consistency 2 intervals}}}{=}&P(\Delta Y_2^{e_1=0}=1\mid A=a,L)W_{Y,a}W_{E,a}~, \notag
\end{align}
where we have defined the weights
\begin{align*}
    W_{Y,a} &= \frac{P(\Delta Y_2^a=1\mid E_2^a=1,\Delta Y_1^a=0,A=a,L)}{P(\Delta Y_2^a=1\mid E_2^a=1,\Delta Y_1^a=0, E_1^a=0,A=a,L)} ~,\\
    W_{E,a}&= \frac{P(E_2^a=1\mid \Delta Y_1^a=0,A=a,L)}{P(E_2^a=1\mid E_1^a=0,\Delta Y_1^a=0,A=a,L)} ~.
\end{align*}
When \eqref{eq: no UY}-\eqref{eq: no UE} hold, then $W_{Y,a}=W_{E,a}=1$, which implies \eqref{eq: homogeneity in Y and E}. Taking the ratio of \eqref{eq: identified single world inequality} for $a=1$ vs. $a=0$, and using \eqref{eq: homogeneity in Y and E},
\begin{align*}
    &\frac{E[\Delta Y_2^{a=1,e_1=0,e_2=1}\mid L]}{E[\Delta Y_2^{a=0,e_1=0,e_2=1}\mid L]}=\frac{P(\Delta Y_2=1\mid \Delta Y_1=0, A=1,L)}{P(\Delta Y_2=1\mid \Delta Y_1=0, A=0,L)}  \cdot \frac{P(E_2^{a=0,e_1=0}=1\mid L)}{P(E_2^{a=1,e_1=0}=1\mid L)} \\
   =  &\frac{P(\Delta Y_2=1\mid \Delta Y_1=0, A=1,L)}{P(\Delta Y_2=1\mid \Delta Y_1=0, A=0,L)} ~,
\end{align*}
where we have used Assumption~\ref{ass:no effect on exposure 2 intervals} to cancel the ratio of exposure probabilities in the final line.

\subsubsection{Proof that \eqref{eq: homogeneity in Y and E} implies Assumption~\ref{ass: exclusion 2 intervals}}
Expressing $E[\Delta Y_1+\Delta Y_2\mid A=a,L]$ as the convex combination 
\begin{align*}
    P(\Delta Y_1=1\mid A=a,L)\cdot 1 + (1-P(\Delta Y_1=1\mid A=a,L))E[\Delta Y_2\mid \Delta Y_1=0, A=a,L]
\end{align*}
implies that $E[\Delta Y_2\mid \Delta Y_1=0,A=a,L]\leq E[\Delta Y_1+\Delta Y_2\mid A=a,L]$. Next, writing $E[\Delta Y_2\mid \Delta Y_1=0,A=a,L]$ as $E[\Delta Y_2\mid A=a,L]/(P(\Delta Y_1=0,\mid A=a,L))$ gives $E[\Delta Y_2\mid \Delta Y_1=0,A=a,L]\geq E[\Delta Y_2\mid A=a,L]$.

\section{Motivation for Assumption~\ref{ass: exclusion 2 intervals}\label{app: motivating exclusion restriction}}

In this section we assume a conventional causal model where all nodes are intervenable and where counterfactuals are defined recursively, according to Definition~1 in \citet{thomas_s_richardson_single_2013}. In particular, suppose that it is possible to intervene on $\Delta Y_1$, for example through a form of post exposure treatment that prevents the infection from developing.  An example of such an intervention is antiretroviral post-exposure prophylaxis (PEP) for HIV \citep{dehaan_pep_2022}. Next, assume the population level exclusion restriction
\begin{align}
    P(\Delta Y_2^{a,e_1=1,\Delta y_1=0}=1\mid L ) = P(\Delta Y_2^{a,e_1=0,\Delta y_1=0}=1\mid L ) \text{ w.p. 1}~, \label{eq: exclusion Y intervention}
\end{align}
which can be violated by arrows $E_1\rightarrow E_2$ or $E_1\rightarrow \Delta Y_2$ in Figure~\ref{fig:SWIG_waning_2_intervals}(A). Then,
\begin{align*}
    &P(\Delta Y_2^{a}=1\mid L) \\
    = &P(\Delta Y_1^a=0,\Delta Y_2^a=1 \mid L) \\
    = &P(\Delta Y_1^a=0,\Delta Y_2^{a,\Delta Y_1^a}=1 \mid L) \\
    =&P(\Delta Y_1^a=0,\Delta Y_2^{a,\Delta y_1=0}=1 \mid L) \\
    =&P(\Delta Y_2^{a,\Delta y_1=0}=1\mid L) \\
    &\quad - P(\Delta Y_1^a=1,\Delta Y_2^{a,\Delta y_1=0}=1 \mid L) ~.
\end{align*}
We have used the recursive definition of counterfactuals (Definition~1 in \citet{thomas_s_richardson_single_2013}) in the third line. 
Using the fact that 
\begin{align*}
    0\leq P(\Delta Y_1^a=1, \Delta Y_2^{a,\Delta y_1=0}=1\mid L) \leq P(\Delta Y_1^a=1\mid L) ~,
\end{align*}
it follows that
\begin{align*}
    P(\Delta Y_2^a=1\mid L) \leq P(\Delta Y_2^{a,\Delta y_1=0}=1\mid L) \leq P(\Delta Y_1^a=1\mid L)+P(\Delta Y_2^a=1\mid L) ~.
\end{align*}
Finally, we obtain \eqref{eq: sufficient single-world inequality} from the fact that and $ P(\Delta Y_k^a=1\mid L)= P(\Delta Y_k=1\mid A=a,L)$ for $k\in\{1,2\}$ by Assumptions~\ref{ass:consistency 2 intervals} and \ref{ass: treatment exchageability 2 intevals}-\ref{ass: exposure positivity 2 intervals}, and
\begin{align*}
    &P(\Delta Y_2^{a,\Delta y_1=0}=1\mid L) \\
    = &P(\Delta Y_2^{a,E_1^a,\Delta y_1=0}=1\mid L) \\
    \stackrel{\eqref{eq: exclusion Y intervention}}{=} &P(\Delta Y_2^{a,e_1=0,\Delta y_1=0}=1\mid L) \\
    =&P(\Delta Y_2^{a,e_1=0,\Delta Y_1^{a,e_1=0}}=1\mid L) \\
    = &P(\Delta Y_2^{a,e_1=0}=1\mid L) ~.
\end{align*}
The penultimate line used the fact that $P(\Delta Y_1^{a,e_1=0}=0\mid L)=1$, which holds because
\begin{align*}
    &P(\Delta Y_1^{a,e_1=0}=0\mid L) \\
    \stackrel{\text{Assumptions~\ref{ass:consistency 2 intervals},\ref{ass: treatment exchageability 2 intevals},\ref{ass: exposure positivity 2 intervals}}}{=}&P(\Delta Y_1^{a,e_1=0}=0\mid A=a,L) \\
    \stackrel{\text{Assumption~\ref{ass: exposure exchageability 2 intevals}}}{=}& P(\Delta Y_1^{a,e_1=0}=0\mid E_1^a=0,A=a,L) \\
    \stackrel{\text{Assumption~\ref{ass:consistency 2 intervals}}}{=} &P(\Delta Y_1^{a}=0\mid E_1^a=0,A=a,L) \\
    \stackrel{\text{Assumption~\ref{ass: exposure necessity 2 intervals}}}{=} &1~.
\end{align*}

\subsection{Alternative motivation using cross world assumptions}
Instead of \eqref{eq: exclusion Y intervention},  suppose the following cross-world equality holds:
    \begin{align}
        P(\Delta Y_2^{a,e_1=0}=1\mid \Delta Y_1^a=0,L) &=P(\Delta Y_2^{a}=1\mid \Delta Y_1^a=0,L) \text{ w.p. 1}~. \label{eq:cross world 2 intervals i} 
    \end{align}
The equality \eqref{eq:cross world 2 intervals i} is motivated by the following: the infectious outcome during time interval 2 would be the same in those who were naturally uninfected, regardless of whether or not they would have been isolated during time interval 1.  
However, this justification may be deceptively simple: the equality \eqref{eq:cross world 2 intervals i} is difficult to justify in principle because it involves a cross-world quantity on the left hand side, which is difficult to interpret. Therefore, although \eqref{eq:cross world 2 intervals i} may provide intuition for \eqref{eq: sufficient single-world inequality}, we do not endorse the assumption as a sufficient justification for \eqref{eq: sufficient single-world inequality}.

Expression \eqref{eq:cross world 2 intervals i} implies the single world inequality \eqref{eq: sufficient single-world inequality}.
To see this, we have from the laws of probability that
\begin{align*}
    &P(\Delta Y_2^{a,e_1=0}=1\mid L)\\
    = &P(\Delta Y_2^{a,e_1=0}=1\mid \Delta Y_1^a=0, L)P(\Delta Y_1^a=0\mid L)\\
    &\quad + P(\Delta Y_2^{a,e_1=0}=1\mid \Delta Y_1^a=1, L)P(\Delta Y_1^a=1\mid L) ~.
\end{align*}
The inequality \eqref{eq: sufficient single-world inequality} follows from substituting \eqref{eq:cross world 2 intervals i} on the right hand side, and using the inequality $0\leq P(\Delta Y_2^{a,e_1=0}=1\mid \Delta Y_1^a=1, L)\leq 1$.

\section{Identification with multiple versions of exposure\label{app: multiple treatment versions}}
\subsection{Identification assumptions}

In this section, we establish conditions that allow identification and hypothesis testing of vaccine waning in the presence of multiple versions of treatment.  Suppose that multiple versions of the observed exposure $E_k=1$ are present in the observed data. For example, among individuals with $E_k=1$, there may be  subgroups that are exposed 1) \textit{more than} once during interval $k$, 2) to larger (or smaller) infectious inocula per infectious exposure or 3) at different anatomical barriers, for example gastrointestinal versus respiratory mucosa. However, the counterfactuals $\Delta Y_1^{a,e_1=1}$ and $\Delta Y_2^{a,e_1=0,e_2=1}$ refer to potential outcomes under a particular controlled (binary) exposure $e_k=1$, e.g.\ by intranasal challenge with a \textit{particular} quantity of infectious inoculum, in keeping with Section~\ref{sec:interventions on exposure}. 
We denote the exposure version (quantity of the controlled infectious inoculum) by $q$. 
Throughout this section we will continue to assume that there is only one version of \textit{not} being exposed, such that Assumption~\ref{ass:consistency 2 intervals} holds whenever $e_1=0$ and $e_2=0$.

To accommodate multiple versions of the observed exposure $E_k=1$, we will consider a modified version $G$ of the vaccine trial described in Section~\ref{sec:observed data}, motivated by \citet{vanderweele_causal_2013} and \citet {vanderweele_constructed_2022}. In the modified trial $G$ (Figure~\ref{fig:DAGs_multiple_versions}), individuals are first randomized to vaccine versus placebo, $A\in\{0,1\}$, and then to a version of infectious exposure $Q_1,Q_2\in\mathcal{Q}$ during intervals 1 and 2 respectively. The support $\mathcal{Q}$ contains a collection of well-defined controlled procedures to expose individuals to infectious inocula in different ways; for example by intranasal challenge with a pipette containing a random quantity $Q_k$ of infectious inoculum, drawn from a pre-specified distribution. As before, we assume that there is only one version of $Q_k$ where individuals are \textit{not} exposed to the infectious agent, and denote this by $Q_k=0$. 
In the modified trial $G$, we define the exposure status $E_k$ to be a coarse-grained version of $Q_k$: for $k\in\{1,2\}$, let 
\begin{align}
    E_k=I(Q_k\neq 0)  \text{ and } E_2^{q_1=0}=I(Q_2^{q_1=0}\neq 0) ~.\label{eq: Ek definition}
\end{align}

To establish a relation between the modified trial $G$ and the original vaccine trial, we introduce the following assumption.
\begin{assumption}[Equivalence of the modified trial $G$\label{ass:equivalence_modified_trial}]
    The modified trial $G$ is conducted in an identical population to the trial in Section~\ref{sec:observed data}, with a choice of $\mathcal{Q}$ and randomization rule for $A,Q_1,Q_2$ such that
    \begin{align}
        P_G(A,L,E_1,\Delta Y_1,E_2,\Delta Y_2)=P(A,L,E_1,\Delta Y_1,E_2,\Delta Y_2)~,\label{eq:margin1} \\
        P_G(A,L,E_1,\Delta Y_1^{q_1=0},E_2^{q_1=0},\Delta Y_2^{q_1=0})=P(A,L,E_1,\Delta Y_1^{e_1=0},E_2^{e_1=0},\Delta Y_2^{e_1=0})~, \label{eq:margin2}
    \end{align}
    where right hand side is the distribution of the trial in Section~\ref{sec:observed data}.
\end{assumption}
Importantly, the original trial in Section~\ref{sec:observed data} and the modified trial $G$ are not necessarily identical, but we require that Assumption~\ref{ass:equivalence_modified_trial} holds. The equivalence assumption can fail if the controlled infectious exposure versions in $\mathcal{Q}$ are not representative of the infectious exposures in the original trial, meaning that there does not exist any randomization rules for $A,Q_1,Q_2$ for controlled exposures $\mathcal{Q}$ in the trial $G$  such that $P_G$ satisfies \eqref{eq:margin1}-\eqref{eq:margin2}.

\begin{figure}
    \centering
% \subfloat[]{
%  \resizebox{0.6\columnwidth}{!}{
    \begin{tikzpicture}
\tikzset{line width=1.5pt, outer sep=0pt,
ell/.style={draw,fill=white, inner sep=2pt,
line width=1.5pt},
swig vsplit={gap=5pt,
inner line width right=0.5pt}};
\node[name=A,ell,  shape=ellipse] at (3,0) {$A$};
\node[name=Q1,ell,  shape=ellipse] at (6,0) {$Q_1$};
% \node[name=E1,ell,  shape=ellipse] at (8,0) {$E_1$};
\node[name=Q2,ell,  shape=ellipse] at (12,0) {$Q_2$};
% \node[name=E2,ell,  shape=ellipse] at (14,0) {$E_2$};
\node[name=Y1,ell,  shape=ellipse] at (6,2) {$\Delta Y_1$};
\node[name=Y2,ell,  shape=ellipse] at (12,2) {$\Delta Y_2$};
\node[name=L,ell,  shape=ellipse] at (3,2) {$L$};
% \node[name=UAY,ell,  shape=ellipse] at (0,1) {$U_{AY}$};
% \node[name=UEY1,ell,  shape=ellipse] at (12,2) {$U_{EY,1}$};
% \node[name=UEY2,ell,  shape=ellipse] at (12,0) {$U_{EY}$};
\node[name=UY,ell,  shape=ellipse] at (9,4) {$U_{Y}$};
\node[name=UE,ell,  shape=ellipse] at (9,-2) {$U_{E}$};
\begin{scope}[transparency group, opacity=0.3] 
    % \path [->] (A) edge (E1);
    \path [->,>={Stealth[black]}] (A) edge (Y1);
    \path [->,>={Stealth[black]}] (A) edge (Y2);
    % \path [->,>={Stealth[black]}] (A) edge[bend right] (E2);
    %\path [->,>={Stealth[black]}] (E1) edge (E2);
    % \path[->,>={Stealth[black]}]  (Q1) edge (E1);
    \path[->,>={Stealth[black]}]  (Q1) edge  (Y1);
    \path[->,>={Stealth[black]}]  (Q2) edge  (Y2);
    % \path[->,>={Stealth[black]}]  (Q2) edge  (E2);
    \path [->,>={Stealth[black]}] (Y1) edge (Y2);
    \path [->,>={Stealth[black]}] (Y1) edge (Q2);
    \path [->,>={Stealth[black]}] (L) edge (Q1);
    \path [->,>={Stealth[black]}] (L) edge (Q2);
    \path [->,>={Stealth[black]}] (L) edge (Y1);
    \path [->,>={Stealth[black]}] (L) edge[bend left=20] (Y2);
    \path[->,>={Stealth[black]}]  (UY) edge[black] (Y1);
    \path[->,>={Stealth[black]}]  (UY) edge[black] (Y2);
    % \path[->,>={Stealth[black]}]  (UE) edge[black] (E1);
    \path[->,>={Stealth[black]}]  (UE) edge[black] (Q2);
    \path[->,>={Stealth[black]}]  (UE) edge[black] (Q1);
\end{scope}
\begin{scope}[>={Stealth[black]},
              every edge/.style={draw=black,very thick}]
    
    % \path[->,>={Stealth[red]}]  (UAY) edge[red] (A);
    % \path[->,>={Stealth[red]}]  (UAY) edge[red] (Y1);
    % \path[->,>={Stealth[red]}]  (UEY1) edge[red] (E1.140);
    % \path[->,>={Stealth[red]}]  (UEY1) edge[red, bend right] (Y1);
    % \path[->,>={Stealth[red]}]  (UEY2) edge[red, bend left] (E2.215);
    % \path[->,>={Stealth[red]}]  (UEY2) edge[red] (Y2);
    % \path[->,>={Stealth[red]}]  (E1) edge[red] (Y2);
    % \path[->,>={Stealth[red]}]  (E1) edge[red] (E2);
    % \path[->,>={Stealth[blue]}]  (UY) edge[blue] (Y1);
    % \path[->,>={Stealth[blue]}]  (UY) edge[blue] (Y2);
    % \path[->,>={Stealth[blue]}]  (UE) edge[blue, bend left] (E1.215);
    % \path[->,>={Stealth[blue]}]  (UE) edge[blue] (E2);
    % \path[->,>={Stealth[red]}]  (A) edge[red] (E1);
\end{scope}
\end{tikzpicture}

\caption{Modified trial $G$ with exposure versions $Q_1$ and $Q_2$}
    \label{fig:DAGs_multiple_versions}
\end{figure}
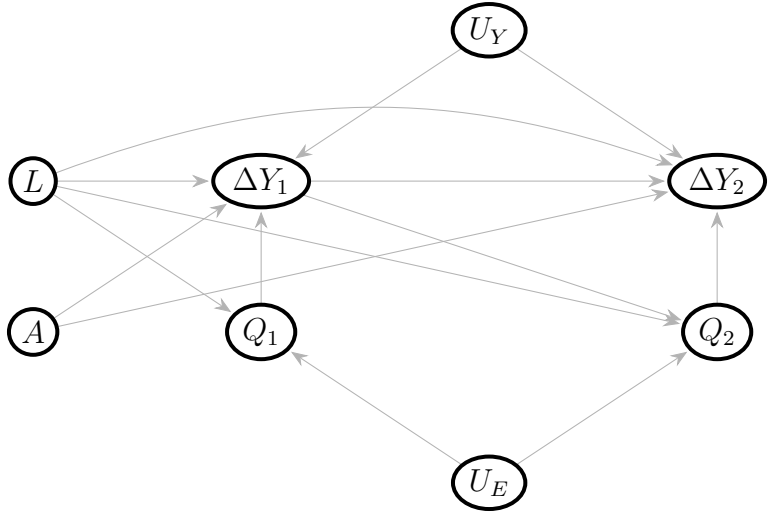

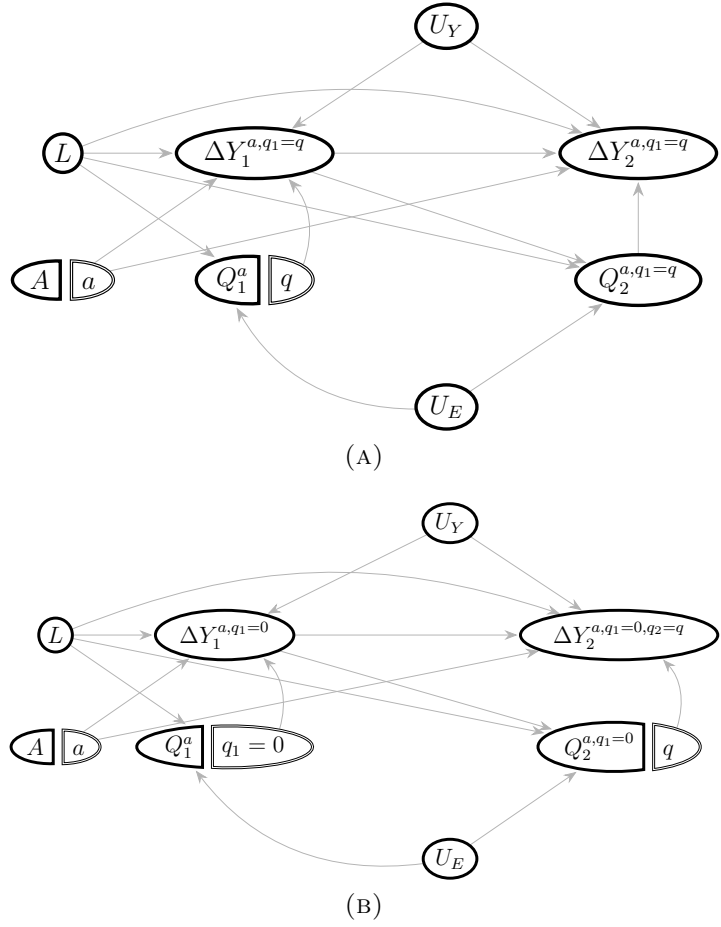
\begin{figure}
    \centering
\subfloat[]{
 \resizebox{0.6\columnwidth}{!}{
    \begin{tikzpicture}
\tikzset{line width=1.5pt, outer sep=0pt,
ell/.style={draw,fill=white, inner sep=2pt,
line width=1.5pt},
swig vsplit={gap=5pt,
inner line width right=0.5pt}};
\node[name=A,shape=swig vsplit] at (3,0){
\nodepart{left}{$A$}
\nodepart{right}{$a$} };
\node[name=Q1,shape=swig vsplit] at (6,0){
\nodepart{left}{$Q_1^{a}$}
\nodepart{right}{$q$} };
% \node[name=E1,ell,  shape=ellipse] at (8,0) {$E_1^{a,q}$};
\node[name=Q2,ell,  shape=ellipse] at (12,0) {$Q_2^{a,q_1=q}$};
% \node[name=E2,ell,  shape=ellipse] at (14,0) {$E_2^{a,q}$};
\node[name=Y1,ell,  shape=ellipse] at (6,2) {$\Delta Y_1^{a,q_1=q}$};
\node[name=Y2,ell,  shape=ellipse] at (12,2) {$\Delta Y_2^{a,q_1=q}$};
\node[name=L,ell,  shape=ellipse] at (3,2) {$L$};
% \node[name=UAY,ell,  shape=ellipse] at (0,1) {$U_{AY}$};
% \node[name=UEY1,ell,  shape=ellipse] at (12,2) {$U_{EY,1}$};
% \node[name=UEY2,ell,  shape=ellipse] at (12,0) {$U_{EY}$};
\node[name=UY,ell,  shape=ellipse] at (9,4) {$U_{Y}$};
\node[name=UE,ell,  shape=ellipse] at (9,-2) {$U_{E}$};
\begin{scope}[transparency group, opacity=0.3] 
    % \path [->] (A) edge (E1);
    \path [->,>={Stealth[black]}] (A) edge (Y1);
    \path [->,>={Stealth[black]}] (A) edge (Y2);
    % \path [->,>={Stealth[black]}] (A) edge[bend right] (E2);
    %\path [->,>={Stealth[black]}] (E1) edge (E2);
    % \path[->,>={Stealth[black]}]  (Q1) edge (E1);
    \path[->,>={Stealth[black]}]  (Q1.20) edge[bend right]  (Y1);
    \path[->,>={Stealth[black]}]  (Q2) edge  (Y2);
    % \path[->,>={Stealth[black]}]  (Q2) edge  (E2);
    \path [->,>={Stealth[black]}] (Y1) edge (Y2);
    \path [->,>={Stealth[black]}] (Y1) edge (Q2);
    \path [->,>={Stealth[black]}] (L) edge (Q1);
    \path [->,>={Stealth[black]}] (L) edge (Q2);
    \path [->,>={Stealth[black]}] (L) edge (Y1);
    \path [->,>={Stealth[black]}] (L) edge[bend left=20] (Y2);
    \path[->,>={Stealth[black]}]  (UY) edge[black] (Y1);
    \path[->,>={Stealth[black]}]  (UY) edge[black] (Y2);
    % \path[->,>={Stealth[black]}]  (UE) edge[black] (E1);
    \path[->,>={Stealth[black]}]  (UE) edge[black] (Q2);
    \path[->,>={Stealth[black]}]  (UE) edge[black, bend left] (Q1.240);
\end{scope}
\begin{scope}[>={Stealth[black]},
              every edge/.style={draw=black,very thick}]
    
    % \path[->,>={Stealth[red]}]  (UAY) edge[red] (A);
    % \path[->,>={Stealth[red]}]  (UAY) edge[red] (Y1);
    % \path[->,>={Stealth[red]}]  (UEY1) edge[red] (E1.140);
    % \path[->,>={Stealth[red]}]  (UEY1) edge[red, bend right] (Y1);
    % \path[->,>={Stealth[red]}]  (UEY2) edge[red, bend left] (E2.215);
    % \path[->,>={Stealth[red]}]  (UEY2) edge[red] (Y2);
    % \path[->,>={Stealth[red]}]  (E1) edge[red] (Y2);
    % \path[->,>={Stealth[red]}]  (E1) edge[red] (E2);
    % \path[->,>={Stealth[blue]}]  (UY) edge[blue] (Y1);
    % \path[->,>={Stealth[blue]}]  (UY) edge[blue] (Y2);
    % \path[->,>={Stealth[blue]}]  (UE) edge[blue, bend left] (E1.215);
    % \path[->,>={Stealth[blue]}]  (UE) edge[blue] (E2);
    % \path[->,>={Stealth[red]}]  (A) edge[red] (E1);
\end{scope}
\end{tikzpicture}
 }
}

\subfloat[]{
 \resizebox{0.6\columnwidth}{!}{
    \begin{tikzpicture}
\tikzset{line width=1.5pt, outer sep=0pt,
ell/.style={draw,fill=white, inner sep=2pt,
line width=1.5pt},
swig vsplit={gap=5pt,
inner line width right=0.5pt}};
\node[name=A,shape=swig vsplit] at (3,0){
\nodepart{left}{$A$}
\nodepart{right}{$a$} };
\node[name=Q1,shape=swig vsplit] at (6,0){
\nodepart{left}{$Q_1^{a}$}
\nodepart{right}{$q_1=0$} };
% \node[name=E1,ell,  shape=ellipse] at (9,0) {$E_1^{a,q_1=0}$};
\node[name=Q2,shape=swig vsplit] at (13,0){
\nodepart{left}{$Q_2^{a,q_1=0}$}
\nodepart{right}{$q$} };
% \node[name=E2,ell,  shape=ellipse] at (16.5,0) {$E_2^{a,q_1=0,q_2=q}$};
\node[name=Y1,ell,  shape=ellipse] at (6,2) {$\Delta Y_1^{a,q_1=0}$};
\node[name=Y2,ell,  shape=ellipse] at (13,2) {$\Delta Y_2^{a,q_1=0,q_2=q}$};
\node[name=L,ell,  shape=ellipse] at (3,2) {$L$};
% \node[name=UAY,ell,  shape=ellipse] at (0,1) {$U_{AY}$};
% \node[name=UEY1,ell,  shape=ellipse] at (12,2) {$U_{EY,1}$};
% \node[name=UEY2,ell,  shape=ellipse] at (12,0) {$U_{EY}$};
\node[name=UY,ell,  shape=ellipse] at (10,4) {$U_{Y}$};
\node[name=UE,ell,  shape=ellipse] at (10,-2) {$U_{E}$};
\begin{scope}[transparency group, opacity=0.3] 
    % \path [->] (A) edge (E1);
    \path [->,>={Stealth[black]}] (A) edge (Y1);
    \path [->,>={Stealth[black]}] (A) edge (Y2);
    % \path [->,>={Stealth[black]}] (A) edge[bend right] (E2);
    %\path [->,>={Stealth[black]}] (E1) edge (E2);
    % \path[->,>={Stealth[black]}]  (Q1) edge (E1);
    \path[->,>={Stealth[black]}]  (Q1.20) edge[bend right]  (Y1);
    \path[->,>={Stealth[black]}]  (Q2.20) edge[bend right]  (Y2);
    % \path[->,>={Stealth[black]}]  (Q2) edge  (E2);
    \path [->,>={Stealth[black]}] (Y1) edge (Y2);
    \path [->,>={Stealth[black]}] (Y1) edge (Q2);
    \path [->,>={Stealth[black]}] (L) edge (Q1);
    \path [->,>={Stealth[black]}] (L) edge (Q2);
    \path [->,>={Stealth[black]}] (L) edge (Y1);
    \path [->,>={Stealth[black]}] (L) edge[bend left=20] (Y2);
    \path[->,>={Stealth[black]}]  (UY) edge[black] (Y1);
    \path[->,>={Stealth[black]}]  (UY) edge[black] (Y2);
    % \path[->,>={Stealth[black]}]  (UE) edge[black] (E1);
    \path[->,>={Stealth[black]}]  (UE) edge[black] (Q2);
    \path[->,>={Stealth[black]}]  (UE) edge[black, bend left] (Q1.220);
\end{scope}
\begin{scope}[>={Stealth[black]},
              every edge/.style={draw=black,very thick}]
    
    % \path[->,>={Stealth[red]}]  (UAY) edge[red] (A);
    % \path[->,>={Stealth[red]}]  (UAY) edge[red] (Y1);
    % \path[->,>={Stealth[red]}]  (UEY1) edge[red] (E1.140);
    % \path[->,>={Stealth[red]}]  (UEY1) edge[red, bend right] (Y1);
    % \path[->,>={Stealth[red]}]  (UEY2) edge[red, bend left] (E2.215);
    % \path[->,>={Stealth[red]}]  (UEY2) edge[red] (Y2);
    % \path[->,>={Stealth[red]}]  (E1) edge[red] (Y2);
    % \path[->,>={Stealth[red]}]  (E1) edge[red] (E2);
    % \path[->,>={Stealth[blue]}]  (UY) edge[blue] (Y1);
    % \path[->,>={Stealth[blue]}]  (UY) edge[blue] (Y2);
    % \path[->,>={Stealth[blue]}]  (UE) edge[blue, bend left] (E1.215);
    % \path[->,>={Stealth[blue]}]  (UE) edge[blue] (E2);
    % \path[->,>={Stealth[red]}]  (A) edge[red] (E1);
\end{scope}
\end{tikzpicture}
 }
}

\caption{Interventions on exposure to the infectious agent in the modified trial $G$ during  interval 1 (A) and  interval 2 (B)}
    \label{fig:SWIGs_multiple_versions}
\end{figure}

In this section, our aim is to identify
\begin{align*}
    \text{VE}_1^\mathrm{challenge}(l,q)&=1-\frac{E_G[\Delta Y_1^{a=1,q_1=q}\mid L=l]}{E_G[\Delta Y_1^{a=0,q_1=q}\mid L=l]} ~, \\
    \text{VE}_2^\mathrm{challenge}(l,q)&=1-\frac{E_G[\Delta Y_2^{a=1,q_1=0,q_2=q}\mid L=l]}{E_G[\Delta Y_2^{a=0,q_1=0,q_2=q}\mid L=l]} ~,
\end{align*}
the challenge effect under infectious exposure $q$ in trial $G$ (Figure~\ref{fig:SWIGs_multiple_versions}), in terms of the observed distribution $P(A,L,\Delta Y_1,\Delta Y_2)$ in the original trial. To this end, we introduce the following assumptions.

\begin{assumption}[Consistency (multiple exposure versions)\label{ass:consistency_multiple_versions}]
    We assume that interventions on treatment $A$ and exposures $Q_1,Q_2$ are well-defined such that the following consistency conditions hold in trial $G$ for all $a\in\{0,1\}$ and $q_1,q_2\in\mathcal{Q}$:
    \begin{enumerate}[(i)]
        \item $\text{ if }    A=a \text{ then } Q_1=Q_1^a, Q_2^{q_1=0}= Q_2^{a,q_1=0}, \Delta Y_2^{q_1=0}=\Delta Y_2^{a,q_1=0}$~,
        \item $\text{ if }   A=a,Q_1=q_1  \text{ then } \Delta Y_1=\Delta Y_1^{a,q_1}$~,
        \item $\text{ if } A=a,Q_2^{q_1=0}=q_2  \text{ then } \Delta Y_2^{q_1=0}=\Delta Y_2^{a,q_1=0,q_2}$~.
    \end{enumerate}
\end{assumption}
\begin{assumption}[Treatment exchangeability (multiple exposure versions)\label{ass: treatment exch multiple versions}] For all $q_1,q_2\in\mathcal{Q}$,
    \begin{align*}
        \Delta Y_1^{a,q_1},\Delta Y_2^{a,q_1,q_2}\independent A\mid L ~.
    \end{align*}
\end{assumption}

\begin{assumption}[Exposure exchangeability (multiple exposure versions)\label{ass: exposure exchageability multiple versions}]\hfill

For all $a\in\{0,1\}$ and $q_1,q_2\in\mathcal{Q}$,
\begin{align*}
    \Delta Y_1^{a,q_1} \independent Q_1^a\mid A=a,L \text{ and } \Delta Y_2^{a,q_1=0,q_2} \independent Q_2^{a,q_1=0}\mid A=a,L ~.
\end{align*}

\end{assumption}

\begin{assumption}[No waning of placebo (multiple exposure versions)\label{ass: no waning placebo multiple versions}]
    For all $q\in\mathcal{Q}$,
    \begin{align*}
        E_G[\Delta Y_1^{a=0,q_1=q}\mid L] =E_G[\Delta Y_2^{a=0,q_{1}=0,q_2=q}\mid L] \text{ w.p.~1}~.
    \end{align*}
\end{assumption}

\begin{assumption}[Stationarity of exposure versions among the exposed\label{ass: stationarity environment}] For all $a\in\{0,1\}$ and $q\in\mathcal{Q}$,
    \begin{align*}
        P_G(Q_1\leq q\mid E_1=1,A=a,L)=P_G(Q_2^{q_1=0}\leq q\mid E_2^{q_1=0}=1,A=a,L) \text{ w.p.~1}~.
    \end{align*}
\end{assumption}
Assumption~\ref{ass: stationarity environment} is closely related to the ``Similar Study Environment'' assumption by \citet{fintzi_assessing_2021}:
\begin{quote}
    ``The proportional hazards model allows for the attack rate to change with time. But if the pathogen mutates to a form that is resistant to vaccine effects, efficacy may appear to wane. Another possibility is if human behavior changes in such a way that the vaccine is less effective. \textbf{For example, if there is less mask wearing in the community over the study, the viral inoculum at infection may increase over the study and overwhelm the immune response for later cases. Vaccines may work less well against larger inoculums and thus VE might appear to wane}. For viral mutation, analyses could be run separately for different major strains provided they occur both prior and post crossover.''
\end{quote}

Assumption~\ref{ass: stationarity environment} allows us to discern whether changes in cumulative incidences over time are due to changes in the exposure versions or changes in the challenge effect, as also alluded to in the quote from \citet{fintzi_assessing_2021}. In particular, this consideration also applies to conventional approaches to vaccine waning, i.e.\ the direct comparison of $\text{VE}_1^\mathrm{obs}$ vs. $\text{VE}_2^\mathrm{obs}$ as considered by e.g.\ \citet{fintzi_assessing_2021}.  Assumption~\ref{ass: stationarity environment} does \textit{not} require that the same number of individuals are exposed during intervals 1 and 2. In other words, the assumption is not necessarily violated if $P(E_1=1\mid L)<P(E_2=1\mid L)$, e.g.\ if the prevalence of infection increases in the larger population that embeds the trial participants. For instance, in an all-or-nothing model of vaccine protection \citep{halloran_design_2012}, immune individuals will never contract infection, regardless of the number of exposures during a given time interval, and those who are not immune will always contract infection if exposed. However, in a leaky model of vaccine protection, individuals with multiple exposures during interval $k$ will have a greater risk of the outcome than an individual with a single exposure. If $P(E_1=1\mid L)\ll P(E_2=1\mid L)$ w.p.~1, e.g. i \ due to a large difference in the lengths of intervals $k=1$ vs. $k=2$, or due to a large change in the infection prevalence, then multiple exposures are more likely during interval 2 compared to 1, which could violate Assumption~\ref{ass: stationarity environment}.

Violations of Assumption~\ref{ass: stationarity environment} can be mitigated by collecting certain data in a particular randomized experiment. Consider the sequential blinded crossover trial discussed by \citet{follmann_deferred-vaccination_2021}, where individuals are randomly assigned to vaccine versus control at time 1, and then cross over to the other treatment arm at time 2. By comparing the incidence \textit{during the same interval of calendar time} of those who received the vaccine at time 1 (early) versus time 2 (late), one can minimize differences in the study environment between early and late vaccinees.

To state the next assumption, we define the dose-response relations $\varphi_k(q)$ and an auxiliary function $h$.
\begin{align*}
    \varphi_1(q)&=E_G[\Delta Y_1^{a=1,q_1=q}\mid L=l] ~, \\
    \varphi_2(q)&=E_G[\Delta Y_2^{a=1,q_1=0,q_2=q}\mid L=l]~,\\
    h(q)&=\frac{\varphi_1(q)}{E_G[\varphi_1(Q_1)\mid E_1=1,A=1,L=l]}-\frac{\varphi_2(q)}{E_G[\varphi_2(Q_1)\mid E_1=1,A=1,L=l]}~.
\end{align*}
\begin{assumption}[Existence of a representative exposure (weak)\label{ass: weak existence of representative inoculum}] There exists a representative exposure $q_l^\ast\in\mathcal{Q}$ such that
\begin{align*}
    h(q_l^\ast)=0 ~.
\end{align*}
\end{assumption}

Since $E_G[h(Q_1)\mid E_1=1,A=1,L=l]=0$, Assumption~\ref{ass: weak existence of representative inoculum} is implied by the mean value theorem for a non-trivial class of dose-response relations $\varphi_k(q)$ and distributions of $Q_k$.

\begin{assumption}[Existence of a representative exposure (strong)\label{ass: strong existence of representative inoculum}] There exists a representative exposure $q_l^{\ast\ast}\in\mathcal{Q}$ such that
    \begin{align}
        E_G[\Delta Y_1^{a,q_1=q_l^{\ast\ast}}\mid L=l] &=   E_G[E_G[\Delta Y_1^{a,Q_1}\mid L=l]\mid  E_1=1,A=a, L=l] ~, \label{eq: strong I}\\
        E_G[\Delta Y_2^{a,q_1=0,q_2=q_l^{\ast\ast}}\mid L=l] &= E_G[E_G[\Delta Y_2^{a,q_1=0,Q_2^{q_1=0}}\mid L=l]\mid E_2^{q_1=0}=1,A=a,L=l]  ~, \label{eq: strong II}
    \end{align}
    for all $a\in\{0,1\}$.
\end{assumption}
Assumption~\ref{ass: strong existence of representative inoculum} can be regarded as a consistency assumption in conditional expectation, and states that an infectious inoculum of version $q_l^{\ast\ast}$ leads to the same outcome in conditional expectation as an average over a random version $Q_1$ among individuals with $E_1=1$, and likewise for a random version $Q_2^{q_1=0}$ among individuals with $E_2^{q_1=0}=1$. In other words, individuals with $Q_1>q_l^{\ast\ast}$ are perfectly balanced by individuals with $Q_1<q_l^{\ast\ast}$, for both treatment groups $A\in\{0,1\}$, and correspondingly for the second time interval. One particular scenario where this happens is if the conditional distribution of $Q_1,Q_2^{q_1=0}$ is narrow and centers around the particular version $q_l^{\ast\ast}$. 
This is equivalent to the statement that the dose-response relations $E_G[\Delta Y_1^{a,q}\mid L=l]$ and $E_G[\Delta Y_2^{a,q_1=0,q_2=q}\mid L=l]$, viewed as functions of $q$, do not vary in $q$ over the range of treatment versions with non-negligible probability, i.e.\ that all the observed exposure versions lead to the same outcomes in conditional expectation. In turn, this occurs in an all-or-nothing model of vaccine protection, but not necessarily in a leaky model \citep{halloran_design_2012}.

\subsection{Identification results}

\begin{proposition}\label{prp: randomized exposure version}
    Under Assumptions~\ref{ass:consistency 2 intervals}~(i) and (ii) for $e_1=0$ and Assumptions~\ref{ass: exposure positivity 2 intervals}-\ref{ass: exposure exchageability multiple versions},
    \begin{align}
       &E[\Delta Y_1^a\mid E_1^a=1,A=a, L=l]=E_G[E_G[\Delta Y_1^{a,Q_1}\mid L=l]\mid  E_1=1,A=a, L=l] \label{eq: dY1 version marginalization}~, \\
       &E[\Delta Y_2^{a,e_1=0}\mid E_2^{a,e_1=0}=1,A=a,L=l] \notag\\
       &\quad = E_G[E_G[\Delta Y_2^{a,q_1=0,Q_2^{q_1=0}}\mid L=l]\mid E_2^{q_1=0}=1,A=a,L=l] ~.\label{eq: dY2 version marginalization}
    \end{align}
    Under the additional Assumptions~\ref{ass: exposure necessity 2 intervals}-\ref{ass: exclusion 2 intervals} and Assumption~\ref{ass: treatment exchageability 2 intevals},
    \begin{align}
        1-\frac{E_G[E_G[\Delta Y_1^{a=1,Q_1}\mid L=l]\mid  E_1=1,A=1, L=l]}{E_G[E_G[\Delta Y_1^{a=0,Q_1}\mid L=l]\mid  E_1=1,A=0, L=l]}  &=\text{VE}_1^\mathrm{obs}(l) \label{eq: VE1 marginalization}~, \\
        1-\frac{E_G[E_G[\Delta Y_2^{a=1,q_1=0,Q_2^{q_1=0}}\mid L=l]\mid E_2^{q_1=0}=1,A=1,L=l]}{E_G[E_G[\Delta Y_2^{a=0,q_1=0,Q_2^{q_1=0}}\mid L=l]\mid E_2^{q_1=0}=1,A=0,L=l]} &\in[\mathcal{L}_2(l),\mathcal{U}_2(l)]~.  \label{eq: VE2 marginalization}
    \end{align}
\end{proposition}
\begin{proof}
Expression~\eqref{eq: dY1 version marginalization} follows from
    \begin{align*}
        &E[\Delta Y_1^a\mid E_1^a=1,A=a, L] \\
        \stackrel{\text{Assumption~\ref{ass:consistency 2 intervals}~(i)}}{=} &E[\Delta Y_1\mid E_1=1,A=a, L] \\
        \stackrel{\text{Assumption~\ref{ass:equivalence_modified_trial}}}{=} &E_G[\Delta Y_1\mid E_1=1,A=a, L] \\
        =&E_G[E_G[\Delta Y_1\mid Q_1,E_1=1,A=a, L]\mid  E_1=1,A=a, L]\\
        \stackrel{\eqref{eq: Ek definition}}{=}&E_G[E_G[\Delta Y_1\mid Q_1,A=a, L]\mid  E_1=1,A=a, L]\\
        \stackrel{\text{Assumption~\ref{ass:consistency_multiple_versions}}}{=}&E_G[E_G[\Delta Y_1^{a,Q_1^a}\mid Q_1^a,A=a, L]\mid  E_1=1,A=a, L]\\
        \stackrel{\text{Assumption~\ref{ass: exposure exchageability multiple versions}}}{=}& E_G[E_G[\Delta Y_1^{a,Q_1^a}\mid A=a, L]\mid  E_1=1,A=a, L]\\
        \stackrel{\text{Assumption~\ref{ass: treatment exch multiple versions}}}{=}& E_G[E_G[\Delta Y_1^{a,Q_1^a}\mid L]\mid  E_1=1,A=a, L] \\
        \stackrel{\text{Assumption~\ref{ass:consistency_multiple_versions}~(i)}}{=}&E_G[E_G[\Delta Y_1^{a,Q_1}\mid L]\mid  E_1=1,A=a, L] ~.
    \end{align*}
    The penultimate equality used the positivity condition $P_G(A=a\mid L)>0$ for all $a\in\{0,1\}$ w.p.~1, which follows from Assumption~\ref{ass: exposure positivity 2 intervals} and definition of the trial $G$ (Assumption~\ref{ass:equivalence_modified_trial}). Likewise, \eqref{eq: dY2 version marginalization} follows from
\begin{align*}
        &E[\Delta Y_2^{a,e_1=0}\mid E_2^{a,e_1=0}=1,A=a,L]\\
        \stackrel{\text{Assumption~\ref{ass:consistency 2 intervals}~(i)}}{=} &E[\Delta Y_2^{e_1=0}\mid E_2^{e_1=0}=1,A=a,L]\\
        \stackrel{\text{Assumption~\ref{ass:equivalence_modified_trial}}}{=} &E_G[\Delta Y_2^{q_1=0}\mid E_2^{q_1=0}=1,A=a,L] \\
        =&E_G[E_G[\Delta Y_2^{q_1=0}\mid Q_2^{q_1=0}, E_2^{q_1=0}=1,A=a,L]\mid E_2^{q_1=0}=1,A=a,L]\\
        \stackrel{\eqref{eq: Ek definition}}{=}&E_G[E_G[\Delta Y_2^{q_1=0}\mid Q_2^{q_1=0},A=a,L]\mid E_2^{q_1=0}=1,A=a,L]\\
        \stackrel{\text{Assumption~\ref{ass:consistency_multiple_versions}}}{=}&E_G[E_G[\Delta Y_2^{a,q_1=0,Q_2^{a,q_1=0}}\mid Q_2^{a,q_1=0},A=a,L]\mid E_2^{q_1=0}=1,A=a,L] \\
        \stackrel{\text{Assumption~\ref{ass: exposure exchageability multiple versions}}}{=}& E_G[E_G[\Delta Y_2^{a,q_1=0,Q_2^{a,q_1=0}}\mid A=a,L]\mid E_2^{q_1=0}=1,A=a,L] \\
        \stackrel{\text{Assumption~\ref{ass: treatment exch multiple versions}}}{=}& E_G[E_G[\Delta Y_2^{a,q_1=0,Q_2^{a,q_1=0}}\mid L]\mid E_2^{q_1=0}=1,A=a,L]\\
        \stackrel{\text{Assumption~\ref{ass:consistency_multiple_versions}~(i)}}{=}&E_G[E_G[\Delta Y_2^{a,q_1=0,Q_2^{q_1=0}}\mid L]\mid E_2^{q_1=0}=1,A=a,L]~.
    \end{align*}
    In the first and second lines, we have used that quantities under intervention $e_1=0$ are well-defined, which follows from Assumption~\ref{ass:consistency 2 intervals} (ii) for $e_1=0$.

    We proceed similarly to Section~\ref{sec: proof of theorem}.
    \begin{align}
        &\frac{E[\Delta Y_1^{a=1}\mid E_1^{a=1}=1,A=1,L]}{E[\Delta Y_1^{a=0}\mid E_1^{a=0}=1,A=0,L]} \notag\\
        =&\frac{E[I(E_1^{a=1}=1)\Delta Y_1^{a=1}\mid A=1,L]}{E[I(E_1^{a=0}=1)\Delta Y_1^{a=0}\mid A=0,L]}\times\frac{P(E_1^{a=0}=1\mid A=0,L)}{P(E_1^{a=1}=1\mid A=1,L)} \notag\\
        \stackrel{\text{Assumptions~\ref{ass: treatment exchageability 2 intevals}, \ref{ass: exposure positivity 2 intervals}}}{=}&\frac{E[I(E_1^{a=1}=1)\Delta Y_1^{a=1}\mid A=1,L]}{E[I(E_1^{a=0}=1)\Delta Y_1^{a=0}\mid A=0,L]}\times\frac{P(E_1^{a=0}=1\mid L)}{P(E_1^{a=1}=1\mid L)} \notag\\
        \stackrel{\text{Assumption~\ref{ass:no effect on exposure 2 intervals}}}{=}&\frac{E[I(E_1^{a=1}=1)\Delta Y_1^{a=1}\mid A=1,L]}{E[I(E_1^{a=0}=1)\Delta Y_1^{a=0}\mid A=0,L]} \notag\\
        \stackrel{\text{Assumption~\ref{ass: exposure necessity 2 intervals}}}{=}&\frac{E[\Delta Y_1^{a=1}\mid A=1,L]}{E[\Delta Y_1^{a=0}\mid A=0,L]} \notag\\
        \stackrel{\text{Assumptions~\ref{ass:consistency 2 intervals}~(i)}}{=}&\frac{E[\Delta Y_1\mid A=1,L]}{E[\Delta Y_1\mid A=0,L]} ~, \label{eq: dY1 ratio}
    \end{align}
    and
    \begin{align}
        &\frac{E[\Delta Y_2^{a=1,e_1=0}\mid E_2^{a=1,e_1=0}=1,A=1,L]}{E[\Delta Y_2^{a=0,e_1=0}\mid E_2^{a=0,e_1=0}=1,A=0,L]} \notag\\
        = &\frac{E[I(E_2^{a=1,e_1=0}=1)\Delta Y_2^{a=1,e_1=0}\mid A=1,L]}{E[I(E_2^{a=0,e_1=0}=1)\Delta Y_2^{a=0,e_1=0}\mid A=0,L]}\times\frac{P(E_2^{a=0,e_1=0}=1\mid A=0,L)}{P(E_2^{a=1,e_1=0}=1\mid A=1,L)} \notag\\
        \stackrel{\text{Assumptions~\ref{ass: treatment exchageability 2 intevals}, \ref{ass: exposure positivity 2 intervals}}}{=}&\frac{E[I(E_2^{a=1,e_1=0}=1)\Delta Y_2^{a=1,e_1=0}\mid A=1,L]}{E[I(E_2^{a=0,e_1=0}=1)\Delta Y_2^{a=0,e_1=0}\mid A=0,L]}\times\frac{P(E_2^{a=0,e_1=0}=1\mid L)}{P(E_2^{a=1,e_1=0}=1\mid L)} \notag\\
        \stackrel{\text{Assumption~\ref{ass:no effect on exposure 2 intervals}}}{=}& \frac{E[I(E_2^{a=1,e_1=0}=1)\Delta Y_2^{a=1,e_1=0}\mid A=1,L]}{E[I(E_2^{a=0,e_1=0}=1)\Delta Y_2^{a=0,e_1=0}\mid A=0,L]} \notag\\
        \stackrel{\text{Assumption~\ref{ass: exposure necessity 2 intervals}}}{=}& \frac{E[\Delta Y_2^{a=1,e_1=0}\mid A=1,L]}{E[\Delta Y_2^{a=0,e_1=0}\mid A=0,L]} \notag \\
        \stackrel{\text{Assumption~\ref{ass:consistency 2 intervals}~(i)}}{=}& \frac{E[\Delta Y_2^{e_1=0}\mid A=1,L]}{E[\Delta Y_2^{e_1=0}\mid A=0,L]} \label{eq: dY2 ratio ii} \\
        \stackrel{\text{Assumption~\ref{ass: exclusion 2 intervals}}}{\in}&\left[\frac{E[\Delta Y_2\mid A=1,L]}{E[\Delta Y_1 + \Delta Y_2\mid A=0,L]},\frac{E[\Delta Y_1+\Delta Y_2\mid A=1,L]}{E[ \Delta Y_2\mid A=0,L]} \right] ~. \label{eq: dY2 ratio}
    \end{align}
    Combining \eqref{eq: dY1 version marginalization} and \eqref{eq: dY1 ratio} gives \eqref{eq: VE1 marginalization}. Likewise, \eqref{eq: dY2 version marginalization} and \eqref{eq: dY2 ratio} imply \eqref{eq: VE2 marginalization}.
\end{proof}

Expressions~\eqref{eq: dY1 version marginalization}-\eqref{eq: dY2 version marginalization} state that the conditional exposure risk among the exposed is equal to a conditional mean of the dose-response relations $\varphi_1(Q_1)$ and $\varphi(Q_2^{q_1=0})$ over random exposure versions $Q_1$ and $Q_2^{q_1=0}$. Thus, \eqref{eq: VE1 marginalization}-\eqref{eq: VE2 marginalization} allows us to interpret \eqref{eq: VE_challenge_1 and VE_obs_1} and \eqref{eq:L_VE_challenge_2}-\eqref{eq:U_VE_challenge_2}  as identification formulas for a randomized exposure intervention, where an investigator draws a version $Q_1$ and $Q_2^{q_1=0}$ at random according to the conditional distribution functions $F_{Q_1\mid E_1=1,A=a,L=l}$ and $F_{Q_2^{q_1=0}\mid E_2^{q_1=0}=1,A=a,L=l}$. However, a contrast of \eqref{eq: VE1 marginalization}~vs.~\eqref{eq: VE2 marginalization} could be non-null due to changes in the distribution of exposure versions over time, unless Assumptions~\ref{ass: stationarity environment} holds.

Let $\mathrm{H}_0$ be the strict null hypothesis that the vaccine does not wane for any exposure version $q\in\mathcal{Q}$, 
    \begin{align*}
        \mathrm{H}_0:~E[\Delta Y_1^{a=1,q_1=q}\mid L] = E[\Delta Y_2^{a=1,q_1=0,q_2=q}\mid L] \text{ w.p.~1 for all } q\in\mathcal{Q} ~.
    \end{align*}
\begin{proposition}\label{prp: testing under the null}
    Under Assumptions~\ref{ass:consistency 2 intervals}~(i) and (ii) for $e_1=0$,  Assumptions~\ref{ass: exposure necessity 2 intervals}-\ref{ass: exclusion 2 intervals} and Assumptions~\ref{ass: treatment exchageability 2 intevals}-\ref{ass: stationarity environment},
    \begin{align}
        \mathrm{H}_0\implies \mathcal{L}_\psi(l) \leq 1 \leq \mathcal{U}_\psi(l) ~. \label{eq: hypothesis test}
    \end{align}
\end{proposition}

\begin{proof}
    Evaluating the ratio of \eqref{eq: dY1 version marginalization} and \eqref{eq: dY2 version marginalization} for $a=0$, and using Assumptions~\ref{ass: no waning placebo multiple versions} and \ref{ass: stationarity environment}, gives
    \begin{align}
        \frac{E[\Delta Y_1^{a=0}\mid E_1^{a=0}=1,A=0, L]}{E[\Delta Y_2^{a=0,e_1=0}\mid E_2^{a=0,e_1=0}=1,A=0,L]}=1~. \label{eq: ratio placebo step}
    \end{align}
    Similarly, by evaluating the ratio of \eqref{eq: dY1 version marginalization} and \eqref{eq: dY2 version marginalization} for $a=1$ under $\mathrm{H}_0$ and Assumption~\ref{ass: stationarity environment}, we find that
    \begin{align}
        \frac{E[\Delta Y_1^{a=1}\mid E_1^{a=1}=1,A=1, L]}{E[\Delta Y_2^{a=1,e_1=0}\mid E_2^{a=1,e_1=0}=1,A=1,L]}=1~. \label{eq: ratio vaccine step}
    \end{align}
    Taking the ratio of \eqref{eq: ratio placebo step} and \eqref{eq: ratio vaccine step} gives
    \begin{align}
        1=&\frac{\frac{E[\Delta Y_1^{a=1}\mid E_1^{a=1}=1,A=1,L]}{E[\Delta Y_1^{a=0}\mid E_1^{a=0}=1,A=0,L]}}{\frac{E[\Delta Y_2^{a=1,e_1=0}\mid E_2^{a=1,e_1=0}=1,A=1,L]}{E[\Delta Y_2^{a=0,e_1=0}\mid E_2^{a=0,e_1=0}=1,A=0,L]}}~. \label{eq: ratio=1}
    \end{align}
    Using \eqref{eq: dY1 ratio} and \eqref{eq: dY2 ratio} in \eqref{eq: ratio=1} gives the final result.
\end{proof}

By testing whether the observed data violates \eqref{eq: hypothesis test}, one can test the null hypothesis $\mathrm{H}_0$. Proposition~\ref{prp: testing under the null} clarifies that it is possible to test for the presence of vaccine waning even under \textit{arbitrary} distributions of versions of exposure using analogous assumptions to Theorem~\ref{thm: waning minimal ass 2 intevals}, as long as the distribution of exposure versions is stationary over time and the effect of placebo does not wane. A violation of $\mathrm{H}_0$ implies that there exists at least one infectious inoculum $q$ for which the vaccine wanes, but it does not establish for which inocula $q$ the vaccine wanes. However, it would be surprising if $\text{VE}_1^\mathrm{challenge}(l,q)>\text{VE}_2^\mathrm{challenge}(l,q)$ for some $q$ while $\text{VE}_1^\mathrm{challenge}(l,q)<\text{VE}_2^\mathrm{challenge}(l,q)$ for other versions $q$, and therefore a violation of $\mathrm{H}_0$ gives meaningful insight into vaccine waning, even though it does not tell us by how much the vaccine wanes for each exposure version $q$. Furthermore, the power to reject $\mathrm{H}_0$ is driven by exposure versions that appear frequently, or wane substantially, in the observed data, and therefore a rejection of $\mathrm{H}_0$ gives insight on waning of such exposure versions.

In the following propositions, we clarify conditions which allow us to interpret previous identification results for $\text{VE}_1^\mathrm{challenge}$ and $\text{VE}_2^\mathrm{challenge}$ in terms of controlled exposures to \textit{non-random}, representative infectious inocula.

Let $\psi(l,q)=E_G[\Delta Y_1^{a=1,q_1=q}\mid L=l]/E_G[\Delta Y_2^{a=1,q_1=0,q_2=q}\mid L=l]$. 
\begin{proposition}\label{prp: identification weak existence}
    Under Assumptions~\ref{ass:consistency 2 intervals}~(i) and (ii) for $e_1=0$, Assumptions~\ref{ass: exposure necessity 2 intervals}-\ref{ass: exclusion 2 intervals} and Assumptions~\ref{ass: treatment exchageability 2 intevals}-\ref{ass: weak existence of representative inoculum}, $\mathcal{L}_\psi(l) \leq \psi(l,q^{\ast})\leq \mathcal{U}_\psi(l)$.
\end{proposition}
\begin{proof}
        Multiplying both sides of \eqref{eq: dY1 ratio} by $E_G[\Delta Y_1^{a=1,q_1=q}\mid L]/E_G[\Delta Y_1^{a=0,q_1=q}\mid L]$ and both sides  of \eqref{eq: dY2 ratio ii} by $E_G[\Delta Y_2^{a=1,q_1=0,q_2=q}\mid L]/E_G[\Delta Y_2^{a=0,q_1=0,q_2=q}\mid L]$ gives
        \begin{align}
            \scriptstyle \frac{P_G(\Delta Y_1^{a=1,q_1=q}=1\mid L)}{P_G(\Delta Y_1^{a=0,q_1=q}=1\mid L)} &= \scriptstyle \frac{P(\Delta Y_1=1\mid A=1,L)}{P(\Delta Y_1=1\mid A=0,L)}\times\frac{\frac{P_G(\Delta Y_1^{a=1,q_1=q}=1\mid L)}{P(\Delta Y_1^{a=1}=1\mid E_1^{a=1}=1,A=1,L)}}{\frac{P_G(\Delta Y_1^{a=0,q_1=q}=1\mid L)}{P(\Delta Y_1^{a=0}=1\mid E_1^{a=0}=1,A=0,L)}} ~, \label{eq: VE1q without consistency}\\
            \scriptstyle \frac{P_G(\Delta Y_2^{a=1,q_1=0,q_2=q}=1\mid L)}{P_G(\Delta Y_2^{a=0,q_1=0,q_2=q}=1\mid L)} &= \scriptstyle \frac{P(\Delta Y_2^{e_1=0}=1\mid A=1,L)}{P(\Delta Y_2^{e_1=0}=1\mid A=0,L)}\times\frac{\frac{P_G(\Delta Y_2^{a=1,q_1=0,q_2=q}=1\mid L)}{P(\Delta Y_2^{a=1,e_1=0}=1\mid E_2^{a=1,e_1=0}=1,A=1,L)}}{\frac{P_G(\Delta Y_2^{a=0,q_1=0,q_2=q}=1\mid L)}{P(\Delta Y_2^{a=0,e_1=0}=1\mid E_2^{a=0,e_1=0}=1,A=0,L)}} ~. \label{eq: VE2q without consistency}
        \end{align}
    Assumptions~\ref{ass: stationarity environment} and \ref{ass: weak existence of representative inoculum} together imply that
    \begin{align}
        \scriptstyle \frac{P_G(\Delta Y_1^{a=1,q_1=q_l^\ast}=1\mid L)}{E_G[E_G[\Delta Y_1^{a=1,Q_1}\mid L]\mid  E_1=1,A=1, L]}= \scriptstyle \frac{P_G(\Delta Y_2^{a=1,q_1=0,q_2=q_l^\ast}=1\mid L)}{E_G[E_G[\Delta Y_2^{a=1,q_1=0,Q_2^{q_1=0}}\mid L]\mid E_2^{q_1=0}=1,A=1,L]}~. \label{eq: simplified VE12 ratio factor}
    \end{align}
    Next, Assumptions~\ref{ass: no waning placebo multiple versions} and \ref{ass: stationarity environment} imply that
    \begin{align}
        \scriptstyle \frac{P_G(\Delta Y_1^{a=0,q_1=q_l^\ast}=1\mid L)}{E_G[E_G[\Delta Y_1^{a=0,Q_1}\mid L]\mid  E_1=1,A=0, L]}= \scriptstyle \frac{P_G(\Delta Y_2^{a=0,q_1=0,q_2=q_l^\ast}=1\mid L)}{E_G[E_G[\Delta Y_2^{a=0,q_1=0,Q_2^{q_1=0}}\mid L]\mid E_2^{q_1=0}=1,A=0,L]}~. \label{eq: simplified VE12 ratio factor A=0}
    \end{align}
    Taking the ratio of \eqref{eq: simplified VE12 ratio factor} and \eqref{eq: simplified VE12 ratio factor A=0}, and using \eqref{eq: dY1 version marginalization}-\eqref{eq: dY2 version marginalization} gives
    \begin{align}
        \frac{\frac{P_G(\Delta Y_1^{a=1,q_1=q_l^\ast}=1\mid L)}{P(\Delta Y_1^{a=1}=1\mid E_1^{a=1}=1,A=1,L)}}{\frac{P_G(\Delta Y_1^{a=0,q_1=q_l^\ast}=1\mid L)}{P(\Delta Y_1^{a=0}=1\mid E_1^{a=0}=1,A=0,L)}}=\frac{\frac{P_G(\Delta Y_2^{a=1,q_1=0,q_2=q_l^\ast}=1\mid L)}{P(\Delta Y_2^{a=1,e_1=0}=1\mid E_2^{a=1,e_1=0}=1,A=1,L)}}{\frac{P_G(\Delta Y_2^{a=0,q_1=0,q_2=q_l^\ast}=1\mid L)}{P(\Delta Y_2^{a=0,e_1=0}=1\mid E_2^{a=0,e_1=0}=1,A=0,L)}} ~. \label{eq: VE12 ratio factor}
    \end{align}
    Finally, taking the ratio of \eqref{eq: VE1q without consistency} and \eqref{eq: VE2q without consistency}, and using Assumption~\ref{ass: exclusion 2 intervals} to bound $P(\Delta Y_2^{e_1=0}=1\mid A=1,L)/P(\Delta Y_2^{e_1=0}=1\mid A=0,L)$ and \eqref{eq: VE12 ratio factor} to cancel the remaining unidentified fractions for $q=q_l^\ast$ gives the final result.
\end{proof}
Proposition~\ref{prp: identification weak existence} clarifies that $\mathcal{L}_\psi(l),\mathcal{U}_\psi(l)$ can be interpreted as bounds on a ratio of challenge effects in intervals 1 and 2 for a representative exposure $q_l^\ast$ under analogous assumptions used to identify $\psi(l)$ in Section~\ref{sec:identification}, even when multiple exposure versions are present, as long as the distribution of exposure versions is constant across intervals 1 and 2 and satisfies Assumption~\ref{ass: weak existence of representative inoculum}.

\begin{proposition}\label{prp: identification strong existence}
Assumptions~\ref{ass:consistency 2 intervals}~(i) and (ii) for $e_1=0$, Assumptions~\ref{ass: exposure necessity 2 intervals}-\ref{ass: exclusion 2 intervals} and Assumptions~\ref{ass: treatment exchageability 2 intevals}-\ref{ass: exposure exchageability multiple versions} combined with  Assumption~\ref{ass: strong existence of representative inoculum}, imply that
    \begin{align}
        &\text{VE}_1^\mathrm{challenge}(l,q_l^{\ast\ast}) =\frac{E[\Delta Y_1\mid A=1,L=l]}{E[\Delta Y_1\mid A=0,L=l]} ~, \label{eq: strong VE1} \\
        \mathcal{L}_2(l) \leq  &\text{VE}_2^\mathrm{challenge}(l,q_l^{\ast\ast}) \leq \mathcal{U}_2(l) \label{eq: strong VE2}~.
    \end{align}
\end{proposition}
\begin{proof}
By \eqref{eq: dY1 version marginalization}-\eqref{eq: dY2 version marginalization}, Assumption~\ref{ass: strong existence of representative inoculum} implies that
\begin{align}
    P_G(\Delta Y_1^{a,q_1=q_l^{\ast\ast}}=1\mid L)&=P(\Delta Y_1^{a}=1\mid E_1^{a}=1,A=a,L)~, \label{eq: consistency in expectation 1}\\
   P_G(\Delta Y_2^{a,q_1=0,q_2=q_l^{\ast\ast}}=1\mid L)&=P(\Delta Y_2^{a,e_1=0}=1\mid E_2^{a,e_1=0}=1,A=a,L)  \label{eq: consistency in expectation 2}
\end{align}
for all $a\in\{0,1\}$. We obtain \eqref{eq: strong VE1}  from using  \eqref{eq: consistency in expectation 1} in \eqref{eq: VE1q without consistency} and likewise we obtain \eqref{eq: strong VE2}  from using  \eqref{eq: consistency in expectation 2} and Assumption~\ref{ass: exclusion 2 intervals} in \eqref{eq: VE2q without consistency}.

\end{proof}

Similarly to Proposition~\ref{prp: testing under the null}, the identification results in Propositions~\ref{prp: identification weak existence}-\ref{prp: identification strong existence}  do not tell us exactly for which exposure versions $q_l^\ast,q_l^{\ast\ast}$ we (potentially) identify vaccine waning.

\section{Extension to multiple time intervals and loss to follow-up\label{app: extension to multiple times}}
We assume that Definition~1 in \citet{thomas_s_richardson_single_2013} holds, which implies that interventions on  $A,E_k,C_k$ for $k\in\{1,\dots,K\}$ are well-defined. 
We use an underbar to denote future variables, e.g.\ $\underline{Y}_k=(Y_k,\dots, Y_K)$, and an overbar to denote the history of a random variable through time $k$, e.g.\ $\overline{E}_k=(E_1,\dots,E_k)$.  Under an additional intervention to prevent losses to follow-up, denoted by $\overline{c}=0$, the challenge effect at time $k$ is defined as
\begin{align*}
    \text{VE}_k^\mathrm{challenge}(l) = 1-\frac{E[\Delta Y_k^{a=1,\overline{e}_{k-1}=0,e_k=1,\overline{c}=0}\mid L=l]}{E[\Delta Y_k^{a=0,\overline{e}_{k-1}=0,e_k=1,\overline{c}=0}\mid L=l]} ~.
\end{align*}
Suppose that the following assumptions hold for all $a,k,\overline{e}_k,l$.
\begin{assumption}[Exposure necessity ($K$ intervals)\label{ass: exposure necessity K intervals}]
 \begin{align*}
     E_k^{a,\overline{e}_{k-1}=0,\overline{c}=0}=0\implies \Delta Y_k^{a,\overline{e}_{k-1}=0,\overline{c}=0}=0  ~.
 \end{align*}
\end{assumption}

\begin{assumption}[No treatment effect on exposure in the unexposed ($K$ intervals)\label{ass:no effect on exposure K intervals}]
\begin{align*}
    E_k^{a=0,\overline{e}_{k-1}=0,\overline{c}=0}=E_k^{a=1,\overline{e}_{k-1}=0,\overline{c}=0}  ~.
\end{align*}
\end{assumption}

\begin{assumption}[Exposure effect restriction ($K$ intervals)\label{ass: exclusion K intervals}]
\begin{align*}
    E[\Delta Y_k^{\overline{c}=0}\mid A=a, L] \leq E[\Delta Y_k^{\overline{e}_{k-1}=0,\overline{c}=0}\mid A=a,L] \leq E[Y^{\overline{c}=0}_k\mid A=a, L] \text{ w.p. 1}~.
\end{align*}
\end{assumption}

\begin{assumption}[Exposure exchangeability ($K$ intervals)\label{ass: exposure exchageability K intevals}]\hfill
\begin{align*}
    \Delta Y_k^{a,\overline{e}_{k-1}=0,e_k,\overline{c}=0}\independent E_k^{a,\overline{e}_{k-1}=0,\overline{c}=0}\mid A=a,L  ~.
\end{align*}
\end{assumption}

\begin{assumption}[Treatment exchangeability ($K$ intervals)\label{ass: treatment exchageability K intevals}]\hfill
\begin{align*}
     E_k^{a,\overline{e}_{k-1},\overline{c}=0}, \Delta Y_k^{a,\overline{e}_{k-1},e_k,\overline{c}=0}   \independent A \mid L ~.
\end{align*}
\end{assumption}

\begin{assumption}[Exchangeability for loss to follow-up ($K$ intervals)\label{ass: exchangeability censoring K intervals}]
\begin{align}
     \underline{Y}_{k+1}^{a,\overline{c}=0} \independent C_{k+1}^{a,\overline{c}=0}\mid C_k^{a,\overline{c}=0},Y_{k}^{a,\overline{c}=0},A=a,L ~. \label{eq: exchangeability censoring K intervals} 
\end{align}    
\end{assumption}

\begin{assumption}[Positivity for loss to follow-up ($K$ intervals\label{ass: positivity censoring K intervals})]
    \begin{align}
    &f_{Y_k,C_k,A,L}(0,0,a,l) >0\notag\\
    &\quad\implies P(C_{k+1}=0\mid Y_k=0,C_k=0,A=a,L=l) >0 \text{ for all } l~. \label{eq: positivity censoring K intervals}
\end{align}
\end{assumption}

Let $\Lambda_{k,a,l}=P(Y_k=1\mid Y_{k-1}=0,C_k=0,A=a,L=l)$ be a discrete time hazard of the outcome. 
\begin{assumption} [Rare events ($K$ intervals)\label{ass: rare events K intervals}]
   \begin{align}
       \sum_{k=1}^K \Lambda_{k,a,l} \ll 1 \text{ for all } a,l ~.
   \end{align}
\end{assumption}

\begin{theorem}[Bounds for $K$ intervals\label{thm: bounds for K intervals}]
    Under Assumption~\ref{ass: exposure positivity 2 intervals} and Assumptions~\ref{ass: exposure necessity K intervals}-\ref{ass: treatment exchageability K intevals},
    \begin{align*}
        \mathcal{L}_k(l)\leq \text{VE}_k^\mathrm{challenge}(l)\leq \mathcal{U}_k(l) \text{ for $k\in\{2,\dots,K\}$}~,
    \end{align*}
    where
    \begin{align*}
        \mathcal{L}_k(l)&=1-\frac{E[Y^{\overline{c}=0}_k\mid A=1,L=l]}{E[\Delta Y^{\overline{c}=0}_k\mid A=0,L=l]} ~,\\
        \mathcal{U}_k(l)&=1-\frac{E[\Delta Y^{\overline{c}=0}_k\mid A=1,L=l]}{E[ Y^{\overline{c}=0}_k\mid A=0,L=l]} ~,\\
    \end{align*}
    whenever $E[\Delta Y^{\overline{c}=0}_k\mid A=a,L=l]>0$ for all $a\in\{0,1\}$.
\end{theorem}
\begin{proof}
    \begin{align*}
        &P(\Delta Y_k^{\overline{e}_{k-1}=0,\overline{c}=0}=1\mid A=a,L) \\
        =&P(\Delta Y_k^{a,\overline{e}_{k-1}=0,\overline{c}=0}=1\mid A=a,L) \\
        \stackrel{\text{Assumption~\ref{ass: exposure necessity K intervals}}}{=} &P(E_k^{a,\overline{e}_{k-1}=0,\overline{c}=0}=1,\Delta Y_k^{a,\overline{e}_{k-1}=0,\overline{c}=0}=1\mid A=a,L) \\
        =&P(E_k^{a,\overline{e}_{k-1}=0,\overline{c}=0}=1,\Delta Y_k^{a,\overline{e}_{k-1}=0,E_{k}^{a,\overline{e}_{k-1}=0,\overline{c}=0},\overline{c}=0}=1\mid A=a,L) \\
        =&P(E_k^{a,\overline{e}_{k-1}=0,\overline{c}=0}=1,\Delta Y_k^{a,\overline{e}_{k-1}=0,e_k=1,\overline{c}=0}=1\mid A=a,L) \\
        \stackrel{\text{Assumption~\ref{ass: exposure exchageability K intevals}}}{=}&P(\Delta Y_k^{a,\overline{e}_{k-1}=0,e_k=1,\overline{c}=0}=1\mid A=a,L) P(E_k^{a,\overline{e}_{k-1}=0,\overline{c}=0}=1\mid A=a,L) \\
        \stackrel{\text{Assumptions~\ref{ass: exposure positivity 2 intervals}, \ref{ass: treatment exchageability K intevals}}}{=}&P(\Delta Y_k^{a,\overline{e}_{k-1}=0,e_k=1,\overline{c}=0}=1\mid L) P(E_k^{a,\overline{e}_{k-1}=0,\overline{c}=0}=1\mid L)~.
    \end{align*}
    Therefore, by Assumption~\ref{ass: exclusion K intervals}
    \begin{align*}
        &E[\Delta Y^{\overline{c}=0}_k\mid A=a,L] \\
        &\quad \leq E[\Delta Y_k^{a,\overline{e}_{k-1}=0,e_k=1,\overline{c}=0}\mid L]P(E_k^{a,\overline{e}_{k-1}=0,\overline{c}=0}=1\mid L) \\
        &\qquad \leq E[Y^{\overline{c}=0}_k\mid A=a,L] ~.
    \end{align*}
    Taking the ratio for $a=1$ vs. $a=0$, and using Assumption~\ref{ass:no effect on exposure K intervals} to cancel the resulting quotient of exposure probabilities gives
    \begin{align}
        \frac{E[\Delta Y^{\overline{c}=0}_k\mid A=1,L]}{E[Y^{\overline{c}=0}_k\mid A=0,L]} \leq \frac{E[\Delta Y_k^{a=1,\overline{e}_{k-1}=0,e_k=1,\overline{c}=0}\mid L]}{E[\Delta Y_k^{a=0,\overline{e}_{k-1}=0,e_k=1,\overline{c}=0}\mid L]} \leq \frac{E[Y^{\overline{c}=0}_k\mid A=1,L]}{E[\Delta Y^{\overline{c}=0}_k\mid A=0,L]} ~. \label{eq: generalized ID K and censoring}
    \end{align}
    
\end{proof}

The lower and upper limits of \eqref{eq: generalized ID K and censoring} are straightforward to identify and estimate using techniques from survival analysis, as described in Section~\ref{sec:maintext estimation}. Furthermore, under Assumptions~\ref{ass: exchangeability censoring K intervals} and \ref{ass: positivity censoring K intervals}, we can identify the lower and upper limits of Theorem~\ref{thm: bounds for K intervals} using the fact that
\begin{align}
    E[Y_k^{\overline{c}=0}\mid A=a,L] = \sum_{r=1}^k\prod_{v=1}^{r-1} (1-\Lambda_{v,a,l})\cdot \Lambda_{r,a,L}~,  \label{eq: identification cum inc K intervals}
\end{align}
see e.g.\ Appendix~C in \citet{janvin_causal_2023}. Under Assumption~\ref{ass: rare events K intervals}, the product in \eqref{eq: identification cum inc K intervals} simplifies and we obtain the approximate bounds
\begin{align}
    \mathcal{L}_k(l) = 1-\frac{\sum_{k^\prime=1}^{k} \Lambda_{k^\prime,a=1,l}}{\Lambda_{k,a=0,l}} \quad\text{ and }\quad \mathcal{U}_k(l) = 1-\frac{\Lambda_{k,a=1,l}}{\sum_{k^\prime=1}^{k} \Lambda_{k^\prime,a=0,l}} ~. \label{eq: bounds K intervals rare event}
\end{align}

\section{Logistic regression with individual-level data\label{app: estimation individual level data}}

Suppose we have access to individual baseline variables $L_i$, treatment $A_i$, loss to follow-up (censoring) indicator $C_{i,k}$ and outcome $\Delta Y_{i,k}$ from $k\in\{1,\dots,K\}$ time intervals for individuals $i\in\{1,\dots,n\}$. As we assume individuals in the sample are i.i.d., we suppress the subscript $i$.

Let $f_k(a,l;\beta_k)=E[Y_k\mid Y_{k-1}=0,C_k=0,A=a,L=l;\beta_k]$ be a parametric model for $\Lambda_{k,a,l}$ for each time interval $k\in\{1,\dots,K\}$, e.g.\ a logistic regression model. Suppose that the number of time intervals $K$ is fixed, and that the parameter $\beta_k=(\beta_{1,k},\dots, \beta_{d,k})^T$ has a fixed dimension $d$. Denote the maximum likelihood estimator of $\beta_k$ by $\widehat\beta_k$ and let $\widehat\Lambda_{k,a,l;\widehat\beta_k} =f_k(a,l;\widehat\beta_k)$ be a prediction of $\Lambda_{k,a,l}$ using the estimated coefficients $\widehat\beta_k$. We can then consistently estimate the bounds $\mathcal{L}_k(l)$ and $\mathcal{U}_k(l)$ using plugin estimators
\begin{align}
    \widehat{ \mathcal{L}}_k(l) = 1-\frac{\sum_{k^\prime=1}^{k} \widehat\Lambda_{k^\prime,a=1,l;\widehat\beta_k}}{\widehat\Lambda_{k,a=0,l;\widehat\beta_k}} \text{ and } \widehat{\mathcal{U}}_k(l) = 1-\frac{\widehat\Lambda_{k,a=1,l;\widehat\beta_k}}{\sum_{k^\prime=1}^{k} \widehat\Lambda_{k^\prime,a=0,l;\widehat\beta_k}} ~. \label{eq: plugin estimators}
\end{align}

\section{Summary data\label{app:summary data}}
\subsection{Identification\label{app: loss to follow-up}}
 
We define each of intervals $k=1$ and $k=2$  by combining several subintervals, summarized in Table~\ref{tab:intervals}, using publicly available summary data from Figure~2 in \citet{thomas_safety_2021}. Let, $j=1,\dots,j_k$ denotes subinterval $j$ of interval $k$, and let $s$ index a short time interval of duration $\Delta s=1$ day, such that $s_{k,j}^-,s_{k,j}^+$ denote the first and last days of subinterval $(k,j)$ respectively. 
Let $\tau_{k,j}$ denote the duration (in days) of subinterval $(k,j)$. Next, let $T_{k,j,a}$ and $N_{k,j,a}$ denote respectively the total person time at risk (in days) and the number of recorded cases of infection during subinterval $(k,j)$ of treatment group $A=a$. All quantities introduced in this paragraph are evaluated in a subset $L=l$ of baseline covariates, even though they are not indexed by $l$ to reduce clutter.

Assume that Definition~1 in \citet{thomas_s_richardson_single_2013} holds, which implies that interventions on loss to follow-up $C_s$ are well-defined at all times $s$. We denote the discrete-time hazard of $Y_s$ by $\lambda_{s,a,l}=P(Y_s=1 \mid C_s=0,Y_{s-1}=0, A=a,L=l)/\Delta s$, and assume the following. 
\begin{assumption}[Constant subinterval hazard\label{ass: constant hazard}]
    Within each subinterval $(k,j)$ of every stratum stratum $\{A=a,L=l\}$, the hazard $\lambda_{s,a,l}$ is a constant function of time $s$, denoted by $\lambda_{s,a,l}=\lambda_{k,j,a}$ for all $s\in\{s_{k,j}^-,\dots,s_{k,j}^+\}$.
\end{assumption}

Importantly, we do not assume a constant value of the hazard for different subintervals $(k,j)$. Assumption~\ref{ass: constant hazard} is plausible for short time intervals, such as subintervals $(k,j)$ in Table~\ref{tab:intervals}. Furthermore, Assumption~\ref{ass: constant hazard} can be falsified by inspecting whether the cumulative incidence curves, such as Figure~2 of \citet{thomas_safety_2021}, deviate from the piecewise exponential form implied by the assumption.

To identify $\lambda_{k,j,a}$ in the presence of censoring, we will invoke standard exchangeability and positivity assumptions for loss to follow-up at all times $s$ and for all treatments $a$. 
\begin{assumption}[Exchangeability for loss to follow-up (subinterval)\label{ass: exchangeability censoring}]
   \begin{align*}
       \underline{Y}_{s+1}^{a,\overline{c}=0}\independent C_{s+1}^{a,\overline{c}=0} \mid Y_s^{a,\overline{c}=0},C_s^{a,\overline{c}=0},A=a,L ~.
   \end{align*}
\end{assumption}
Assumption~\ref{ass: exchangeability censoring} precludes the existence of open backdoor paths (i.e.\ confounding) between loss to follow-up and the outcome.

\begin{assumption}[Positivity for loss to follow-up (subinterval)\label{ass: positivity censoring}]
    \begin{align*}
       &f_{Y_s,C_s,A,L}(0,0,a,l) >0 \\
       &\quad\implies P(C_{s+1}=0\mid Y_s=0,C_s=0,A=a,L=l) >0 \text{ for all } l~.
    \end{align*}
\end{assumption}
Assumption~\ref{ass: positivity censoring} states that for any possible combination of treatment assignment and baseline covariates among those who are event free and uncensored in interval $s$, some individuals will remain uncensored during the next interval $s+1$.

\begin{table}[htbp]
  \centering
  \caption{Definition of time indices}
    \resizebox{\linewidth}{!}{
    \begin{tabular}{llllll}
    \toprule
     Interval description & $k$ &  $j$ & $s_{k,j}^-$ & $s_{k,j}^+$ & $\tau_{k,j}$\\
    \midrule
    $\geq 11$ days after dose 1 until dose 2  & 1 & 1 & 12 & 21 & 10 \\
    After dose 2 until $<7$ days after & 1 & 2 & 22 & 28 & 7 \\
    $\geq 7$ days after dose 2 until $<2$ months after &1 & 3  & 29 & 82 & 54 \\
    $\geq 2$ months after dose 2 until $<4$ months after dose 2 &2 & 1   & 83 & 143 & 61 \\
    $\geq 4$ months after dose 2 & 2 & 2 & 144 & 190 & 47 
    \\\hline
    \end{tabular}%
    }
  \label{tab:intervals}%
\end{table}%

An endpoint at day 190 has been chosen in the final row of Table~\ref{tab:intervals}  to ensure that there are still individuals at risk on the final day, i.e.\ that Assumption~\ref{ass: positivity censoring} holds. This is guaranteed since there are recorded events in either treatment group after day 190 in the cumulative incidence plots shown in Figure~2 of \citet{thomas_safety_2021}. By Assumption~\ref{ass: constant hazard}, the hazard $\lambda_{s,a,L}$ is constant during the final subinterval in Table~\ref{tab:intervals}, and we may therefore consider an endpoint on day 190 without introducing any error into the hazard estimate $\widehat\lambda_{k=1,j=2}$, even though $\widehat\lambda_{k=1,j=2}$ may use observations after day 190.

\begin{lemma}\label{lemma: product identification}
    Under Assumptions~\ref{ass: constant hazard}-\ref{ass: positivity censoring}, and for all $k,j,a$, the cumulative incidence of the outcome during subinterval $k$ can be expressed  as
\begin{align}
    E[\Delta Y_1^{\overline{c}=0}\mid A=a,L] &= \sum_{s=1}^{s_{1,j_1}^+}\prod_{s^\prime =1}^{s-1} (1-\lambda_{s^\prime,a,L}\Delta s)\lambda_{s,a,L}\Delta s ~, \label{eq: product lemma i} \\
    E[\Delta Y_2^{\overline{c}=0}\mid A=a,L] &= \sum_{s=s_{2,1}^-}^{s_{2,j_2}^+} \prod_{s^\prime =1}^{s-1} (1-\lambda_{s^\prime,a,L}\Delta s)\lambda_{s,a,L}\Delta s ~, \label{eq: product lemma ii}
\end{align}
and $\lambda_{s,a,L}=\sum_{k,j}I(s_{k,j}^-\leq s \leq s_{k,j}^+)\lambda_{k,j,a}$ is a piece-wise constant hazard identified by
\begin{align*}
    \lambda_{k,j,a}= \frac{E[N_{k,j,a}]}{E[T_{k,j,a}]} ~.
\end{align*}
\end{lemma}
\begin{proof}
    For a derivation of \eqref{eq: product lemma i}-\eqref{eq: product lemma ii}, see e.g.\ Appendix C of \citet{janvin_causal_2023}.  Next,
    \begin{align*}
        &\lambda_{s,a,L}\cdot \Delta s \\
        =&E[\Delta Y_s  \mid Y_{s-1}=0, C_{s}=0, A=a,L] \\
        =&\frac{E[I(C_{s}=Y_{s-1}=0)\Delta Y_s  \mid  A=a,L]}{E[I(C_{s}=Y_{s-1}=0) \mid  A=a,L]} \\
         =&\frac{E[\Delta Y_s\mid  A=a,L]}{E[I(C_{s}=Y_{s-1}=0) \mid  A=a,L]}~.
    \end{align*}
    Hence, by Assumption~\ref{ass: constant hazard},
    \begin{align*}
        &\lambda_{k,j,a}  \\
        =& \frac{1}{\Delta s/(\lambda_{k,j,a} \cdot \Delta s)} \cdot \frac{\sum_{s=s_{k,j}^-}^{s_{k,j}^+} E[\Delta Y_s  \mid  A=a,L]}{\sum_{s^\prime=s_{k,j}^-}^{s_{k,j}^+} E[\Delta Y_{s^\prime} \mid  A=a,L]} \\
        =&\frac{\sum_{s=s_{k,j}^-}^{s_{k,j}^+} E[\Delta Y_s  \mid  A=a,L]}{\Delta s\sum_{s^\prime=s_{k,j}^-}^{s_{k,j}^+} E[\Delta Y_{s^\prime} \mid  A=a,L]/(\lambda_{k,j,a} \cdot \Delta s)} \\
        =&\frac{E[\sum_{s=s_{k,j}^-}^{s_{k,j}^+}\Delta Y_s  \mid  A=a,L]}{ E[\sum_{s=s_{k,j}^-}^{s_{k,j}^+}\Delta s \cdot I(C_{s}=Y_{s-1}=0) \mid  A=a,L]} ~.
    \end{align*}
    The numerator and denominator are equal to the expected number of events and expected person time at risk per individual in stratum $\{A=a,L\}$ during subinterval $(k,j)$. Therefore, 
    \begin{align*}
        \lambda_{k,j,a} = \frac{E[N_{k,j,a}/n]}{E[T_{k,j,a}/n]}=\frac{E[N_{k,j,a}]}{E[T_{k,j,a}]} ~.
    \end{align*}
\end{proof}
Under Assumptions~\ref{ass: rare events K intervals} and \ref{ass: constant hazard}, \eqref{eq: product lemma i}-\eqref{eq: product lemma ii} simplify approximately to $E[\Delta Y_k^{\overline{c}=0}\mid A=a,L]=\Lambda_{k,a}$, where 
\begin{align}
    \Lambda_{k,a}=\sum_{j=1}^{j_k}\lambda_{k,j,a}\tau_{k,j} \label{eq: Lambda_ka definition}
\end{align}

is the cumulative hazard in interval $k$. Thus, under Assumption~\ref{ass: exposure positivity 2 intervals}, Assumptions~\ref{ass: exposure necessity K intervals}-\ref{ass: treatment exchageability K intevals} and Assumptions~\ref{ass: rare events K intervals}-\ref{ass: positivity censoring}, Expression \eqref{eq: generalized ID K and censoring} gives $\mathcal{L}_2(l)\leq \text{VE}_2^\mathrm{challenge}(l)\leq \mathcal{U}_2(l)$, with the approximate bounds
\begin{align}
    \mathcal{L}_2(l)=1-\frac{ \Lambda_{k=1,a=1}+\Lambda_{k=2,a=1}}{\Lambda_{k=2,a=0}} \text{ and } \mathcal{U}_2(l)= 1-\frac{\Lambda_{k=2,a=1}}{\Lambda_{k=1,a=0}+\Lambda_{k=2,a=0}} ~. \label{eq: bounds individual level data}
\end{align}

\subsection{Estimation\label{app:estimation}} In this section, we consider a setting where Assumptions~\ref{ass: rare events K intervals} and \ref{ass: constant hazard} hold. 
We first define the estimator
\begin{align*}
    \widehat \lambda_{k,j,a}=\frac{N_{k,j,a}}{T_{k,j,a}} ~.
\end{align*}
Under Assumption~\ref{ass: rare events K intervals}, an asymptotic variance estimator of $\widehat\lambda_{k,j,a}$ is 
\begin{align}
    \widetilde\var\widehat\lambda_{k,j,a}= \frac{\widehat\lambda_{k,j,a}^2}{N_{k,j,a}}  \label{eq: lambda var estimator} ~.
\end{align}

An estimator for the vaccine efficacy within subinterval $(k,j)$ is
\begin{align*}
    \widehat{\text{VE}}^{\mathrm{obs}}_{k,j}=1- \frac{\widehat\lambda_{k,j,a=1}}{\widehat\lambda_{k,j,a=0}} ~,
\end{align*}
with log transformed confidence interval
\begin{align}
    CI(\widehat{\text{VE}}^{\mathrm{obs}}_{k,j}) = 1-\frac{\widehat\lambda_{k,j,a=1}}{\widehat\lambda_{k,j,a=0}}\cdot \exp\left( \pm  z_{1-\alpha/2}\cdot \sqrt{\frac{1}{N_{k,j,a=0}}+\frac{1}{N_{k,j,a=1}}}  \right) ~, \label{eq: CI_VE_kj}
\end{align}
also described by \citet{ewell_comparing_1996}. \citet{wei_confidence_2022} found that the coverage probability of \eqref{eq: CI_VE_kj} can be lower than than the nominal level when the VE is close to 1 in finite sample simulations.
The log transform ensures that the upper confidence interval for VE does not exceed 1, and was found to improve the error rate for confidence intervals of the cumulative hazard based on the Nelson-Aalen estimator in finite sample simulations \citep{bie_confidence_1987}. In Table~\ref{tab:validation CI} (Appendix~\ref{app: further analyses}), we apply \eqref{eq: CI_VE_kj} to the data from Figure~2 of \citet{thomas_safety_2021}, which yields identical point estimates and nearly identical confidence intervals to those reported by \citet{thomas_safety_2021}.

Next, we estimate the cumulative hazard $\Lambda_{k,a}$ by
\begin{align}
    \widehat\Lambda_{k,a}=\sum_{j=1}^{j_k} \widehat\lambda_{k,j,a}\tau_{k,j}, \quad \widehat\var\widehat\Lambda_{k,a}=\sum_{j=1}^{j_k} \widetilde\var\widehat\lambda_{k,j,a}\tau_{k,j}^2  ~, \label{eq:cumulative hazard estimator}
\end{align}
which we use to define plugin estimators in Table~\ref{tab:estimators}. Variance estimators of the log transformed estimators are defined in Table~\ref{tab:variance estimates}.

\begin{table}[htbp]
  \centering
  \caption{Estimators of \eqref{eq: observed VE}, \eqref{eq:L_VE_challenge_2}-\eqref{eq:U_VE_challenge_2} and \eqref{eq: LB minimal ass 2 intervals}-\eqref{eq: UB minimal ass 2 intervals} with losses to follow-up under Assumptions~\ref{ass: rare events K intervals} and \ref{ass: constant hazard}}
    \begin{tabular}{ll}
    \toprule
    Estimator & Definition \\
    \midrule
    $\widehat{\text{VE}}_1^{\text{obs}},\widehat{\text{VE}}_1^{\text{challenge}}$ & $1-\widehat\Lambda_{k=1,a=1}/\widehat\Lambda_{k=1,a=0}$ \\
    $\widehat{\text{VE}}_2^{\text{obs}} $ & $1-\widehat\Lambda_{k=2,a=1}/\widehat\Lambda_{k=2,a=0}$ \\
   $\widehat{\mathcal{L}}_2$ &  $1-(\widehat\Lambda_{k=1,a=1}+\widehat\Lambda_{k=2,a=1})/\widehat\Lambda_{k=2,a=0}$  \\
   $\widehat{\mathcal{U}}_2$ & $1-\widehat\Lambda_{k=2,a=1}/(\widehat\Lambda_{k=1,a=0}+\widehat\Lambda_{k=2,a=0})$  \\
    $\widehat{\mathcal{L}}_\psi$ & $\widehat\Lambda_{k=1,a=1}/\widehat\Lambda_{k=1,a=0}\cdot \widehat\Lambda_{k=2,a=0}/(\widehat\Lambda_{k=1,a=1}+\widehat\Lambda_{k=2,a=1})$  \\
    $\widehat{\mathcal{U}}_\psi$ & $\widehat\Lambda_{k=1,a=1}/\widehat\Lambda_{k=1,a=0}\cdot (\widehat\Lambda_{k=1,a=0}+\widehat\Lambda_{k=2,a=0})/\widehat\Lambda_{k=2,a=1}$ 
    \end{tabular}%
  \label{tab:estimators}%
\end{table}%
\begin{table}[htbp]
  \centering
  \caption{Variance of log transformed estimators in Table~\ref{tab:estimators}}
    \resizebox{\linewidth}{!}{
    \begin{tabular}{ll}
    \toprule
    Log transformed variance estimator & Definition\\
    \midrule
    $\widetilde\var \log(1-\widehat{\text{VE}}_1^\mathrm{obs}),\widetilde\var \log(1-\widehat{\text{VE}}_1^\mathrm{challenge})$ & $\frac{\widehat\var \widehat\Lambda_{k=1,a=0}}{\widehat\Lambda_{k=1,a=0}^2} + \frac{\widehat\var \widehat\Lambda_{k=1,a=1}}{\widehat\Lambda_{k=1,a=1}^2}$ \\
    $\widetilde\var \log(1-\widehat{\text{VE}}_2^\mathrm{obs}) $ & $\frac{\widehat\var \widehat\Lambda_{k=2,a=0}}{\widehat\Lambda_{k=2,a=0}^2} + \frac{\widehat\var \widehat\Lambda_{k=2,a=1}}{\widehat\Lambda_{k=2,a=1}^2}$ \\
    $\widetilde\var \log(1-\widehat{\mathcal{L}}_2)$ & $\frac{\widehat\var \widehat\Lambda_{k=2,a=0}}{\widehat\Lambda_{k=2,a=0}^2}+ \frac{\widehat\var\widehat\Lambda_{k=1,a=1}+\widehat\var\widehat\Lambda_{k=2,a=1}}{(\widehat\Lambda_{k=1,a=1}+\widehat\Lambda_{k=2,a=1})^2}$   \\
    $\widetilde\var \log(1-\widehat{\mathcal{U}}_2) $ & $\frac{\widehat\var \widehat\Lambda_{k=2,a=1}}{\widehat\Lambda_{k=2,a=1}^2} +\frac{\widehat\var \widehat\Lambda_{k=1,a=0}+\widehat\var \widehat\Lambda_{k=2,a=0}}{(\widehat\Lambda_{k=1,a=0}+\widehat\Lambda_{k=2,a=0})^2}$ \\
    $\widetilde\var\log\widehat{\mathcal{L}}_\psi$ & $\frac{\widehat\var \widehat\Lambda_{k=1,a=0}}{\widehat\Lambda^2_{k=1,a=0} } + \frac{\widehat\var \widehat\Lambda_{k=2,a=0}}{\widehat\Lambda^2_{k=2,a=0} } + \frac{\widehat\var \widehat\Lambda_{k=2,a=1}}{(\widehat\Lambda_{k=1,a=1}+\widehat\Lambda_{k=2,a=1})^2} + \frac{\widehat\var \widehat\Lambda_{k=1,a=1}}{(\widehat\Lambda_{k=1,a=1}+\widehat\Lambda_{k=2,a=1})^2}\cdot \frac{\widehat\Lambda_{k=2,a=1}^2}{\widehat\Lambda_{k=1,a=1}^2}$ \\
    $\widetilde\var\log\widehat{\mathcal{U}}_\psi$ & $\frac{\widehat\var \widehat\Lambda_{k=1,a=1}}{\widehat\Lambda^2_{k=1,a=1} } + \frac{\widehat\var \widehat\Lambda_{k=2,a=1}}{\widehat\Lambda^2_{k=2,a=1} } + \frac{\widehat\var \widehat\Lambda_{k=2,a=0}}{(\widehat\Lambda_{k=1,a=0}+\widehat\Lambda_{k=2,a=0})^2} + \frac{\widehat\var \widehat\Lambda_{k=1,a=0}}{(\widehat\Lambda_{k=1,a=0}+\widehat\Lambda_{k=2,a=0})^2}\cdot \frac{\widehat\Lambda_{k=2,a=0}^2}{\widehat\Lambda_{k=1,a=0}^2}$ 
    \end{tabular}%
    }
  \label{tab:variance estimates}%
\end{table}%

Finally, we construct asymptotic two-sided $1-\alpha$ confidence intervals using an exponential transformation of estimators in Tables~\ref{tab:estimators}-\ref{tab:variance estimates}. For example, the confidence interval of $\text{VE}_1^\mathrm{challenge}$ is
\begin{align*}
    CI(\text{VE}_1^\mathrm{challenge}) = 1- \exp\left\{ \log (1-\widehat{\text{VE}}_1^\mathrm{challenge}) \mp  z_{1-\alpha/2}\cdot\sqrt{\widetilde\var\log (1-\widehat{\text{VE}}_1^\mathrm{challenge})} \right\} ~.
\end{align*}
For the quantities $\mathcal{L}_2,\mathcal{U}_2, \mathcal{L}_\psi$ and $\mathcal{U}_\psi$, which are used to compute bounds, we construct one-sided confidence intervals to ensure a coverage level of $1-\alpha$ for the lower confidence limit of the lower bound, and the upper confidence limit of the upper bound.

\subsubsection{Large sample properties of $\widehat\lambda_{k,j,a}$}

We use the delta method (see e.g.\ \citet[Section 1.8]{lehmann_theory_1998}) to derive the limiting distribution of the estimators introduced in the previous subsection, and show that the proposed confidence intervals are asymptotically valid in large samples under Assumption~\ref{ass: rare events K intervals}. Define $N_{k,j,i}\in\{0,1\}$ to be an indicator of the outcome of interest and $T_{k,j,i}\in[0,\tau_{k,j}]$ to be the time at risk of individual $i$ during subinterval $(k,j)$. If an individual is censored (lost to follow-up) or experiences the outcome during $(k,j)$, then $T_{k,j,i}< \tau_{k,j}$. Conversely, if individual $i$ is censored or has an event before time interval $(k,j)$, we define $T_{k,j,i}=0$. Hence $N_{k,j,a}=\sum_{i=1}^{n} I(A_i=a)N_{k,j,i}$ and $T_{k,j,a}=\sum_{i=1}^{n} I(A_i=a)T_{k,j,i}$. This allows us to write $\widehat\lambda_{k,j,a}$ as
\begin{align*}
    \widehat\lambda_{k,j,a}= \frac{N_{k,j,a}/n}{T_{k,j,a}/n} = \frac{\frac{1}{n} \sum_{i=1}^n I(A_i=a)N_{k,j,i} }{\frac{1}{n}\sum_{i^\prime =1}^n I(A_{i^\prime}=a)T_{k,j,i^\prime}} =\frac{\frac{1}{n} \sum_{i=1}^n N_{k,j,a,i} }{\frac{1}{n}\sum_{i^\prime =1}^nT_{k,j,a,i^\prime}} ~,
\end{align*}
where we have defined the short-hand notation $N_{k,j,a,i}=I(A_i=a)N_{k,j,i}$, $T_{k,j,a,i}=I(A_i=a)T_{k,j,i}$. Since $N_{k,j,a,i}$ and $T_{k,j,a,i}$ are bounded and therefore have finite mean, and also finite variance by Popoviciu's inequality on variances, we obtain the asymptotic distribution of their empirical means using the central limit theorem (CLT).  Let $X^{(n)}=(1/n) \sum_{i=1}^n X_i$ where $X_i$ is a vector containing components $N_{k,j,a,i}$ and $T_{k,j,a,i}$ for all indices $k,j,a$. Next, let  $m=E[X_i]$ and $ \Sigma_{r,s}=\cov(X_{i,r},X_{i,s})$, where $\Sigma_{r,s}$ is the $(r,s)$ component of covariance matrix $\Sigma$, and $X_{i,r}$ is $r$-th component of $X_i$. By the multivariate CLT, 
\begin{align}
    \sqrt{n}(X^{(n)}-m)\xlongrightarrow{d} \mathcal{N}(0,\Sigma) \label{eq: CLT} ~.
\end{align}
Next, we define the transformation $h$ such that 
\begin{align}
    h(m)= \boldsymbol{\lambda} ~, \label{eq: transformation} 
\end{align}
where $\boldsymbol{\lambda}$ is a vector with components $\lambda_{k,j,a}$ for all $k,j,a$. The components of $h$ are ratios of pairs of components of $m$, and consequently $h$ has continuous partial derivatives. Likewise, let $\widehat{\boldsymbol{\lambda}}=h(X^{(n)})$ denote the corresponding vector of estimates $\widehat\lambda_{k,j,a}$. Using \eqref{eq: CLT} and \eqref{eq: transformation} in Theorem 8.22 in \citet{lehmann_theory_1998} then gives
\begin{align}
    \sqrt{n}(\widehat{\boldsymbol{\lambda}}-\boldsymbol{\lambda})\xlongrightarrow{d} \mathcal{N}(0,B\Sigma B^T)~, \label{eq: asymptotic distribution}
\end{align}
where $B$ is a matrix of partial derivatives $B_{r,s}=\partial h_r/\partial X^{(n)}_s$. The covariance matrix $B\Sigma B^T$ has entries equal to the asymptotic covariances of $\widehat\lambda_{k,j,a}$ and $\widehat\lambda_{k^\prime, j^\prime, a^\prime}$, given by
\begin{align*}
    n\cov (\widehat\lambda_{k,j,a},\widehat\lambda_{k^\prime, j^\prime, a^\prime}) \xlongrightarrow{p} \nabla g(\mu_{k,j,a})^T \Omega_{k,j,a}^{k^\prime, j^\prime, a^\prime} \nabla g(\mu_{k^\prime, j^\prime, a^\prime}) ~,
\end{align*}
where $g(x,y)=x/y$, $\mu_{k,j,a}=(E[N_{k,j,a,i}],E[T_{k,j,a,i}])^T$ and
\begin{align*}
    \Omega_{k,j,a}^{k^\prime, j^\prime, a^\prime} = \begin{pmatrix}
\cov (N_{k,j,a,i},N_{k^\prime, j^\prime, a^\prime,i}) & \cov (N_{k,j,a ,i},T_{k^\prime, j^\prime, a^\prime,i}) \\
\cov (T_{k,j,a,i},N_{k^\prime, j^\prime, a^\prime,i}) & \cov (T_{k,j,a,i},T_{k^\prime, j^\prime, a^\prime,i}) 
\end{pmatrix} ~.
\end{align*}

We proceed by showing that 
\begin{align}
     &\nabla g(\mu_{k,j,a})^T\Omega_{k,j,a}^{k^\prime, j^\prime, a^\prime } \nabla  g(\mu_{k^\prime, j^\prime, a^\prime }) \notag\\
    =&\begin{cases}
        \frac{\lambda_{k,j,a}^2}{E[N_{k,j,a,i}]}\Bigg( 1-\lambda_{k,j,a}\tau_{k,j}\frac{E[T_{k,j,a,i}]}{\tau_{k,j}}+\gamma_{k,j,a}\left(1-\frac{E[T_{k,j,a,i}]}{\tau_{k,j}}\right)\lambda_{k,j,a}\tau_{k,j} \\
     \quad -2\eta_{k,j,a}\sqrt{\left(1-\lambda_{k,j,a}\tau_{k,j}\frac{E[T_{k,j,a,i}]}{\tau_{k,j}}\right)\gamma_{k,j,a}\left(1-\frac{E[T_{k,j,a,i}]}{\tau_{k,j}}\right)\lambda_{k,j,a}\tau_{k,j} } \Bigg)  \text{ if } (k,j,a)=(k^\prime,j^\prime,a^\prime) ~, \\
    0 \text{ otherwise } ~,
    \end{cases} \notag\\\label{eq: asymptotic var}
\end{align}
where
\begin{align*}
    \gamma_{k,j,a} &=  \frac{\var T_{k,j,a,i}}{E[T_{k,j,a,i}]^2(\tau_{k,j}/E[T_{k,j,a,i}]-1)} ~, \\
    \eta_{k,j,a} &= \frac{\cov (N_{k,j,a},T_{k,j,a})}{\sqrt{\var N_{k,j,a}\cdot \var T_{k,j,a}}} ~.
\end{align*}
In other words, $B\Sigma B^T$ is a diagonal matrix, and therefore $\widehat\lambda_{k,j,a},\widehat\lambda_{k^\prime,j^\prime,a^\prime}$ are asymptotically uncorrelated.

To begin the derivation of \eqref{eq: asymptotic var}, consider the case $a\neq a^\prime$. Then,
\begin{align*}
    &\nabla g(\mu_{k,j,a=1})^T \Omega_{k,j,a=1}^{k^\prime, j^\prime, a=0} \nabla g(\mu_{k^\prime, j^\prime, a=0}) \\
    = &\begin{pmatrix}
       \frac{1}{ E[\Delta T_{k,j,a=1,i}]} & -\frac{E[N_{k,j,a=1,i}]}{E[T_{k,j,a=1,i}]^2}
    \end{pmatrix}\begin{pmatrix}
        -E[N_{k,j,a=1,i}]E[N_{k^\prime, j^\prime, a=0,i}] & -E[N_{k,j,a=1,i}]E[T_{k^\prime, j^\prime, a=0,i}] \\
        -E[T_{k,j,a=1,i}]E[N_{k^\prime, j^\prime, a=0,i}] & -E[T_{k,j,a=1,i}]E[T_{k^\prime, j^\prime, a=0,i}]
    \end{pmatrix} \\
    &\quad \times 
    \begin{pmatrix}
        \frac{1}{E[\Delta T_{k^\prime, j^\prime ,a=0,i}] }\\ -\frac{E[N_{k^\prime, j^\prime ,a=0,i}]}{E[T_{k^\prime, j^\prime ,a=0,i}]^2}
    \end{pmatrix}\\
    =&0  ~.
\end{align*}
This reflects the fact that observations in different treatment groups are independent since individuals are i.i.d. We derived the covariance matrix from the fact that
\begin{align*}
    &\cov (N_{k,j,a=1,i},N_{k^\prime,j^\prime,a=0,i}) \\
   =&E[N_{k,j,a=1,i}N_{k^\prime,j^\prime,a=0,i}]-E[N_{k,j,a=1,i}]E[N_{k^\prime,j^\prime,a=0,i}] \\
    =&E[I(A_i=1)I(A_i=0)N_{k,j,i}N_{k^\prime,j^\prime,i}]-E[N_{k,j,a=1,i}]E[N_{k^\prime,j^\prime,a=0,i}] \\
    =&-E[N_{k,j,a=1,i}]E[N_{k^\prime,j^\prime,a=0,i}] ~,
\end{align*}
and likewise
\begin{align*}
    \cov (N_{k,j,a=1,i},T_{k^\prime,j^\prime,a=0,i}) &= -E[N_{k,j,a=1,i}]E[T_{k^\prime,j^\prime,a=0,i}] ~, \\
    \cov (T_{k,j,a=1,i},T_{k^\prime,j^\prime,a=0,i}) &= -E[T_{k,j,a=1,i}]E[T_{k^\prime,j^\prime,a=0,i}] ~.
\end{align*}

Next, consider the case $(k,j) \neq (k^\prime,j^\prime)$. Then,
\begin{align}
    &\nabla g(\mu_{k,j,a})^T\Omega_{k,j,a}^{k^\prime, j^\prime, a} \nabla g(\mu_{k^\prime, j^\prime, a}) \notag\\
    =& \frac{1}{E[T_{k,j,a,i}]}\frac{1}{E[T_{k^\prime, j^\prime, a,i}]}\cov (N_{k,j,a,i},N_{k^\prime, j^\prime, a,i}) \notag\\
    &\qquad - \frac{1}{E[T_{k,j,a,i}]}\frac{E[N_{k^\prime, j^\prime, a,i}]}{E[T_{k^\prime, j^\prime, a,i}]^2} \cov (N_{k,j,a,i},T_{k^\prime, j^\prime, a,i}) \notag\\
    &\qquad - \frac{E[N_{k,j,a,i}]}{E[T_{k,j,a,i}]^2}\cdot \frac{1}{E[T_{k^\prime, j^\prime, a,i}] }\cov (T_{k,j,a,i},N_{k^\prime, j^\prime, a,i})  \notag\\
    &\qquad + \frac{E[N_{k,j,a,i}]}{E[T_{k,j,a,i}]^2}\frac{E[N_{k^\prime, j^\prime, a,i}]}{E[T_{k^\prime, j^\prime, a,i}]^2} \cov (T_{k,j,a,i},T_{k^\prime, j^\prime, a,i}) ~, \label{eq: cov expansion}
\end{align}
and
\begin{align}
    \cov (N_{k,j,a,i},N_{k^\prime, j^\prime, a,i}) &= -E[N_{k,j,a,i}]E[N_{k^\prime, j^\prime, a,i}] \label{eq: Omega NN ii}\\
    \cov (N_{k,j,a,i},T_{k^\prime, j^\prime, a,i}) &= E[N_{k,j,a,i}](I(kj\text{ after } k^\prime j^\prime)\tau_{k^\prime, j^\prime} - E[T_{k^\prime, j^\prime, a,i}])\label{eq: Omega NT ii} \\
    \cov (T_{k,j,a,i},N_{k^\prime, j^\prime, a,i}) &= E[N_{k^\prime, j^\prime, a,i}](I(k^\prime, j^\prime \text{ after } k j)\tau_{k, j} - E[T_{k, j, a,i}]) \label{eq: Omega TN ii} \\
    \cov (T_{k,j,a,i},T_{k^\prime, j^\prime, a,i}) &= E[T_{k,j,a,i}]E[T_{k^\prime, j^\prime, a,i}]\bigg( I(kj \text{ after } k^\prime j^\prime )\frac{\tau_{k^\prime, j^\prime } }{E[T_{k^\prime, j^\prime a,i}]} \notag\\
    &\quad + I(k^\prime j^\prime  \text{ after } k j )\frac{\tau_{k, j} }{E[T_{k, j, a,i}]}\bigg) - E[T_{k,j,a,i}]E[T_{k^\prime j^\prime a,i}] ~. \label{eq: Omega TT ii}
\end{align}

We derive \eqref{eq: Omega NN ii}-\eqref{eq: Omega TT ii} in turn. First,
\begin{align*}
    &\cov (N_{k,j,a,i},N_{k^\prime, j^\prime, a,i} )\\
    =&E[N_{k,j,a,i}N_{k^\prime, j^\prime, a,i}]-E[N_{k,j,a,i}]E[N_{k^\prime, j^\prime, a,i}] \\
    =&E[N_{k,j,a,i}I(N_{k,j,a,i}=0)N_{k^\prime, j^\prime, a,i}]-E[N_{k,j,a,i}]E[N_{k^\prime, j^\prime, a,i}] \\
    =&-E[N_{k,j,a,i}]E[N_{k^\prime, j^\prime, a,i}] ~.
\end{align*}
Without loss of generality, we have taken $(k,j)$ to be prior to $(k^\prime, j^\prime)$ in the third line. 

Next, we consider $\cov (N_{k,j,a,i},T_{k^\prime, j^\prime, a,i})$. Suppose $(k,j)$ occurs before $(k^\prime, j^\prime)$. Then,
\begin{align*}
    &\cov (N_{k,j,a,i},T_{k^\prime, j^\prime, a,i}) \\
    =& E[N_{k,j,a,i}I(N_{k,j,a,i}=0)T_{k^\prime, j^\prime, a,i}]-E[N_{k,j,a,i}]E[T_{k^\prime, j^\prime, a,i}] \\
    =& -E[N_{k,j,a,i}]E[T_{k^\prime, j^\prime, a,i}] ~.
\end{align*}
However, if  $(k,j)$ occurs after  $(k^\prime, j^\prime)$, 
\begin{align*}
    &\cov (N_{k,j,a,i},T_{k^\prime, j^\prime a,i}) \\
    =&E[I(T_{k^\prime, j^\prime, a,i}=\tau_{k^\prime, j^\prime })T_{k^\prime, j^\prime, a,i} N_{k,j,a,i}]-E[N_{k,j,a,i}]E[T_{k^\prime, j^\prime, a,i}] \\
    =& \tau_{k^\prime, j^\prime }E[N_{k,j,a,i}]-E[N_{k,j,a,i}]E[T_{k^\prime, j^\prime a,i}]~,
\end{align*}
hence \eqref{eq: Omega NT ii} holds. Expression \eqref{eq: Omega TN ii} follows by permuting $(k,j)\leftrightarrow (k^\prime, j^\prime)$ in \eqref{eq: Omega NT ii}. 

Next, to derive \eqref{eq: Omega TT ii}, suppose again that $(k,j)$ is prior to $(k^\prime, j^\prime)$. Then,
\begin{align*}
    &\cov (T_{k,j,a,i},T_{k^\prime, j^\prime, a,i}) \\
    =&E[T_{k,j,a,i}I(T_{k,j,a,i}=\tau_{k,j})T_{k^\prime, j^\prime, a,i}]-E[T_{k,j,a,i}]E[T_{k^\prime, j^\prime, a,i}] \\
    =&\tau_{k,j}E[T_{k^\prime, j^\prime, a,i}]-E[T_{k,j,a,i}]E[T_{k^\prime, j^\prime, a,i}] ~,
\end{align*}
and vice versa when $(k,j)$ occurs after $(k^\prime, j^\prime)$.

Substituting \eqref{eq: Omega NN ii}-\eqref{eq: Omega TT ii} into \eqref{eq: cov expansion} establishes that $\nabla g(\mu_{k,j,a})^T\Omega_{k,j,a}^{k^\prime, j^\prime, a} \nabla g(\mu_{k^\prime, j^\prime a})=0$.

Finally, we consider the diagonal entries of $B\Sigma B^T$, given by $\nabla g(\mu_{k,j,a})^T\Omega_{k,j,a}^{k,j,a} \nabla  g(\mu_{k,j,a})$. Evaluating the gradient gives $\nabla g(\mu_{k,j,a})=(1/E[T_{k,j,a,i}],-E[N_{k,j,a,i}]/E[T_{k,j,a,i}]^2)^T$. The entries of the matrix $\Omega_{k,j,a}^{k,j,a}$ are
\begin{align}
    &\var N_{k,j,a,i} = E[N_{k,j,a,i}](1-E[N_{k,j,a,i}]) ~,\label{eq: Omega NN}\\
    &\var T_{k,j,a,i} = \gamma_{k,j,a}\left(\tau_{k,j}E[T_{k,j,a,i}]-E[T_{k,j,a,i}]^2\right) ~,\label{eq: Omega TT}\\
    &\cov (N_{k,j,a,i},T_{k,j,a,i}) = \notag\\
    &\qquad\eta_{k,j,a}\sqrt{E[N_{k,j,a,i}](1-E[N_{k,j,a,i}])E[T_{k,j,a,i}]^2\gamma_{k,j,a}\left(\frac{\tau_{k,j}}{E[T_{k,j,a,i}]}-1\right)} ~. \label{eq: Omega NT}
\end{align}

Expression \eqref{eq: Omega NN} holds since $N_{k,j,a,i}\sim \mathrm{Ber}(E[N_{k,j,a,i}])$. Next, \eqref{eq: Omega TT} holds by definition of $\gamma_{k,j,a}$. 
Since $T_{k,j,a,i}$ is bounded by $0\leq T_{k,j,a,i} \leq \tau_{k,j}$,
\begin{align*}
    \var T_{k,j,a,i} &= \tau^2_{k,j} \var \left(\frac{T_{k,j,a,i}}{\tau_{k,j}}\right) \\
    &= \tau_{k,j}^2\left(E\left[\left(\frac{T_{k,j,a,i}}{\tau_{k,j}}\right)^2 \right]-E\left[\frac{T_{k,j,a,i}}{\tau_{k,j}} \right]^2 \right) \\
    &\leq \tau_{k,j}^2\left(E\left[\frac{T_{k,j,a,i}}{\tau_{k,j}} \right]-E\left[\frac{T_{k,j,a,i}}{\tau_{k,j}} \right]^2 \right) \\
    &=E[T_{k,j,a,i}](\tau_{k,j}-E[T_{k,j,a,i}])~,
\end{align*}
and therefore $0\leq \gamma_{k,j,a} \leq 1$.
Finally, expression \eqref{eq: Omega NT} follows from the definition of $\eta_{k,j,a}$, and  $0\leq |\eta_{k,j,a}| \leq 1$ by the Cauchy-Schwarz inequality.

Thus, using $\lambda_{k,j,a}=E[N_{k,j,a,i}]/E[T_{k,j,a,i}]$ and \eqref{eq: Omega NN}-\eqref{eq: Omega NT} in $\nabla g(\mu_{k,j,a})^T\Omega_{k,j,a}^{k,j,a} \nabla  g(\mu_{k,j,a})$ gives \eqref{eq: asymptotic var}.

Expressions \eqref{eq: asymptotic distribution} and \eqref{eq: asymptotic var} imply that
\begin{align}
   \sqrt{n}( \widehat\lambda_{k,j,a}-\lambda_{k,j,a} )\xlongrightarrow{d}\mathcal{N}(0,\nabla g(\mu_{k,j,a})^T\Omega_{k,j,a}^{k,j,a} \nabla g(\mu_{k,j,a}))~. \label{eq: variance lambda}
\end{align}
We construct a variance estimator of $\widehat\lambda_{k,j,a}$ by estimating $\nabla g(\mu_{k,j,a})^T\Omega_{k,j,a}^{k,j,a} \nabla g(\mu_{k,j,a})$. 
By Slutsky's theorem and the continous mapping theorem,
\begin{align*}
    \frac{\widehat\lambda_{k,j,a}^2}{\frac{1}{n}\sum_{i=1}^n N_{k,j,a,i}} \xlongrightarrow{p}  \frac{\lambda_{k,j,a}^2}{E[N_{k,j,a,i}]} ~,
\end{align*}
and thus, by \eqref{eq: asymptotic var}, 
\begin{align}
   \frac{\frac{\widehat\lambda_{k,j,a}^2}{N_{k,j,a}}}{\frac{1}{n}\nabla g(\mu_{k,j,a})^T\Omega_{k,j,a}^{k,j,a} \nabla g(\mu_{k,j,a})}\xlongrightarrow{p} 1 \label{eq: avar lambda}
\end{align}
as $n\longrightarrow \infty$ and $  \lambda_{k,j,a}\tau_{k,j}\longrightarrow 0$, using the fact that $\gamma_{k,j,a},|\eta_{k,j,a}|, E[T_{k,j,a,i}]/\tau_{k,j}\in [0,1]$. This motivates the variance estimator $\widetilde\var\widehat\lambda_{k,j,a}$, defined in \eqref{eq: lambda var estimator}, in the limit of rare events (Assumption~\ref{ass: rare events K intervals}).

\subsubsection{Composite intervals\label{appsec:composite intervals}}

To derive the log transformed confidence interval \eqref{eq: CI_VE_kj}, let $h_{\log} (\boldsymbol{\lambda})=\log (\lambda_{k,j,a=1}/\lambda_{k,j,a=0})$. Then, by \eqref{eq: asymptotic distribution} and Theorem 8.22 in \citet{lehmann_theory_1998},
\begin{align*}
    \sqrt{n}\left(\log \frac{\widehat\lambda_{k,j,a=1}}{\widehat\lambda_{k,j,a=0}}-\log\frac{\lambda_{k,j,a=1}}{\lambda_{k,j,a=0}}\right)\xlongrightarrow{d} \mathcal{N}(0,\nabla h_{\log}(\boldsymbol{\lambda} )^T B\Sigma B^T \nabla h_{\log} (\boldsymbol{\lambda}))~,
\end{align*}
where $\nabla h_{\log}(\boldsymbol{\lambda} )^T B\Sigma B^T \nabla h_{\log} (\boldsymbol{\lambda})=\sum_{a=0}^1 (1/\lambda_{k,j,a}^2) \nabla g(\mu_{k,j,a})^T\Omega_{k,j,a}^{k,j,a} \nabla g(\mu_{k,j,a})$. Thus,
\begin{align*}
    \frac{\frac{1}{N_{k,j,a=0}}+\frac{1}{N_{k,j,a=0}}}{\frac{1}{n} \nabla h_{\log}(\boldsymbol{\lambda} )^T B\Sigma B^T \nabla h_{\log} (\boldsymbol{\lambda}))} \xlongrightarrow{p} 1
\end{align*}
as $n\longrightarrow \infty$ and $  \lambda_{k,j,a}\tau_{k,j}\longrightarrow 0$, which establishes $1/N_{k,j,a=0}+1/N_{k,j,a=1}$ as an asymptotic variance estimator of $\log (\widehat\lambda_{k,j,a=1}/\widehat\lambda_{k,j,a=0})$ under Assumption~\ref{ass: rare events K intervals}.

Let $\boldsymbol{\Lambda}$ be a vector containing components $\Lambda_{k,a}$ \eqref{eq: Lambda_ka definition} for all $k,a$, and let $h_\Lambda(\boldsymbol{\lambda})=\boldsymbol{\Lambda}$. Correspondingly, let $\widehat{\boldsymbol{\Lambda}}=h_\Lambda(\widehat{\boldsymbol{\lambda}})$. Then, by using \eqref{eq: asymptotic distribution} in Theorem 8.22 in \citet{lehmann_theory_1998},
\begin{align}
    \sqrt{n}(\widehat{\boldsymbol{\Lambda}}- \boldsymbol{\Lambda} ) \xlongrightarrow{d}\mathcal{N}(0,\Sigma_\Lambda)~, \label{eq: cum haz limiting dist}
\end{align}
where $\Sigma_\Lambda$ is a diagonal covariance matrix with asymptotic variances
\begin{align*}
    n\var \widehat\Lambda_{k,a} \xlongrightarrow{p} \sum_{j=1}^{j_k}\tau_{k,j}^2\nabla g(\mu_{k,j,a})^T\Omega_{k,j,a}^{k,j,a} \nabla g(\mu_{k,j,a}) ~.
\end{align*}
Consequently,
\begin{align*}
    \frac{\widehat\var\widehat\Lambda_{k,a}}{\frac{1}{n}\sum_{j=1}^{j_k}\tau_{k,j}^2\nabla g(\mu_{k,j,a})^T\Omega_{k,j,a}^{k,j,a} \nabla g(\mu_{k,j,a})} \xlongrightarrow{p} 1~,
\end{align*}
as $n\longrightarrow \infty$ and $\lambda_{k,j,a}\tau_{k,j}\longrightarrow 0$, which motivates the variance estimator $\widehat\var\widehat\Lambda_{k,a}$ under Assumption~\ref{ass: rare events K intervals}.

We derive the estimator of the asymptotic variance of $\log  (1-\widehat{\mathcal{L}}_2)$, as the other estimators in Table~\ref{tab:variance estimates} follow from similar arguments, using corresponding transformations. Let $h_\mathcal{L}(\boldsymbol{\Lambda})=\log\{(\Lambda_{k=1,a=1}+\Lambda_{k=2,a=1})/\Lambda_{k=2,a=0}\}$. Then, using \eqref{eq: cum haz limiting dist} in Theorem 8.22 in \citet{lehmann_theory_1998},
\begin{align*}
    \sqrt{n} \left(\log  (1-\widehat{\mathcal{L}}_2) - \log (1-\mathcal{L}_2) \right) \xlongrightarrow{d} \mathcal{N}(0, \nabla h_\mathcal{L}(\boldsymbol{\Lambda})^T \Sigma_\Lambda \nabla h_\mathcal{L}(\boldsymbol{\Lambda}) ) ~.
\end{align*}
Hence, the asymptotic variance of $\log  (1-\widehat{\mathcal{L}}_2)$ is
\begin{align}
    & \nabla h_\mathcal{L}(\boldsymbol{\Lambda})^T \Sigma_\Lambda \nabla h_\mathcal{L}(\boldsymbol{\Lambda}) \notag\\
    =&\frac{1}{(\Lambda_{k=1,a=1}+\Lambda_{k=2,a=1})^2}\sum_{k=1}^2 \sum_{j=1}^{j_k} \tau_{k,j}^2 \nabla g(\mu_{k,j,a=1})^T\Omega_{k,j,a=1}^{k,j,a=1} \nabla g(\mu_{k,j,a=1}) \notag\\
    &\qquad + \frac{1}{\Lambda_{k=2,a=0}^2} \sum_{j=1}^{j_2} \tau_{k=2,j}
    ^2\nabla g(\mu_{k=2,j,a=0})^T \Omega_{k=2,j,a=0}^{k=2,j,a=0} \nabla g(\mu_{k=2,j,a=0})  ~. \label{eq: asymptotic var composite}
\end{align}

Finally,
\begin{align*}
    \frac{\frac{\widehat\var \widehat\Lambda_{2,a=0}}{\widehat\Lambda_{2,a=0}^2}+ \frac{\widehat\var\widehat\Lambda_{1,a=1}+\widehat\var\widehat\Lambda_{2,a=1}}{(\widehat\Lambda_{1,a=1}+\widehat\Lambda_{2,a=1})^2}}{\frac{1}{n}\nabla h_\mathcal{L}(\boldsymbol{\Lambda})^T \Sigma_\Lambda \nabla h_\mathcal{L}(\boldsymbol{\Lambda})} \xlongrightarrow{p} 1
\end{align*}
as $n\longrightarrow \infty$ and $\lambda_{k,j,a}\tau_{k,j}\longrightarrow 0$ for all $k,j,a$. This motivates the estimator $\widetilde\var\log  (1-\widehat{\mathcal{L}}_2)$ in Table~\ref{tab:variance estimates} under Assumption~\ref{ass: rare events K intervals}.

\section{Example: RTS,S/AS01 vaccine against malaria\label{sec:RTSS}}
In this section, we apply the estimators described in Section~\ref{sec:maintext estimation} on the synthetic dataset by \citet{benkeser_estimating_2019}. The dataset is publicly available, and resembles the RTS,S/AS01 malaria vaccine trial described by \citet{noauthor_first_2011,noauthor_phase_2012}. Here, individuals were randomly assigned to the RTS,S/AS01 malaria versus a comparator vaccine, meningococcal serogroup C conjugate vaccine (Menjugate, Novartis) by double-blinded assignment. \citet{noauthor_phase_2012} reported a 1-year cumulative incidence of clinical malaria of 0.37 in the RTS,S/AS01 group and 0.48 for the comparator vaccine. Kaplain-Meier estimates of survival in the synthetic RTSS data are given in Figure~\ref{fig:survival_RTSS}.
\begin{figure}
    \centering
    \includegraphics[width=0.6\linewidth]{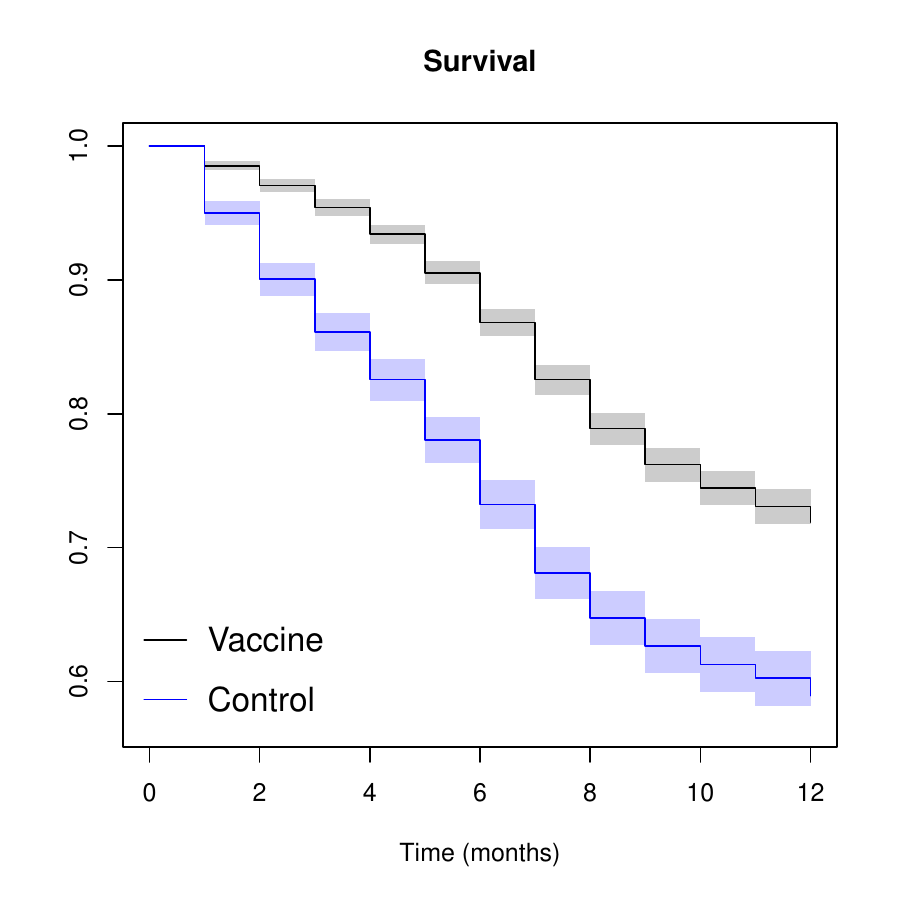}
    \caption{Kaplan-Meier survival estimates for the RTSS data by \citet{benkeser_estimating_2019}. Confidence intervals are shown in shaded colors.}
    \label{fig:survival_RTSS}
\end{figure}
Investigators found that the instantaneous hazard of infection in the RTS,S/AS01 group increased over time relative to the hazard for recipients of the comparator vaccine, and concluded that 
\begin{quote}
    ``[...] [S]tatistical models indicated nonproportionality of hazards over time. This could be due to waning vaccine efficacy, differential acquisition of natural immunity, or other factors that may influence the model, such as heterogeneity of exposure, the vaccine effect at the individual level, or both'' \citep{noauthor_phase_2012}.
\end{quote}
To investigate whether the vaccine protection waned over time, we computed the estimators described in Section~\ref{sec:maintext estimation} without any baseline covariates, reported in Table~\ref{tab:marginal_RTSS}. Here, we let $k=1$ denote the interval from month 1 to the end of month 5 after baseline, and $k=2$ denote the interval from month 6 to the end of month 10.
\begin{table}[htbp]
  \centering
  \caption{Estimates and 95\% bootstrap confidence intervals for the synthetic RTSS dataset by \citet{benkeser_estimating_2019}, computed using the estimators in Section~\ref{sec:maintext estimation} without baseline covariates. One-sided $95\%$ confidence intervals have been used for $\widehat{\mathcal{L}}_2$, $\widehat{\mathcal{U}}_2$, $\widehat{\mathcal{L}}_\psi$ and $\widehat{\mathcal{U}}_\psi$, whereas confidence intervals for $\widehat{\text{VE}}_1^\mathrm{obs},\widehat{\text{VE}}_2^\mathrm{obs}$ and $\widehat{\psi}^\mathrm{obs}$ are two-sided.}
    \begin{tabular}{lr}
         \hline Estimand & \multicolumn{1}{l}{Estimate (95\% CI)}  \\\hline
     $\widehat{\text{VE}}_1^{\text{obs}},\widehat{\text{VE}}_1^{\text{challenge}}$ & $0.57(0.51,0.62)$ \\
    $\widehat{\text{VE}}_2^{\text{obs}}$ & $0.17(0.07,0.26)$ \\
    $\widehat{\mathcal{L}}_2$ & $-0.52 (-0.69, -)$ \\
    $\widehat{\mathcal{U}}_2$ & $0.59(-,0.61)$ \\
    $\widehat{\mathcal{L}}_\psi$ & $0.28 (0.24, - )$\\
    $\widehat{\mathcal{U}}_\psi$ & $1.04(-,1.16)$ \\
    $\widehat{\psi}^\mathrm{obs}$ & $0.52(0.44,0.61)$
    \end{tabular}%
  \label{tab:marginal_RTSS}%
\end{table}%
Additionally, we have computed the estimator $\widehat{\psi}^\mathrm{obs}=(1-\widehat{\text{VE}}_1^\mathrm{obs})/(1-\widehat{\text{VE}}_2^\mathrm{obs}$). Throughout this section, confidence intervals were computed using non-parametric bootstrap with 500 resamples.

The estimated bounds $\widehat{\mathcal{L}}_\psi,\widehat{\mathcal{U}}_\psi$ contain the null-value 1, corresponding to no waning, and therefore do not rule out the possibility that the decline in $\text{VE}^\mathrm{obs}_k$ over time could be due to the depletion of susceptible individuals. However, $\widehat{\mathcal{U}}_2$ is close to $\text{VE}_1^\mathrm{challenge}$. A conditional analysis (Table~\ref{tab:RTSS} and Figure~\ref{fig:CDF}(A)) using the baseline covariates age, sex and study site leads to the same conclusion: the point estimates of $\widehat{\mathcal{U}}_\psi(L)$ indicate waning ($\widehat{\mathcal{U}}_\psi(L)<1$) for a substantial proportion of the observed values of $L$ (Figure~\ref{fig:CDF}(A)), but confidence intervals include the null value of no waning.

\begin{table}[htbp]
  \centering
  \caption{Waning estimates and 95\% bootstrap confidence intervals for the synthetic RTSS dataset by \citet{benkeser_estimating_2019}, illustrated using 3 different combinations of baseline covariates }
    \resizebox{\linewidth}{!}{
    \begin{tabular}{p{0.4cm}p{0.4cm}p{0.4cm}p{2cm}p{2cm}p{2.4cm}p{2.4cm}p{2.4cm}p{2.4cm}p{2cm}}
    \hline Age (weeks) & Sex   & Study site & $\widehat{\text{VE}}_1^\mathrm{obs}(l)$ & $\widehat{\text{VE}}_2^\mathrm{obs}(l)$ & $\widehat{\mathcal{L}}_2(l)$   & $\widehat{\mathcal{U}}_2(l)$ & $\widehat{\mathcal{L}}_\psi(l)$   & $\widehat{\mathcal{U}}_\psi(l)$ & $\widehat{\psi}^\mathrm{obs}(l)$ \\\hline
    51    & female  & 1     & 0.74(0.62,0.81) & 0.53(0.36,0.64) & 0.30(0.10,-) & 0.73(-,0.78) & 0.38(0.33,-) & 0.96(-,1.09) & 0.56(0.46,0.68) \\
    48    & male    & 5     & 0.68(0.58,0.75) & 0.44(0.27,0.57) & -0.01(-0.21,-) & 0.66(-,0.72) & 0.31(0.27,-) & 0.94(-,1.07) & 0.56(0.47,0.68) \\
    58    & male    & 3     & 0.55(0.43,0.64) & 0.23(0.04,0.37) & -0.51(-0.72,-) & 0.55(-,0.61) & 0.30(0.25,-) & 1.00(-,1.14) & 0.58(0.50,0.69) \\
    \end{tabular}
    }
  \label{tab:RTSS}%
\end{table}%
\begin{figure}
    \centering
\subfloat[]{
\includegraphics[width=0.5\linewidth]{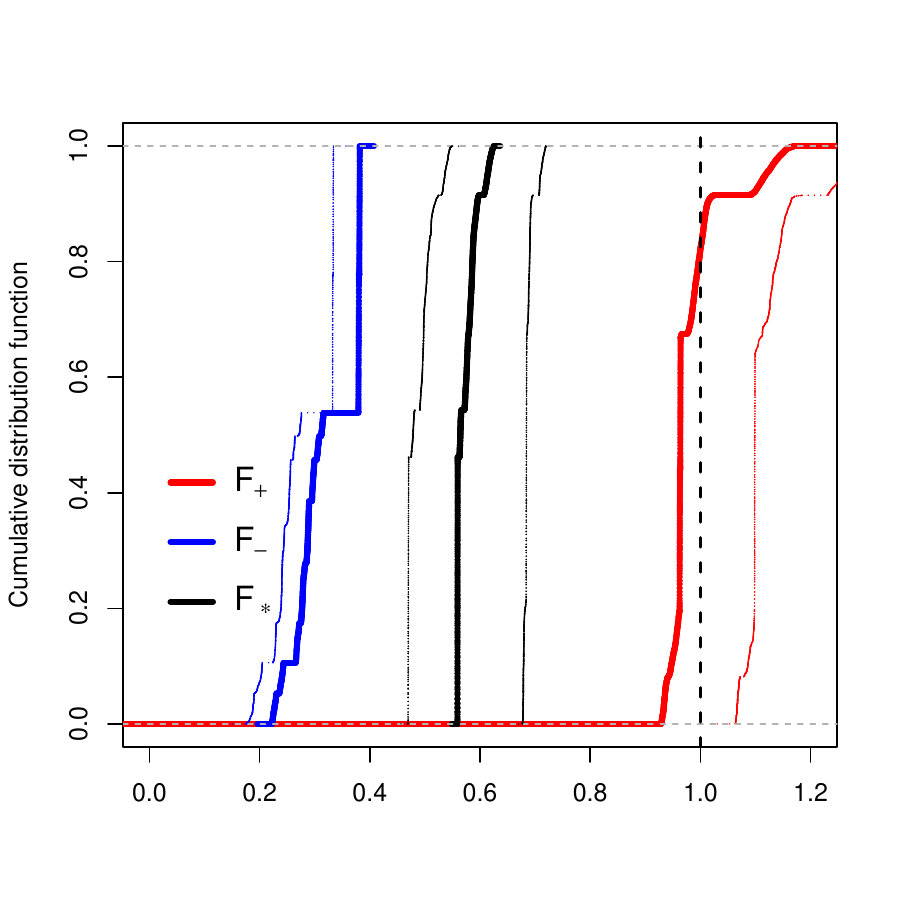}
}
\subfloat[]{
\includegraphics[width=0.5\linewidth]{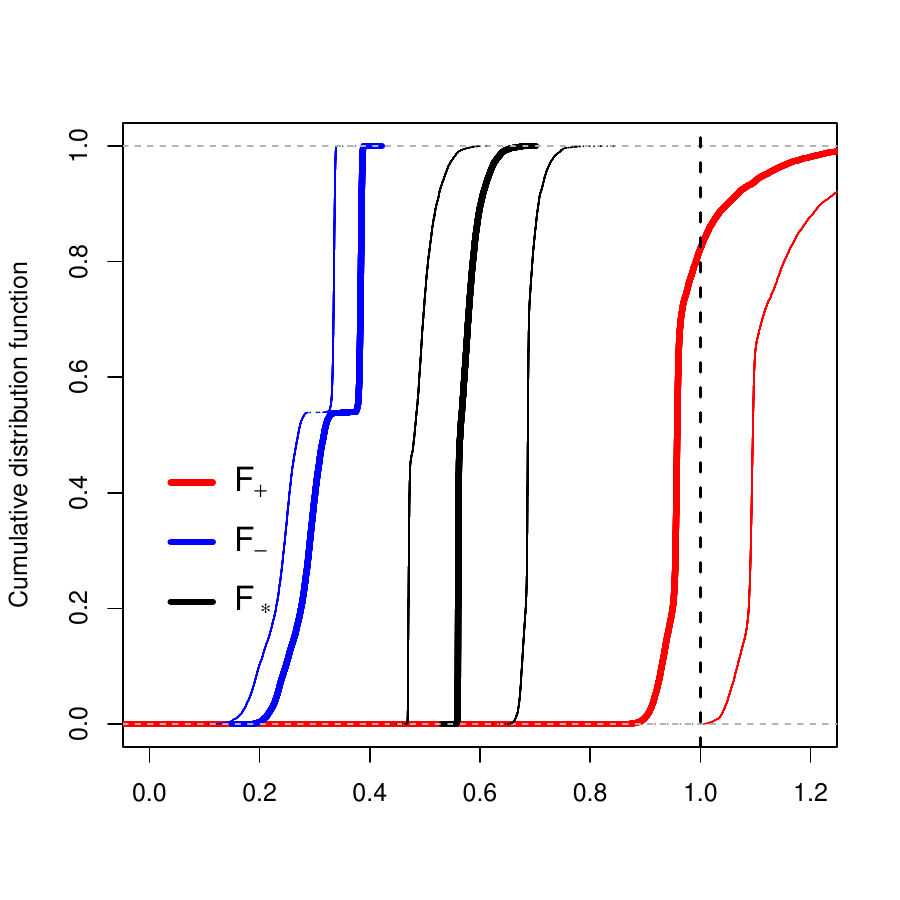}
}
\caption{Cumulative distribution functions (cdf) of bounds and waning estimands. The blue curve is the empirical cdf $F_-$ of $\widehat{\mathcal{L}}_\psi(L)$, over the observed values of $L$. The black curve is the empirical cdf  $F_*$ of $\widehat{\psi}^\mathrm{obs}(L)$, and the red curve is the empirical  cdf  $F_+$ of $\widehat{\mathcal{U}}_\psi(L)$. (A) Shows estimates  with baseline covariates age, sex and study site (B) Shows estimates with baseline covariates age, weight for age (Z score), sex, study site, height for age (Z score), weight for height (Z score), arm circumference (Z score), hemoglobin, distance to nearest inpatient clinic, distance to nearest outpatient clinic and an indicator of rainy versus dry season.  Point estimates are shown with thick lines, and 95\% bootstrap confidence intervals with narrow lines.}
    \label{fig:CDF}
\end{figure}

As a sensitivity analysis, we repeated the analysis in Table~\ref{tab:RTSS}  after breaking tied event times within each month by drawing random days, which gave nearly identical results. Furthermore, we  performed a sensitivity analysis for the choice of baseline covariates, by including additional covariates in Figure~\ref{fig:CDF}(B). The inclusion of additional covariates did not substantially change the distributions of estimates.

\section{Simulated example \label{sec:simulation}}

In this section, we present simulations illustrating the use of logistic regression. Consider a hypothetical vaccine trial with individual-level data from months $k\in\{1,\dots K\}$ following vaccination, where we take $K=4$ (see e.g.\ \citet{voysey_safety_2021}).
Suppose the data were drawn from the following data generating mechanism. First sample $A,L$ according to
\begin{align*}
    A &\sim \mathrm{Ber}(1/2) ~,\\
    L &\sim \mathrm{Unif}[0,1] ~.
\end{align*}
Next, for $k\in\{1,\dots,K\}$, sample $C_k$ and $Y_k$ from the hazards
\begin{align}
    &P(C_k=1\mid Y_{k-1}=0,C_{k-1}=0,A=a,L=l)=\beta_{C} ~,\notag\\
    &\Lambda_{k,a,l} =f_k(a,l;\beta_k)=\expit(\beta_{0,k}+\beta_{1,k}a+\beta_{2,k}l) \label{eq: logistic model} ~.
\end{align}
For each time interval $k$, we computed maximum likelihood estimates $\widehat\beta_{k}$ under the logistic model \eqref{eq: logistic model}, and estimated $\mathcal{L}_k(l),\mathcal{U}_k(l)$ using \eqref{eq: plugin estimators}. We computed confidence intervals using non-parametric bootstrap with $500$ resamples of $n=$~10,000 individuals. The resulting estimates and confidence intervals are shown in Figure~\ref{fig:simulation}. In Figure~\ref{fig:simulation}(A) and (B), the estimated bounds place informative constraints on the extent of vaccine waning, and are close in value to the observed vaccine efficacy $\text{VE}_k^\mathrm{obs}$. However, for the choice of parameters in Figure~\ref{fig:simulation}(C) and (D),  the bounds $\mathcal{L}_k(l),\mathcal{U}_k(l)$ differ substantially from $\text{VE}_k^\mathrm{obs}$, which illustrates that $\text{VE}_k^\mathrm{challenge}$ can be far greater (or smaller) than $\text{VE}_k^\mathrm{obs}$.

\begin{figure}
    \centering
\subfloat[]{
\includegraphics[width=0.5\linewidth]{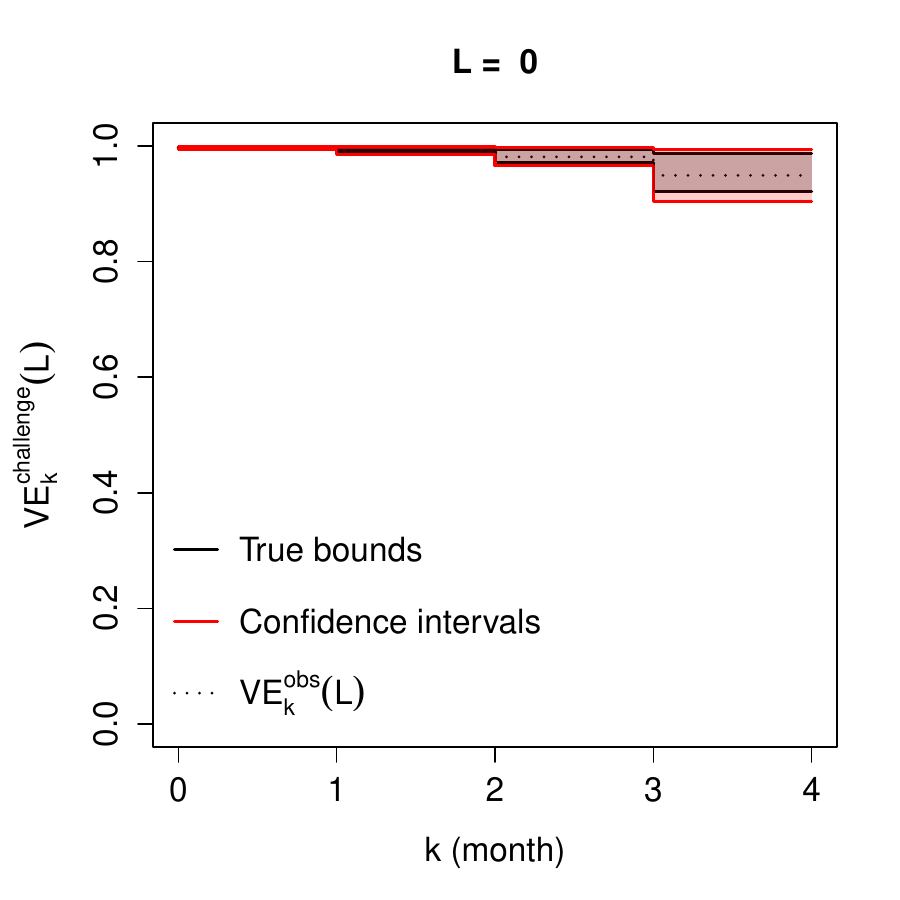}
}
\subfloat[]{
\includegraphics[width=0.5\linewidth]{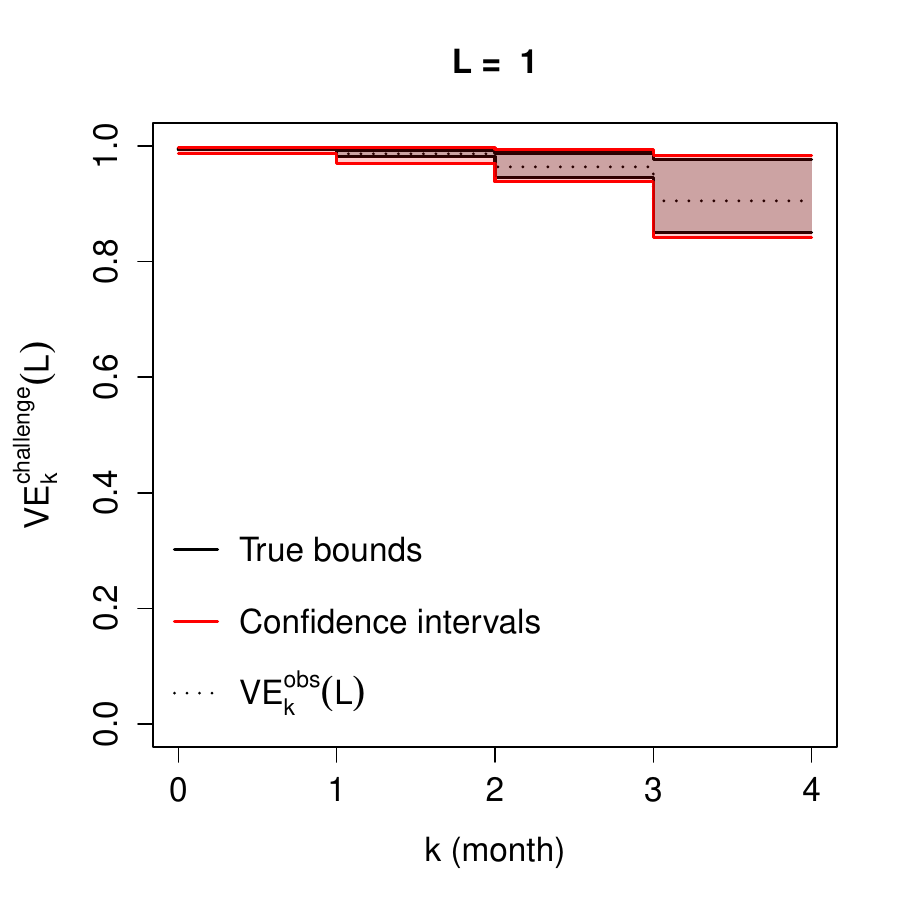}
}
\\
\subfloat[]{
\includegraphics[width=0.5\linewidth]{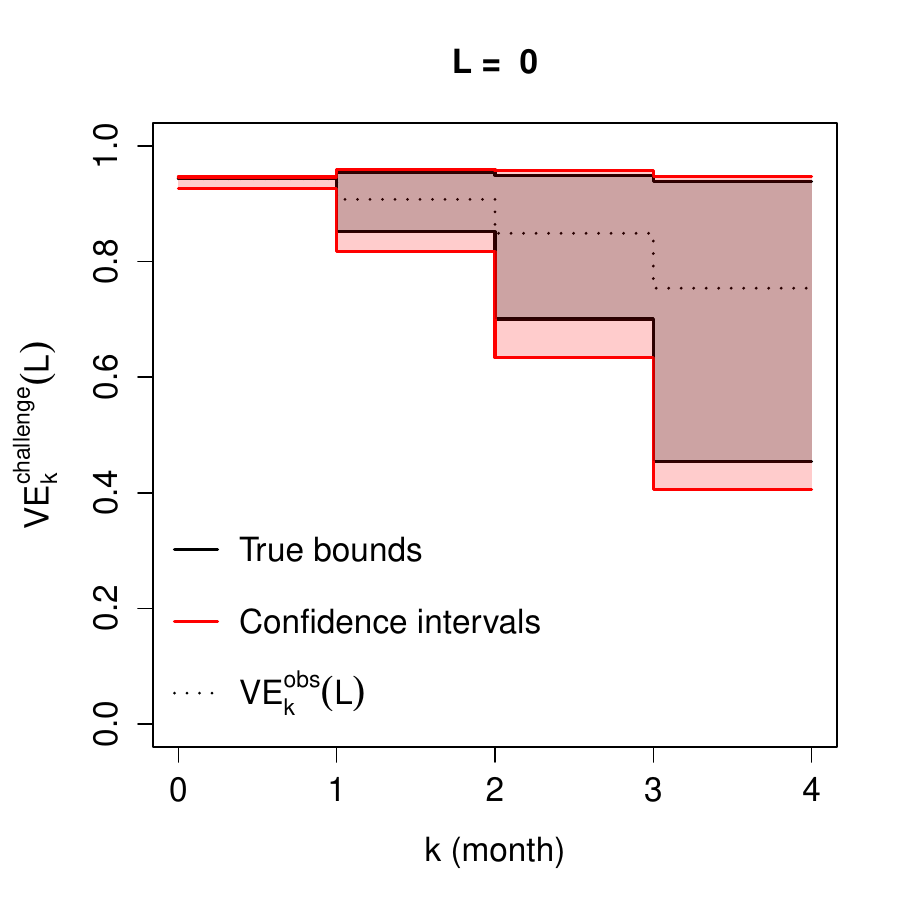}
}
\subfloat[]{
\includegraphics[width=0.5\linewidth]{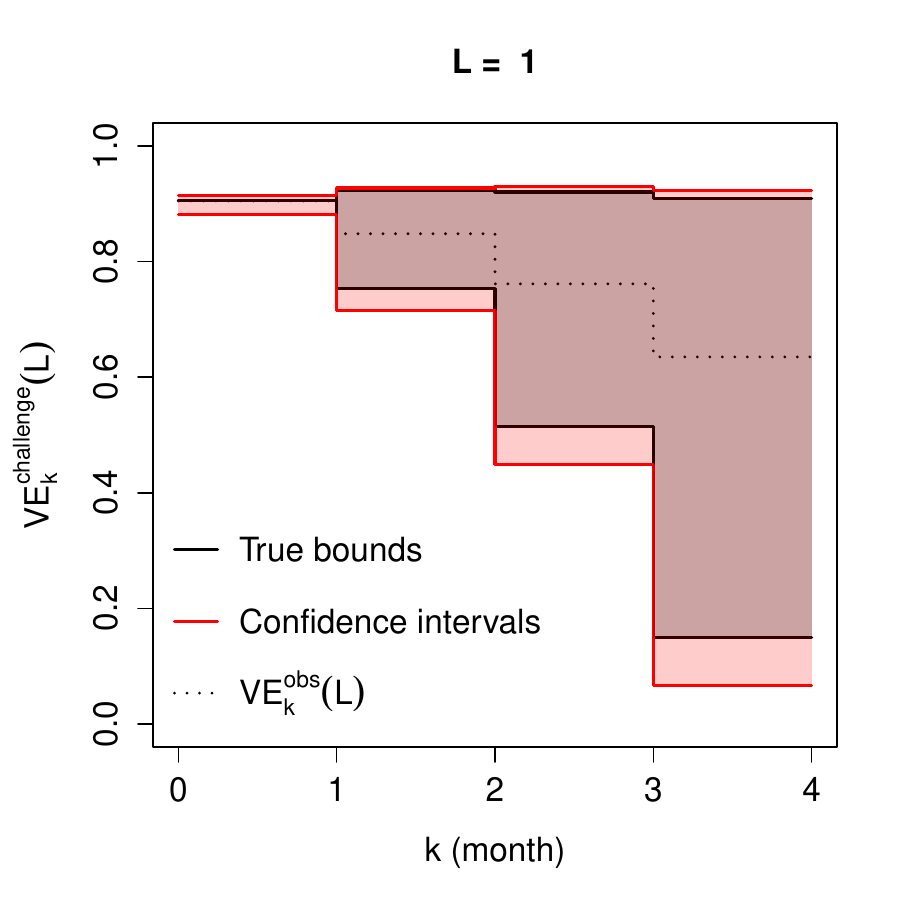}
}
\caption{The true bounds $\mathcal{L}_k,\mathcal{U}_k$ for baseline covariate levels $L=0$  and $L=1$   are shown together with one-sided 95\% confidence intervals of $\widehat{\mathcal{L}}_k,\widehat{\mathcal{U}}_k$, computed from non-parametric bootstrap with 500 resamples from a population of 10,000 individuals. (A) and (B) show a choice of parameters $(\beta_C,\beta_k)$ that gives narrow bounds, where $\text{VE}_2^\mathrm{challenge}$ is close to $\text{VE}_2^\mathrm{obs}$,  while (C) and (D) show another choice of parameters that gives wider bounds.}
    \label{fig:simulation}
\end{figure}

\section{Further analyses\label{app: further analyses}}

\subsection{Sensitivity analysis of waning estimates}
For $k\in\{1,2\}$, let $Q_k$ denote the exact time of the controlled exposure to a fixed quantity of the infectious agent during interval $k$, supposing that individuals are isolated from infectious exposures before and after this time. Suppose that Assumptions~\ref{ass:consistency 2 intervals}~(i) and (ii) for $e_1=0$, Assumptions~\ref{ass: exposure necessity 2 intervals}-\ref{ass: exclusion 2 intervals} and Assumptions~\ref{ass: exposure positivity 2 intervals}-\ref{ass: exposure exchageability multiple versions} hold.  Let $q^\dagger$ denote the time that vaccine recipients are most protected against a controlled infectious exposure. We assume that this is after an immune response has developed a few days following the second dose of treatment, and before the immune response (potentially) begins to wane. Then, by definition, 
\begin{align*}
    E[\Delta Y_1^{a=1,q^\dagger}\mid L] &\leq E[\Delta Y_1^{a=1,q}\mid L] \text{ w.p.~1}~,
\end{align*}
for all $q\in\mathcal{Q}$, which implies that
\begin{align}
    E[\Delta Y_1^{a=1,q^\dagger}\mid L]\leq E_G[E_G[\Delta Y_1^{a=1,Q_1}\mid L]\mid  E_1=1,A=1, L]~. \label{eq: conservative bound q dagger}
\end{align}
Furthermore, we will assume that Assumption~\ref{ass: no waning placebo multiple versions} holds for any duration of time intervals, such that the timing of the controlled infectious exposure does not affect the risk of infectious outcomes in placebo recipients \textit{even within} interval $k=1$. This implies that 
\begin{align}
     E[\Delta Y_1^{a=0,q^\dagger}\mid L] = E_G[E_G[\Delta Y_1^{a=0,Q_1}\mid L]\mid  E_1=1,A=0, L]~. \label{eq: no waning placebo within k=1}
\end{align}
Combining \eqref{eq: conservative bound q dagger} and \eqref{eq: no waning placebo within k=1} with Proposition~\ref{prp: randomized exposure version} by invoking Assumptions~\ref{ass: exposure necessity 2 intervals}-\ref{ass: exclusion 2 intervals} and Assumption~\ref{ass: treatment exchageability 2 intevals} gives
\begin{align*}
    \frac{E[\Delta Y_1^{a=1,q^\dagger}\mid L]}{E[\Delta Y_1^{a=0,q^\dagger}\mid L]} \leq \frac{E[\Delta Y_1\mid A=1,L]}{E[\Delta Y_1\mid A=0,L]}~.
\end{align*}
Thus, $\widehat{\text{VE}}_1^\mathrm{obs}$ is a conservative estimate of $\text{VE}_1^\mathrm{challenge}(l,q^\dagger)$. Furthermore, under Assumption~\ref{ass: no waning placebo multiple versions} for any duration of time intervals, 
\begin{align*}
    E[\Delta Y_1^{a=0,q^\dagger}\mid L]=E_G[E_G[\Delta Y_2^{a=0,q_1=0,Q_2^{q_1=0}}\mid L]\mid E_2^{q_1=0}=1,A=0,L]~.
\end{align*}
Under the additional Assumptions~\ref{ass: exposure necessity 2 intervals}-\ref{ass: exclusion 2 intervals} and Assumption~\ref{ass: treatment exchageability 2 intevals},  Proposition~\ref{prp: randomized exposure version} gives
\begin{align*}
    \frac{E[\Delta Y_1^{a=1,q^\dagger}\mid L]}{E_G[E_G[\Delta Y_2^{a=1,q_1=0,Q_2^{q_1=0}}\mid L]\mid  E_2^{q_1=0}=1,A=1, L]} \leq \mathcal{U}_\psi~.
\end{align*}
The left hand side is a contrast of a controlled infectious exposure at time $q^\dagger$ in interval 1 versus a randomly chosen time, $Q_2^{a,q_1=0}$, during interval $k=2$. Therefore, $\widehat{\mathcal{U}}_\psi$ conservatively estimates the extent of vaccine waning from interval 1 to 2.

In Tables~\ref{tab:Pfizer CI}-\ref{tab:Pfizer CI 2}, we perform a sensitivity analysis of the waning estimates reported in Table~\ref{tab:Pfizer CI 3}, using different choices of intervals $k=1$ and $k=2$.

\begin{table}[htbp]
  \centering
  \caption{Estimates and 95\% confidence intervals for \eqref{eq: observed VE}, \eqref{eq:L_VE_challenge_2}-\eqref{eq:U_VE_challenge_2} and \eqref{eq: LB minimal ass 2 intervals}-\eqref{eq: UB minimal ass 2 intervals}. Interval 1 ranged from 11 days after dose 1 until 2 months after dose 2 and interval 2 ranged from 2 months after dose 2 until the end of follow-up (chosen as day 190, see Table~\ref{tab:intervals}). Confidence intervals for the bounds $\mathcal{L}_\bullet,\mathcal{U}_\bullet$ were one-sided, whereas two-sided confidence intervals were used for VE estimates.}
    % \resizebox{\linewidth}{!}{
    \begin{tabular}{lr}
         \hline Estimator & \multicolumn{1}{l}{Estimate (95\% CI)}  \\\hline
     $\widehat{\text{VE}}_1^{\text{obs}},\widehat{\text{VE}}_1^{\text{challenge}}$ & $0.95(0.93,0.97)$ \\
    $\widehat{\text{VE}}_2^{\text{obs}}$ & $0.88(0.84,0.90)$ \\
    $\widehat{\mathcal{L}}_2$ & $0.86 (0.83, -)$ \\
    $\widehat{\mathcal{U}}_2$ & $0.91(-,0.93)$ \\
    $\widehat{\mathcal{L}}_\psi$ & $0.33 (0.23, - )$\\
    $\widehat{\mathcal{U}}_\psi$ & $0.54(-,0.84)$
    \end{tabular}%
    % }
  \label{tab:Pfizer CI}%
\end{table}%

\begin{table}[htbp]
  \centering
  \caption{Illustrates two additional choices of interval $k=2$, denoted interval I and interval II.   Interval I ranged from 7 days after dose 2 until 2 months after dose 2 and interval II ranged from 2 months after dose 2 until 4 months after dose 2. Interval $k=1$ ranged from dose 1 until the beginning of interval I (or II). Bootstrap confidence intervals (95\%) for the bounds $\mathcal{L}_{\mathrm{I,II}},\mathcal{U}_{\mathrm{I,II}}$ are one-sided, whereas two-sided confidence intervals have been used for $\text{VE}^\mathrm{obs}_{\mathrm{I,II}}$. The naive contrast of  $\text{VE}^\mathrm{obs}_{\mathrm{I}}$~vs.~$\text{VE}^\mathrm{obs}_{\mathrm{II}}$ suggests that vaccine protection waned, but the bounds $[\mathcal{L}^\mathrm{obs}_{\mathrm{I}},\mathcal{U}^\mathrm{obs}_{\mathrm{I}}]$ and $[\mathcal{L}^\mathrm{obs}_{\mathrm{II}},\mathcal{U}^\mathrm{obs}_{\mathrm{II}}]$ overlap and are wide due to the substantial number of depleted individuals from dose 1 until the beginning of intervals I and II.  }
    % \resizebox{\linewidth}{!}{
    \begin{tabular}{lr}
         \hline Estimator & \multicolumn{1}{l}{Estimate (95\% CI)}  \\\hline
     $\widehat{\text{VE}}_{\mathrm{I}}^{\text{obs}}$ & $0.96(0.93,0.98)$ \\
     $\widehat{\mathcal{L}}_\mathrm{I}$ & $0.83 (0.79, -)$ \\
    $\widehat{\mathcal{U}}_\mathrm{I}$ & $0.97(-,0.98)$ \\
    $\widehat{\text{VE}}_\mathrm{II}^{\text{obs}}$ & $0.90(0.87,0.93)$ \\
    $\widehat{\mathcal{L}}_\mathrm{II}$ & $0.81 (0.77, -)$ \\
    $\widehat{\mathcal{U}}_\mathrm{II}$ & $0.94(-,0.96)$ 
    \end{tabular}%
    % }
  \label{tab:Pfizer CI 2}%
\end{table}%

In the bounds \eqref{eq: bounds K intervals rare event}, we identified $E[Y_k^{\overline{c}=0}\mid A=a,L]$ by $\sum_{k^\prime=1}^k \Lambda_{k^\prime,a,L}$ under Assumption~\ref{ass: rare events K intervals}. The leading order correction is of order $(\sum_{k^\prime=1}^k\Lambda_{k^\prime,a,L})^2$, seen by Taylor expanding the approximation $E[\Delta Y_k^{\overline{c}=0}\mid A=a,L] \approx 1- \exp(\sum_{k^\prime=1}^k \Lambda_{k^\prime,a,L})$, which is small since the cumulative hazard point estimates during intervals 1 and 2 (for the choice of intervals in Table~\ref{tab:Pfizer CI 3})  were given by $\widehat\Lambda_{k=1,a=0}=0.020$, $\widehat\Lambda_{k=1,a=1}=0.001$, $\widehat\Lambda_{k=2,a=0}=0.029$ and $\widehat\Lambda_{k=2,a=1}=0.003$.

\subsection{Comparison against reported confidence intervals\label{sec:validation}}
By computing confidence intervals \eqref{eq: CI_VE_kj} of $\text{VE}_{k,j}^\mathrm{obs}$ (Table~\ref{tab:validation CI}), we illustrate that our approach gives identical point estimates and nearly identical confidence intervals to \citet{thomas_safety_2021}.

\begin{table}[htbp]
  \centering
  \caption{Validation of point estimates and confidence intervals computed from \eqref{eq: CI_VE_kj}, against corresponding numbers reported in Figure~2 of \citet{thomas_safety_2021} }
    \resizebox{\columnwidth}{!}{
    \begin{tabular}{p{ 10 cm }p{ 4cm  }p{ 4cm  }}
          \hline & Confidence interval (computed from \eqref{eq: CI_VE_kj}) & Confidence interval (\citet{thomas_safety_2021})   \\
    \hline Overall & $0.878( 0.853,0.898)$ & $0.878 (0.853 , 0.899)$ \\
    After dose 1 up to dose 2 & $0.584(0.414,0.705)$ & $0.584 (0.408 , 0.712)$ \\
    $<11$ days after dose 1 & $0.182(-0.236,0.459)$ & $0.182 (-0.261 , 0.473)$ \\
    $\geq 11$ days after dose 1 until dose 2 & $0.917(0.794,0.967)$ & $0.917 (0.796 , 0.974)$ \\
    After dose 2 until $<7$ days after & $0.915(0.723,0.974)$ & $0.915 (0.729 , 0.983)$ \\
    $\geq 7$ days after dose 2 & $0.912(0.889,0.930)$ & $0.912 (0.889 , 0.930)$ \\
    $\geq 7$ days after dose 2 until $<2$ months after & $0.962(0.932,0.979)$ & $0.962 (0.933 , 0.981)$\\
    $\geq 2$ months after dose 2 until $<4$ months after dose 2 & $0.901 (0.867,0.927)$ & $0.901 (0.866 , 0.929)$ \\
    $\geq 4$ months after dose 2 & $0.837 ( 0.748,0.895)$ & $0.837 (0.747 , 0.899)$ \\
    \end{tabular}%
     }
  \label{tab:validation CI}%
\end{table}%

\subsection{Subgroup analysis}
Let $\widehat\lambda_{a,l}=N_{a,l}/T_{a,l}$, where $N_{a,l}$ and $T_{a,l}$ are the overall number of recorded events and person time at risk for treatment group $A=a$ and baseline covariates $L=l$. The confidence interval
\begin{align}
    CI(\widehat\lambda_{a,l})=\widehat\lambda_{a,l} \exp\left(\pm \frac{z_{1-\alpha/2}}{\sqrt{N_{a,l}}} \right) \label{CI_lambda_al}
\end{align}
is asymptotically valid in large samples under Assumption~\ref{ass: rare events K intervals} by \eqref{eq: avar lambda} and Theorem 8.22 in \citet{lehmann_theory_1998} with transformation $h(\lambda_{a,l})=\log \lambda_{a,l}$. Assuming a constant hazard $\lambda_{a,l}$ during time $[0,\tau]$, the survival by day $\tau$ is given by $\exp(-\lambda_{a,l}\tau)$ and conversely the cumulative incidence by $\mu_{a,l}=1-\exp(-\lambda_{a,l}\tau)$. Then, using a log-minus-log transformation of the survival  \citep{aalen_survival_2008} together with \eqref{CI_lambda_al} gives the confidence interval
\begin{align}
    CI(\widehat\mu_{a,l}) = 1-\exp \left\{-\widehat\lambda_{a,l} \tau\exp\left(\pm \frac{z_{1-\alpha/2}}{\sqrt{N_{a,l}}} \right)\right\} ~. \label{eq: CI cum inc}
\end{align}
Confidence intervals for conditional hazards and cumulative incidences overlap for the different subgroups $L$ (Table~\ref{tab:subgroup analyses}).

\begin{table}[ht!]
  \centering
  \caption{Hazard point estimates and 95\% confidence intervals \eqref{CI_lambda_al}-\eqref{eq: CI cum inc} by subgroup $A=a,L=l$, computed using Table S7 in the Supplementary Appendix of \citet{thomas_safety_2021}. Hazards are given in units of $10^{-5}\textrm{days}^{-1}$. Cumulative incidences are evaluated on day $\tau=190$.}
    \resizebox{\linewidth}{!}{
    \begin{tabular}{lllll}
    \toprule
          & $CI(\widehat\lambda_{a=1,l})$ &     $CI(\widehat\lambda_{a=0,l})$ & $CI(\widehat{\mu}_{a=1,l})$ & $CI(\widehat{\mu}_{a=0,l})$ \\
    \midrule
    Overall & $3.38  (2.70  , 4.22)$  &      $38.79 ( 36.27 , 41.49 )$ & $0.006(0.005,0.008)$ & $0.071(0.067,0.076)$\\
    At risk: Yes & $3.43  ( 2.46  , 4.77 )$ &        $40.98 ( 37.16 , 45.19 )$ & $0.006(0.005,0.009)$ & $0.075(0.068,0.082)$ \\
    At risk: No & $3.34  ( 2.46  , 4.51  )$ &        $37.03 ( 33.76  , 40.62 )$ & $0.006(0.005,0.009)$ & $0.068(0.062,0.074)$ \\
    Age 16-64 and at risk & $3.81  ( 2.65  , 5.49 )$  &      $ 44.68 ( 40.07 , 49.81 )$ & $0.007(0.005,0.010)$ & $0.081(0.073,0.090)$ \\
    Age 65 or older and at risk & $2.42  ( 1.09  , 5.38  ) $      & $29.65 ( 23.50 , 37.42 )$ & $0.005(0.002,0.010)$ & $0.055(0.044,0.069)$ \\
    Obese: Yes & $3.52  ( 2.41  , 5.13  )$       & $41.96 ( 37.57 , 46.87 )$ & $0.007(0.005,0.010)$ & $0.077(0.069,0.085)$ \\
    Obese: No & $3.31  ( 2.51  , 4.36  )$       & $37.16 ( 34.14 , 40.44 )$ & $0.006(0.005,0.008)$ & $0.068(0.063,0.074)$ \\
    Age 16-64 and obese & $3.91  ( 2.62  , 5.84  ) $      & $44.87 ( 39.79 , 50.60 )$ & $0.007(0.005,0.011)$ & $0.082(0.073,0.092)$ \\
    Age 65 or older and obese &$ 2.03  ( 0.66 , 6.31 )$         & $30.07 ( 22.45 , 40.27 )$  & $0.004(0.001,0.012)$ & $0.056(0.042,0.074)$ \\
    \end{tabular}%
    }
  \label{tab:subgroup analyses}%
\end{table}%

\section{Challenge estimands that use antibody titers\label{appsec:antibodies}}
In this section, we propose how the challenge effect can be used to quantify the relation between immunological markers and vaccine waning. Let
\begin{align*}
    \psi(l_0,l_1)=\frac{E[\Delta Y_1^{a=1,e_1=1}\mid L_1^{a=1}=l_1,L_0=l_0]}{E[\Delta Y_2^{a=1,e_1=0,e_2=1}\mid L_1^{a=1}=l_1,L_0=l_0]} ~,
\end{align*}
where $L_1^{a=1}$ is an immunological marker, e.g.\ antibody titer, measured at the beginning of interval 1 under an intervention that assigns vaccine $a$ and $L_0$ is a vector of baseline covariates. The estimand corresponds to the answer of the following plain-English question: to what extent does the amount of waning depend on the initial antibody response?  If $\psi(l_0,l_1)$ is small for $l_1=l_-$, but large for  $l_1=l_+$, then we might justify that individuals with $l_1=l_-$  are prioritized to receive booster vaccination before individuals with $l_1=l_+$. 

Let $L_2$ denote a measurement of antibody titer at the beginning of time interval 2. Then,
\begin{align*}
    \phi(l_0,l_1,l_2)=\frac{E[\Delta Y_2^{a=1,e_1=0,e_2=1}\mid L_2^{a=1,e_1=0}=l_2,L_1^{a=1}=l_1,L_0=l_0]}{E[\Delta Y_2^{a=1,e_1=0,e_2=1}\mid L_0=l_0]}
\end{align*}
is the relative risk of outcomes under the  antibody profile $(l_1,l_2)$. Unlike the naive contrast
\begin{align*}
    \phi^\mathrm{obs}(l_0,l_1,l_2)=\frac{E[\Delta Y_2\mid \Delta Y_1=0,L_2=l_2,L_1=l_1,L_0=l_0,A=1]}{E[\Delta Y_2\mid  \Delta Y_1=0,L_0=l_0,A=1]} ~,
\end{align*}
the conditional challenge effect $ \phi(l_0,l_1,l_2)$ is not subject to depletion of susceptible individuals during interval 1.  If $ \phi(l_0,l_1,l_2) >1$ then individuals with $(L_1=l_1,L_2=l_2)$ are less protected during interval 2, and should be prioritized for a booster dose among those who initially received the vaccine. This estimand describes heterogeneity in vaccine protection across antibody responses.  

Identification results for $\psi(l_0,l_1)$ and $ \phi(l_0,l_1,l_2)$ can be derived similarly to Theorem~\ref{thm: waning minimal ass 2 intevals} and Proposition~\ref{prp: homogeneity}, but require some additional assumptions, which are also  single world; they are in-principle testable in future challenge experiments \citep{thomas_s_richardson_single_2013}. These estimands will be studied thoroughly in future research.

Finally, antibody measurements can also be used for sensitivity analyses. For example, one could compare the distribution of antibodies in event-free individuals by time $k$ ($\Delta Y_k=0$), for different values of $k$ as a measure of the depletion of susceptible individuals.

\end{document}